%% file: arxiv.tex
\documentclass[12pt]{article}
\input{preamble}
\usepackage{docmute}
\usepackage{chapterbib}
\usepackage{needspace}

\renewcommand{\spacingset}[1]{%
  \renewcommand{\baselinestretch}{1}\small\normalsize}
\AtBeginEnvironment{thebibliography}{%
  \normalfont\normalsize\setstretch{1}%
  \setlength{\bibsep}{\parskip}%
  \interlinepenalty=10000}
\date{October 28, 2026}

\begin{document}

\include{main}

\clearpage
\setcounter{equation}{0}
\setcounter{figure}{0}
\setcounter{table}{0}
\setcounter{algocf}{0}
\setcounter{definition}{0}
\setcounter{theorem}{0}
\setcounter{lemma}{0}
\setcounter{trick}{0}
\setcounter{property}{0}
\setcounter{example}{0}
\setcounter{corollary}{0}
\setcounter{remark}{0}
\setcounter{assumption}{0}
\setcounter{proposition}{0}

\renewcommand{\theHequation}{supp.\arabic{equation}}
\renewcommand{\theHfigure}{supp.\arabic{figure}}
\renewcommand{\theHtable}{supp.\arabic{table}}
\renewcommand{\theHalgocf}{supp.\arabic{algocf}}
\renewcommand{\theHdefinition}{supp.\arabic{definition}}
\renewcommand{\theHtheorem}{supp.\arabic{theorem}}
\renewcommand{\theHlemma}{supp.\arabic{lemma}}
\renewcommand{\theHtrick}{supp.\arabic{trick}}
\renewcommand{\theHproperty}{supp.\arabic{property}}
\renewcommand{\theHexample}{supp.\arabic{example}}
\renewcommand{\theHcorollary}{supp.\arabic{corollary}}
\renewcommand{\theHremark}{supp.\arabic{remark}}
\renewcommand{\theHassumption}{supp.\arabic{assumption}}
\renewcommand{\theHproposition}{supp.\arabic{proposition}}
\include{supplement}
\end{document}

%% file: preamble.tex
\usepackage{amsmath}
\usepackage{graphicx}
\usepackage{enumerate}
\usepackage{natbib}
\usepackage{url} 

\newcommand{\blind}{1}

\usepackage[english]{babel}
\usepackage[utf8x]{inputenc}
\usepackage[T1]{fontenc}
\usepackage{setspace}

\usepackage[font=footnotesize,labelfont=bf]{caption}
\usepackage[table,dvipsnames]{xcolor}
\usepackage{graphicx}
\usepackage{subcaption}
\usepackage{longtable} 
\usepackage{multirow}
\usepackage{listings}
\usepackage{makecell}
\usepackage{array}
\usepackage{float}
\usepackage{placeins}
\usepackage{dsfont}
\usepackage{rotating}
\usepackage{booktabs}
\usepackage{enumitem}
\usepackage{tikz}
\usepackage{pgf}
\usetikzlibrary{decorations.pathreplacing}
\usepackage{amsmath}
\usepackage{amssymb}
\usepackage{amsthm}
\usepackage{bm}
\usepackage[linesnumbered,ruled,vlined,resetcount]{algorithm2e}
\usepackage{mathtools}
\usepackage{mathrsfs}
\usepackage{pifont}
\usepackage{stmaryrd}

\usepackage{natbib}
\usepackage{xr}
\makeatletter
\providecommand{\new@label@record}[2]{}
\makeatother
\usepackage[
  colorlinks,
  citecolor=blue,
  linkcolor=red,
  anchorcolor=red,
  urlcolor=blue
]{hyperref}
\mathtoolsset{showonlyrefs}
\usepackage{authblk}
\usepackage{todonotes}
\usetikzlibrary{arrows,arrows.meta}
\usepackage{comment}

\input{macro}

\newif\ifshowrevision
\showrevisionfalse 
\newcommand{\red}[1]{\ifshowrevision\textcolor{red!60}{#1}\else#1\fi}

\newcommand{\green}[1]{\textcolor{ForestGreen}{#1}}

\renewcommand{\hat}{\widehat}
\newcommand{\Dref}{\cD_{\text{ref}}}
\newcommand{\Dtest}{\cD_{\text{test}}}

\newcommand{\cmark}{\green{\ding{51}}}%
\newcommand{\xmark}{\textcolor{red}{\ding{55}}}%

\newcommand{\be}{{\bm e}}
\newcommand{\bp}{{\bm p}}
\newcommand{\bZ}{{\bm Z}}
\newcommand{\bh}{\textnormal{BH}}
\newcommand{\ebh}{\textnormal{e-BH}}

\newcommand{\boost}{\mathfrak{b}} 
\newcommand{\dif}{\mathrm{d}}

\makeatletter
\renewcommand\paragraph{\@startsection{paragraph}{4}{\z@}%
  {0.75ex \@plus 0.25ex \@minus 0.1ex}%
  {-0.75em}%
  {\normalfont\normalsize\bfseries}}
\AtBeginEnvironment{thebibliography}{\setstretch{1}\small}
\makeatother

\long\def\comment#1{}

\def\spacingset#1{\renewcommand{\baselinestretch}%
{#1}\small\normalsize}

%% file: macro.tex
\theoremstyle{plain}

\newtheorem{theorem}{Theorem}
\newtheorem{lemma}{Lemma}

\newtheorem{corollary}{Corollary}
\newtheorem{remark}{Remark}

\newtheorem{proposition}{Proposition}

\newcommand{\RR}{\mathbb{R}}
\newcommand{\EE}{\mathbb{E}}
\newcommand{\PP}{\mathbb{P}}
\newcommand{\ind}{\mathds{1}}

\newcommand{\Bin}{\mathsf{Bin}}
\newcommand{\fdr}{\textnormal{FDR}}
\newcommand{\fdp}{\textnormal{FDP}}
\newcommand{\indc}[1]{\mathbf{1}\{#1\}}

\newcommand{\argmin}[1]{\underset{#1}{\arg\!\min}}
\newcommand{\Sym}{\operatorname{Sym}}

\usepackage{xspace}

\newcommand{\stepa}[1]{\overset{\rm (a)}{#1}}
\newcommand{\stepb}[1]{\overset{\rm (b)}{#1}}

\providecommand{\given}{}
\renewcommand{\given}{{\,|\,}}

\newcommand{\Biggiven}{\,\Big{|}\,}
\newcommand{\bigggiven}{\,\bigg{|}\,}

\def\@#1\@{\begin{align}#1\end{align}}
\def\$#1\${\begin{align*}#1\end{align*}}

\definecolor{myblue}{rgb}{.8, .8, 1}
\definecolor{mathblue}{rgb}{0.2472, 0.24, 0.6} 
\definecolor{mathred}{rgb}{0.6, 0.24, 0.442893}
\definecolor{mathyellow}{rgb}{0.6, 0.547014, 0.24}

\newcommand{\cA}{{\mathcal{A}}}

\newcommand{\cD}{{\mathcal{D}}}
\newcommand{\cE}{{\mathcal{E}}}
\newcommand{\cF}{{\mathcal{F}}}
\newcommand{\cG}{{\mathcal{G}}}
\newcommand{\cH}{{\mathcal{H}}}
\newcommand{\cI}{{\mathcal{I}}}

\newcommand{\cN}{{\mathcal{N}}}

\newcommand{\cR}{{\mathcal{R}}}
\newcommand{\cS}{{\mathcal{S}}}

\newcommand{\cV}{{\mathcal{V}}}
\newcommand{\cW}{{\mathcal{W}}}

\newcommand{\cZ}{{\mathcal{Z}}}

%% file: main.tex
\spacingset{1}

\noeqref{eq:weighted_exch_cond_dist}


\if1\blind
{
  \title{\bf Full conformal novelty detection}
  \author[1]{Junu Lee\thanks{Corresponding author: \href{mailto:junulee@wharton.upenn.edu}{\texttt{junulee@wharton.upenn.edu}}}}
  \author[2]{Ilia Popov}
  \author[1]{Zhimei Ren}
  \affil[1]{Department of Statistics and Data Science, University of Pennsylvania}
  \affil[2]{Department of Statistics, Stanford University}
  \maketitle
} \fi

\if0\blind
{
  \bigskip
  \bigskip
  \bigskip
  \begin{center}
    {\LARGE\bf Full conformal novelty detection}
\end{center}
  \medskip
} \fi

\bigskip
\begin{abstract}
This paper presents a powerful methodology for flexible full-data nonparametric novelty detection that offers distribution-free 
false discovery rate (FDR) control guarantees.
Building on the full conformal inference framework and the concept of 
e-values, 
we introduce {\em full conformal e-values} to quantify evidence for novelty relative to a given reference dataset. 
These e-values are then utilized by carefully crafted 
multiple testing procedures to 
identify a set of novel units out-of-sample with provable  finite-sample FDR control. 
We showcase several instantiations of e-values, including those which employ a data-driven model selection strategy to amplify power.
Furthermore, our framework is extended to address distribution shift, accommodating scenarios where novelty detection must be performed on data drawn from a shifted distribution relative to the  reference dataset. In all settings, 
our method can \red{achieve high power}---outperforming existing novelty detection methods---even with limited amounts of reference data; this is illustrated by empirical evaluations on synthetic data and an application to a \red{dataset containing regular and} malicious LLM prompts.
\end{abstract}

\noindent%
\vfill

\newpage
\spacingset{1.9} 
\section{Introduction}
\label{sec:intro}
\input{./sections/intro}

\section{Preliminaries}
\label{sec:background}
\input{./sections/background}
\section{Full conformal novelty detection}
\label{sec:fc-ebh}
\input{./sections/fc_ebh}

\vspace{-0.35\baselineskip}
\section{Weighted full conformal novelty detection}
\label{sec:shift}
\input{./sections/shift}
\vspace{-0.35\baselineskip}
\section{Numerical experiments} 
\label{sec:experiments}
\input{./sections/experiments.tex}
\vspace{-0.35\baselineskip}
\section{Discussion}
\label{sec:disc}
\input{./sections/discussion}

\subsubsection*{Acknowledgements}
\begin{samepage}
The authors would like to thank the Wharton Research Computing team for the amazing support provided by the staff members. 
J.L. is supported by a Graduate Research Fellowship from the NSF. 
Z.R. is supported by NSF grant DMS-2413135 and the Wharton AI and Analytics Initiative's AI Research Fund.

The authors used various
generative AI models provided by OpenAI, most recently 
GPT-5.6 Sol, to polish the technical presentation and 
visualizations in this paper and improve the codebase 
in terms of modularity and reproducibility.
\end{samepage}

\bibliographystyle{apalike}
\bibliography{ref}

%% file: sections/intro.tex

This paper considers the problem of novelty detection: given a pool of new observations, 
the goal is to identify \red{those whose}
distributions differ from that of a reference dataset~\citep{wilks1963multivariate,hawkins1980identification,riani2009finding,cerioli2010multivariate}.
This fundamental statistical task arises in a wide range of domains, such as astronomy~\citep{mary2020origin}, 
high-energy physics~\citep{vatanen2012semi},
proteomics~\citep{shuster2022situ,gao2023simultaneous},
and  fraudulent activity detection~\citep{ahmed2016survey}.
In recent years, with the rapid development of artificial intelligence (AI) and machine learning (ML) tools---for example, the advent and subsequent popularity of large language models (LLMs)---there have been new additions to the suite of novelty detection problem domains. We list a couple of relevant examples below.
\paragraph{Malicious prompt detection.} With the widespread use of LLMs, there is a growing concern about adversaries crafting malicious prompts to elicit harmful or undesirable responses from the model~\citep{chao2025jailbreaking}, leading to issues such as misinformation, privacy breaches, or model stealing~\citep{tramer2016stealing,achiam2023gpt,carlini2024stealing}. Detection of adversarial prompts can be framed as a novelty detection problem, where the reference dataset consists of benign prompts and the test set contains both benign and malicious prompts, \red{with the latter interpreted as}  outliers~\citep{oliynyk2023know}.

\paragraph{Detection of copyright violation.} Suppose a creator suspects 
that an AI model has been trained on her copyright-protected material 
(e.g., published books) without 
permission. To detect such infringements, dataset inference~\citep{maini2021dataset,maini2024llm} 
considers a setting where the creator can provide private/unpublished material. 
Then, one can treat (some function of) the private material as the reference dataset, and 
detect copyright violation by finding the outliers in the suspect material (as 
the model would ``respond'' differently to the data upon which it was trained 
versus fresh data).  

In these motivating examples, the goal is to identify the outliers in order to subsequently audit or address them. For example, flagged samples in the copyright violation setting may be collected and sent for legal review. As follow-up decisions may be costly, it is crucial to limit how often we \red{incorrectly identify outliers}. Hence, the false discovery rate (FDR)~\citep{benjamini1995controlling}, which is the expected rate of false discoveries, is a natural error criterion to control. At the same time, the high-dimensional and complex nature of modern datasets---including those from the motivating examples---can hinder the estimation of the inlier and outlier distributions, inflating the downstream FDR.  In this work, we address these challenges from a model-free perspective, imposing no parametric assumptions on the data-generating process. With access to a reference dataset consisting of i.i.d.~ inlier samples, we seek a novelty detection procedure that effectively flags outliers in the newly observed test dataset while rigorously controlling FDR at a pre-specified level.

To this end, a prior work by~\citet{bates2023testing} introduced an FDR-controlling method
based on the sample-splitting variant of conformal inference~\citep{vovk2005algorithmic}, known as split conformal (SC). 
Their approach constructs a {\em conformal p-value} for each test observation,  
and then applies the {\em Benjamini--Hochberg (BH)} procedure to the p-values at the 
target FDR level.
The conformal p-values are constructed by randomly splitting the reference 
data into two folds. One fold is used to train a 
(one-class) classifier, which scores the test units as well as the reference samples in the remaining fold. The p-values are then obtained by comparing the test scores to the reference scores.
The sample splitting step, central to the construction, is undesirable 
since it reduces the number of samples used for training and inference. Moreover, it introduces external randomness to the procedure, which leads 
to unstable discoveries and could potentially be exploited to ``hack'' the selections over different sample splits.

A follow-up work~\citep{marandon2024adaptive} improves the training step 
by using a masked version of the full dataset, with imputed inlier/outlier labels. The masking still induces inefficiency, and the method still depends on a random split. 
\citet{bashari2024derandomized} propose a derandomization scheme which aggregates results over many splits; this aggregation comes with a loss of statistical power as only highly replicable discoveries (over different splits) survive the aggregation. Motivated by these limitations, this 
work provides a methodology based on the full-data version of conformal inference, referred to as full conformal (FC), which is  
(1) guaranteed to achieve finite-sample FDR control; (2)
more powerful than existing strategies, especially when the reference dataset is limited;
and (3) deterministic in its most basic form.

\subsection{Preview of our contributions}
We present {\em $K$-FC ND ($K$-block full conformal novelty detection)}, 
a powerful full-data conformal algorithm  with immense data-driven flexibility that rigorously controls the finite-sample FDR.
Our proposal leverages {\em e-values} (see \citet{ramdas2024hypothesis} for
an overview) as the statistical 
tool for quantifying the evidence of an outlier. Specifically, we assign each test point a {\em full conformal e-value} and select units using e-BH, the e-value analogue of BH~\citep{wang2022false}, or its uniformly more powerful conditionally calibrated version~\citep{lee2024boosting}.
We outline our theoretical and methodological contributions below
and summarize them in the context of
existing methods in Table~\ref{tab:summary}.

\paragraph{Novelty detection with the full conformal paradigm.}
In contrast to existing methods relying on the SC 
paradigm, our method is based on the FC paradigm. 
In our framework, both the model-fitting and inference steps
make use of all the reference data (and even the test data); 
it also flexibly allows us to adaptively choose a model-fitting 
algorithm in a data-driven fashion. Such flexibility is crucial 
for the power of the novelty detection algorithm, 
especially when the number of reference 
samples is limited (like in the LLM-based problem settings, where annotated data is generally
expensive).

\paragraph{Data-driven model selection without sacrificing FDR control.}

    The conformal framework hinges on training an ML model as a scoring function to distinguish inliers from outliers; hence, the quality of the scoring function is heavily affected by the choice of model type (and hyperparameters) and in turn affects the power of the downstream procedure. Given a suite of different model types, we propose a data-driven model selection strategy which seamlessly integrates into $K$-FC ND and allows the adaptive selection of the best model for each test point. Unfortunately, na\"ively using p-values with our model selection strategy does not guarantee FDR control due to the complex dependence structure of the resulting p-values (see Appendix \ref{appd:model_selection_fdr_inflation} for an example of FDR violation). By contrast, e-BH continues to control the FDR---regardless of dependence structure---as long as the e-values are valid even after model selection, which we show to be the case.

\paragraph{Boosting the power of full conformal e-BH.}

Although we design the e-values to be individually powerful, we further boost the power of the e-BH procedure by instantiating and improving upon the conditionally calibrated e-BH (e-BH-CC) framework~\citep{lee2024boosting}. In the conformal setting,
our improvements are twofold. Computationally, we compute the boost exactly (without Monte Carlo samples) and implement additional shortcuts which reduce the computational burden \red{without sacrificing FDR control or empirical power}. Methodologically, we use a novel boosting mechanism which encourages more user flexibility 
and is of independent interest outside our setting.

\paragraph{Novelty detection under distribution shift.}
Another major technical contribution of this work is to 
generalize $K$-FC ND to a setting where the inliers of the test set 
experience a (known) distribution shift from the inliers of the reference set. To perform novelty detection in this setting, we propose a weighted variant of $K$-FC ND
which similarly achieves finite-sample FDR control and 
demonstrates superior empirical power.

We begin with preliminaries in Section \ref{sec:background}. Section \ref{sec:fc-ebh} introduces our novel methodologies, while Section \ref{sec:shift} adapts them to the distribution-shifted setting. Section \ref{sec:experiments} contains a multitude of numerical simulations and an application to the aforementioned malicious prompt detection problem (Section \ref{sec:realdata}) which demonstrate the efficacy of our methods.

\input{tables-and-figures/table_of_contributions.tex}

%% file: tables-and-figures/table_of_contributions.tex
\begin{table}[!t]
\footnotesize
\centering    
\setlength{\tabcolsep}{3pt}
\begin{tabular}{c|c|c|c|c|c}
\toprule 
 & \makecell{Finite-sample\\FDR control} & Full-data & \makecell{Data-driven\\model selection} & Non-random & \makecell{Distribution\\shift}\\ 
\midrule 
\citet{bates2023testing} & \cmark & \xmark & \xmark & \xmark & \xmark\\
\citet{marandon2024adaptive} & \cmark & \cmark${}^*$ & \cmark & \xmark & \xmark\\ 
\citet{bashari2024derandomized} & \cmark & \cmark${}^*$ & \xmark & \cmark${}^\dagger$ & \xmark\\ 
\midrule
\textbf{This work} & \cmark & \cmark & \cmark & \cmark${}^\ddagger$ & \cmark \\
\bottomrule
\end{tabular}
\caption{A summarized comparison between our work and existing strategies. 
For the ``Full-data'' property, the asterisk (*) specifies that the procedure 
still involves splitting the reference data,
losing information in the process.
For the ``Non-random'' property, the obelisk ($\dagger$) specifies that the procedure 
\red{becomes approximately stable when the number of replications is large, with exact derandomization requiring exhaustive averaging}, while the diesis ($\ddagger$) specifies that the procedure is deterministic in its base form but can be made random if desired.
}
\label{tab:summary}
\end{table}

%% file: sections/background.tex

\subsection{Problem setup}
The novelty detection problem setting is formally delineated as follows. 
Let $Z \in \cZ$ denote an observation, or unit.
We are given a reference dataset $\Dref = \{Z_1, \dots, Z_n\}$ and a test dataset $\Dtest = \{Z_{n+1}, \dots, Z_{n+m}\}$. 
The reference dataset $\Dref$ consists of i.i.d.~units drawn from some unspecified distribution $P$ \red{and is independent of} $\Dtest$\red{, which} contains independent units, each with an unknown distribution \red{that may or may not be $P$}. Novelty detection 
\red{aims} to identify which units in $\Dtest$ \textit{do not follow} the distribution $P$, i.e., which $Z_{n+j}$ are outliers in the context of the reference dataset.
We write this as a multiple hypothesis testing problem, in which there are $m$ null hypotheses $H_1, \dots, H_m$ such that $H_j \colon Z_{n+j} \sim P$;
i.e., the $j$th test unit $Z_{n+j}$ is an \textit{inlier}. When we reject $H_j$, \red{it is because we have evidence}  that $Z_{n+j}\not\sim P$ and is thus an  \textit{outlier}. 
In what follows, we let $\cH_0 := \{j \in [m]: H_j \text{ is true}\}$ denote the set of inliers 
and $\cH_1 := [m]\backslash \cH_0$ the set of outliers.
We also use $\pi_0 = |\cH_0|/m$ and  $\pi_1 = 1 - \pi_0$ to denote the fraction of nulls and nonnulls, respectively.

A novelty detection algorithm takes $(\Dref, \Dtest)$ as input and returns a subset $\cR \subseteq [m]$ which indexes the test units identified as outliers. 
The FDR and power of $\cR$ are given by
\@\label{eq:fdr}
& \fdr \coloneqq \EE[\fdp], \textnormal{ where } \fdp \coloneqq  \frac{\sum_{j \in \cH_0}
\ind\{j \in \cR\}}{|\cR| \vee 1}; \\
& \text{Power} \coloneqq  \EE\bigg[\frac{\sum_{j \in \cH_1}\ind\{j\in \cR\}}{|\cH_1|\vee 1}\bigg],
\@
where $a \vee b = \max(a,b)$ and $|A|$ denotes the cardinality of a set $A$.
Here, FDP (false discovery proportion) is the fraction of the selected units that are true inliers. 
The goal is to develop a novelty detection algorithm that controls the FDR at a pre-specified level 
$\alpha \in (0,1)$ while maintaining high power, ideally without introducing external randomness.

\subsection{Conformal novelty detection}
\label{subsec:cnd_lit}
For this novelty detection problem,~\cite{bates2023testing} proposed a method
based on {\em split conformal inference}, a sample-splitting variant of conformal inference. 
In their proposal, 
$\Dref$ is partitioned into a training fold $\cI_1$ and a calibration fold $\cI_2$:
the training fold $\cI_1$ is used to obtain a score function $V:\cZ \to \RR$ that quantifies
the likelihood of a unit being an outlier, with a larger value indicating a higher likelihood;
the score function $V(\cdot)$ is then used to score the units in the calibration set and the test set, 
yielding the calibration scores $\{V_i := V(Z_i)\}_{i\in \cI_2}$ and the test scores $\{V_{n+j}\}_{j\in[m]}$.
For the $j$-th test point, the split conformal (SC) p-value
is constructed by contrasting its score with the calibration scores:
\@ \label{eq:conf_pval}
p_j := \frac{1 + \sum_{i\in \cI_2} \ind\{V_i \ge V_{n+j}\}}{|\cI_2|+1}.
\@
Under the null hypothesis $H_j$ (i.e., $Z_{n+j}$ is an inlier), $V_{n+j}$ is exchangeable with 
$\{V_i\}_{i\in \cI_2}$. Hence, \red{under a no-ties assumption}, the rank of $V_{n+j}$ among $\{V_i\}_{i\in\cI_2}\cup\{V_{n+j}\}$ is uniform 
on $[|\cI_2|+1]$, leading to the super-uniformity of $p_j$. 
Since an outlier tends to have a large score, $p_j$ is likely to be small when $H_j$ is false.~\citet{bates2023testing} prove that  
these SC p-values are {\em positively regression dependent on a subset (PRDS)}, 
and by the result of~\citet{benjamini2001control}, applying the BH procedure to $\{p_j\}_{j\in[m]}$ performs novelty detection while
ensuring FDR control.\footnote{The procedure has appeared in \citet{weinstein2017power,mary2022semi,gao2023simultaneous} 
in different but equivalent forms.}  
Despite the rigorous distribution-free FDR control guarantee, 
the method---referred to as {\em SC ND} hereafter---relies on sample splitting, 
using only a subset of the data for training and the rest for calibration. 
The reduced training sample size may limit the model's ability to capture 
the underlying data distribution, potentially impacting its performance on unseen test data.
At the same time, the smaller calibration sample size constrains the granularity of the p-value, 
which also impacts the power of BH. Moreover, the sample split makes SC ND intrinsically random, leading to unstable novelty detection.

A follow-up work by~\citet{marandon2024adaptive} improves the power 
of SC ND by making better use of all the given data. Their proposed method, 
{\em AdaDetect}, splits the reference data into two folds, 
$\cI_1$  and $\cI_2$, as in SC ND. However, it then 
uses both folds for model fitting: a classifier $\hat \mu: \cZ \to [0,1]$ 
is trained on 
$\{Z_i\}_{i\in\cI_1}$, as well as $\{Z_i\}_{i\in \cI_2} \cup \cD_{\text{test}}$, 
with the former treated as inliers and the latter outliers (or a mixture of both). 
The classifier $\hat \mu$ is then used for scoring the calibration and 
test points by giving an estimated likelihood of being an outlier, e.g., 
$V_i = \hat \mu(Z_i)$. The authors show that 
as long as the training process is invariant to the permutation of 
$\{Z_i\}_{i \in \cI_2} \cup \cD_{\text{test}}$, the conformal 
p-values~\eqref{eq:conf_pval} are valid and satisfy the 
PRDS property. Compared with SC ND, AdaDetect uses partial information 
in the calibration and test sets for model fitting, thereby 
improving the statistical power in some settings, but the issue of randomness remains 
due to the sample splitting step. 

\subsubsection{Using e-values for conformal novelty detection}

Subsequently,~\citet{bashari2024derandomized} propose 
a remedy to address the randomness issue by taking an e-value~\citep{ramdas2024hypothesis}
perspective on the SC p-values. A valid e-value for a null hypothesis $H_0$ is a nonnegative random variable $E$ such that $\EE_{H_0}[E] \le 1$; that is, under $H_0$,  $E$ is on average at most 1. For any $\alpha \in (0,1)$, $E$ directly induces a level-$\alpha$ test by rejecting $H_0$ when observing $E \ge \alpha^{-1}$ (via Markov's inequality). A well-designed e-value generally has expectation much higher than 1 under the alternative.

The e-value is also relevant to multiple hypothesis testing problems, where we associate an e-value $e_j$ to each null hypothesis $H_j$. To control the FDR, we can run the  e-BH procedure~\citep{wang2022false}, which provably controls the FDR at a pre-specified level $\alpha \in (0,1)$ under any joint dependence structure (over \textit{both} null and non-null e-values).
\indent The derandomization scheme
establishes the equivalence between SC ND  (or AdaDetect) and e-BH
applied to the {\em SC e-values},\footnote{\citet{bashari2024derandomized} in fact construct similar but ``compound'' e-values~\citep{ignatiadis2024compound}. 
The above construction, shown to be equivalent to \citet{bates2023testing}, comes from \citet{lee2024boosting}.} defined as follows: for each $j\in[m]$,
\enlargethispage{0.9\baselineskip}
\begingroup
\setlength{\abovedisplayskip}{5pt}
\setlength{\belowdisplayskip}{5pt}
\setlength{\abovedisplayshortskip}{4pt}
\setlength{\belowdisplayshortskip}{4pt}
\@ \label{eq:sc_eval} 
& e_j = (|\cI_2|+1) \cdot \frac{\ind\{V_{n+j} \ge T\}}{1+\sum_{i\in\cI_2}\ind\{V_i \ge T\}}, \\
\text{ where } &T = \inf\bigg\{ t\in \{V_i\}_{i \in \cI_2} \cup \{V_{n+\ell}\}_{\ell\in[m]} \colon \frac{m}{ |\cI_2|+1} 
\cdot \frac{1 + \sum_{i\in \cI_2} \indc{V_i \ge t} }{1\vee \sum_{j=1}^m \indc{V_{n+j} \ge t}   } \le \alpha  \bigg\}.
\@
\endgroup
\red{Above and throughout this paper, we use the convention $\inf\varnothing=+\infty$.}
It then merges different 
sample splits by averaging the corresponding e-values---preserving validity---upon which e-BH is applied. The derandomized versions of SC ND and AdaDetect
stabilize their original procedures; 
however, the averaging step will dampen the signal of all but the most significant outliers, as some splits may give $e_j=0$ and shrink the average.

\subsection{Additional related literature}
Our work contributes to the literature on  model-free novelty detection.
In addition to the works introduced in Section~\ref{subsec:cnd_lit}, some authors 
tackle this problem under different settings. 
For example,~\citet{liang2022integrative} assume there are labeled outliers in 
the reference set;~\citet{zhao2024false}
suppose that the hypotheses have some informative structure;~\citet{magnani2024collective} 
consider detecting the existence of outliers and 
providing a lower bound on the number of outliers.

Another closely related line of work, referred to as {\em conformal selection}~\citep{jin2023selection,jin2023model}, considers a slightly different setting where the test set consists of unlabeled units 
and the goal is to pick out the ones whose labels satisfy a desired pre-specified property.  
It is worth noting that a recent work~\citep{bai2024optimized} proposes a generalization of 
the conformal selection procedure, in which a variant also makes use of the full conformal p-values:
although their problem setting, proposed algorithm, and proof techniques are distinct from ours, 
there turns out to be an elegant connection between a variant of their proposed procedure 
and a special instance of ours, on which we will elaborate in Section~\ref{sec:fc-ebh}.

Our approach is based on the e-BH procedure, which has garnered interest due 
to its dependence-agnostic FDR control.
To further enhance the power of e-BH,~\citet{lee2024boosting} leverage the conditional calibration technique~\citep{fithian2020conditional}, 
using conditional distributional information about the e-values to boost the power. 
We adapt their framework (described in Section~\ref{sec:ebhcc}) to our FC e-values, 
extending it to accommodate more flexible boosting mechanisms and enabling a more efficient implementation.

%% file: sections/fc_ebh.tex

\red{
We now introduce the proposed framework,
{\em $K$-block full conformal novelty detection 
($K$-FC ND)}, where the 
nonconformity scoring models are 
trained on both reference and test data.
Section~\ref{sec:kblocks} describes the proposal in its 
most general form: the test data are partitioned into 
$K$ blocks, each of which is used to train its own 
scorer. 
We then examine the two practical, recommended choices for 
instantiating $K$-FC ND, coinciding with the two 
canonical choices of $K$: $K=1$ 
(Section~\ref{sec:default_setting}), 
where all test points share the same 
trained scorer; and $K=m$
(Section~\ref{sec:mfc}), where each test point
has their own trained scorer.
We also use the latter instantiation to detail 
our  full-data  model selection strategy, where 
the choice of ML model can be adaptive to 
each test point (though our strategy generalizes to 
$K  >1$ in general).
We conclude with Section \ref{sec:ebhcc},
where we implement an 
enhanced version of e-BH-CC to boost the power 
of our e-value-based  $K$-FC ND.  
}

\subsection{$K$-block full conformal novelty detection}
\label{sec:kblocks} 
Letting $K \in[m]$, we split $\Dtest$ into $K$ non-overlapping blocks 
and denote the partition of $\Dtest$ as $B_1 \cup \cdots \cup B_K = \Dtest$. 
The training of score functions is performed $K$ times on augmented reference sets: 
for each block $k \in [K]$,  we train a score function $V^{(k)}(\cdot)$ 
over the entire dataset $\Dref \cup \Dtest$, where the training must be invariant to the ordering of 
the input samples from $\Dref \cup B_k$.
Each of the reference and test units is then scored with $V^{(k)}(\cdot)$, 
giving the collection of scores $\cV^{(k)} = \{V^{(k)}_i\}_{i\in[n+m]}\coloneqq  \{ V^{(k)}(Z_i)\}_{i\in[n+m]}$. 
For each $j\in[m]$, we construct a {\em $K$-FC e-value} $e_j$ to quantify the evidence against $H_j$:
\begin{align}\label{eq:kfc_eval}
    \begin{split} 
        e_j &= (n+1)\cdot \frac{ \ind\{V^{(k)}_{n+j} \ge T_k\}}{1+\sum_{i\in[n]}\ind\{V_i^{(k)} \ge T_k\}},\qquad \red{\text{where }k \in [K] \text{ is such that } Z_{n+j} \in B_k} \\
        \text{and }T_k  &= \inf\left\{ t\in \cV^{(k)}   \colon \frac{m}{ n+1} \cdot \frac{1 + \sum_{i=1}^{ n} \indc{V^{(k)}_i \ge t} }{1\vee \sum_{j=1}^m \indc{V^{(k)}_{n+j} \ge t}   } \le \tilde \alpha  \right\}.
    \end{split}
\end{align}

Here, the form of $e_j$ resembles that of the SC e-values~\eqref{eq:sc_eval},
but the score function and the threshold $T_k$ are now block-dependent, with
the training and calibration steps both using all the reference data.
Note that we also allow the construction of $T_k$ to use a threshold $\tilde \alpha \in (0,1)$
that is potentially distinct from the target FDR level $\alpha$;
\red{we take the default choice to be $\tilde\alpha = \alpha$ and leave 
it unspecified when so}. Finally, we  
apply the e-BH procedure at level $\alpha$ to the e-values $\{e_j\}_{j\in[m]}$ 
and output the rejection set $\cR_\alpha^\ebh(e_1,\ldots,e_m)$.
We summarize the $K$-FC ND procedure in Algorithm~\ref{alg:kfc} and visualize 
its workflow in Figure~\ref{fig:scheme}.

\input{tables-and-figures/kfc_schematic.tex}

The following theorem (with proof deferred to Appendix~\ref{appd:proof_eval_valid}) states that the $K$-FC e-values~\eqref{eq:kfc_eval}  
are valid e-values for \textit{any} choice of $\tilde\alpha \in (0,1)$.
This allows us to use e-BH (at level $\alpha$) on these e-values to get a rejection set with FDR control guaranteed at $\alpha$.

\begin{theorem}\label{thm:kfc_valid}

    Assume that for any $j\in \cH_0$, the distribution of $(Z_1,\ldots,Z_n,Z_{n+j})$
    is exchangeable conditional on $\{Z_{n+\ell}\}_{\ell \in [m]\backslash \{j\}}$.
    \red{Furthermore, assume that the partition $\{B_k\}_{k\in[K]}$ is 
    either fixed or derived independently of the data}.
    Then, the $K$-FC e-values $e_1, \dots, e_m$ constructed as per \eqref{eq:kfc_eval} 
    with any fixed $\tilde\alpha \in (0,1)$ satisfy $\EE[e_j] \le 1$, for all $j \in \cH_0$. 
    Hence, Algorithm~\ref{alg:kfc} controls the FDR at level $\pi_0\alpha $, i.e.,
    $\fdr(\cR_\alpha^\ebh(e_1, \dots, e_m)) \le \pi_0 \alpha$.
    
\end{theorem}
This $\pi_0\alpha$ guarantee can be \red{made less conservative} via null proportion estimation, \red{resulting in FDR control at $\alpha$}; we defer the details to Appendix~\ref{sec:null_prop}.

\begin{algorithm}[!b]
\ifshowrevision\color{red!60}\fi
\footnotesize
\SetAlgoNlRelativeSize{-1}
\SetAlCapFnt{\normalsize}
\SetAlCapNameFnt{\normalsize}
\DontPrintSemicolon
\caption{$K$-block full conformal ND ($K$-FC and LOO-$K$-FC ND)}\label{alg:kfc}
\KwIn{$\Dref=\{Z_i\}_{i=1}^n$; $\Dtest=\{Z_{n+j}\}_{j=1}^m$;
model training routine $f(S;S')$
symmetric in $S$; number of blocks $K$;
LOO indicator \texttt{loo}; target FDR level $\alpha$.}

Partition $\Dtest$ into $K$ blocks $B_1, \dots, B_K$.

\For{$k \in [K]$}{
    $C_k\gets\Dref\cup B_k$ and $O_k\gets\Dtest\backslash B_k$.\;
    \If{$\textnormal{\bfseries not}\ \textnormal{\texttt{loo}}\ \textnormal{\bfseries or}\ O_k\neq\varnothing$}{
        Train $V^{(k)}(\cdot)\gets f(C_k;O_k)$.\;
    }
    \eIf{\textnormal{\texttt{loo}}}{
        \For{$i\colon Z_i\in C_k$}{
            Train $V^{(k,-i)}(\cdot)\gets f(C_k\backslash\{Z_i\};O_k)$ and set
            $V_i^{(k)}\gets V^{(k,-i)}(Z_i)$.\;
        }
    }{
        Set $V_i^{(k)}\gets V^{(k)}(Z_i)$ for every $Z_i\in C_k$.\;
    }
    Set $V_i^{(k)}\gets V^{(k)}(Z_i)$ for every $Z_i\in O_k$.\;
    Compute threshold $T_k$ using $\{V_{i}^{(k)}\}_{i\in[n+m]}$ as in \eqref{eq:kfc_eval}.\;

    \For{$j\colon Z_{n+j}\in B_k$}{
        Compute $e_j$ using $\{V_i^{(k)}\}_{i\in[n+m]}$ and $T_k$ as in \eqref{eq:kfc_eval}.
    }
}
$\mathcal{R}\gets\cR^\ebh_\alpha(e_1,\dots,e_m)$ (e-BH at level $\alpha$).

\KwOut{Rejection set $\mathcal{R}$}
\PrintSemicolon
\end{algorithm}

\paragraph{The LOO-$K$-FC framework.}
We can also extend $K$-FC ND to its 
leave-one-out (LOO) version, where the model which scores unit $Z_i$ does not see $Z_i$ during training. 
To be specific, for each $k\in [K]$ and 
\red{each unit $Z_i \in \Dref \cup B_k$}, we train $V^{(k,-i)}(\cdot)$ 
\red{invariant to the order of $\Dref\cup B_k\setminus\{Z_i\}$ instead.}
The LOO-$K$-FC e-values then \red{use, for each $k \in [K]$, the
LOO scores  $V^{(k)}_i = V^{(k,-i)}(Z_i)$ for $Z_i\in \Dref\cup B_k$
as well as the non-LOO scores $V^{(k)}_{n+\ell} =   V^{(k)}(Z_{n+\ell})$
for $Z_{n+\ell} \in \Dtest\setminus B_k$.  
Algorithm~\ref{alg:kfc} describes the exact details of the 
LOO-$K$-FC ND procedure under its LOO-specific logic.
}
The validity of LOO-$K$-FC e-values follows from the more general  argument in Appendix~\ref{appd:proof_eval_valid}.
The LOO-$K$-FC score function avoids overfitting by not using the unit $Z_i$ 
\red{when training its scoring function $V^{(k, -i)}$;
however, it is more computationally expensive since 
one must  fit $O(m+nK)$ models, as per the 
algorithm.}

We conclude this section with several remarks.

\begin{remark}[The role of $K$]
As we elaborate in the sections that follow, the two canonical choices \red{$K \in \{1,m\}$ 
\textbf{correspond to our two main
practical recommendations} for  
implementing $K$-FC ND.}
For \red{ordinary (non-LOO)} $1$-FC, the former case is computationally light, requiring only a single operation (e.g., model fit/model selection),
but the resulting scores may be suboptimal 
\red{as the scoring} models may be contaminated by the true outliers in the training set.
In contrast, the latter offers more flexibility in training---often yielding 
higher-quality score functions and allowing for data-driven model selection (see Section \ref{sec:mfc})---but comes at higher computational cost \red{due to 
requiring more model fits}.
We present our framework in terms of the general parameter $K$
\red{to highlight the  fact that our recommendations 
are instances of the general $K$-FC ND,
but this paper focuses on detailing 
the two canonical instantiations. We leave a 
deeper investigation of the choice of $K$ and the 
resulting trade-off between computational efficiency 
and statistical power to future work (see Section~\ref{sec:disc}).}

\end{remark}

\begin{remark}[External randomness]
When $1<K<m$, external randomness can be introduced 
\red{if the user chooses to partition the test data 
randomly.
On the other hand, our practical recommendations 
$K \in \{1,m\}$ lead to 
novelty detection algorithms with the 
desirable trait of being deterministic (given 
the data).}
\end{remark}
 
\begin{remark}[$K$-FC p-values]
\label{remark:kfc_pval}
Having obtained the scores $\{\cV^{(k)}\}_{k\in[K]}$, it is natural to consider the corresponding {$K$-FC p-values}, 
defined for each $j\in[m]$ as
\@\label{eq:kfc_pval}
p_j = \frac{1+\sum^n_{i=1}\ind\{V^{(k)}_i \ge V^{(k)}_{n+j}\} }{n+1},
\@
where $B_k$ is the block to which $Z_{n+j}$ belongs. 
The exchangeability among $V^{(k)}_1,\ldots,V^{(k)}_n,V^{(k)}_{n+j}$
ensures the validity of $K$-FC p-values. 

For $K=1$, we show in Section~\ref{sec:default_setting} that 
BH retains FDR control when run on $\{p_j\}_{j\in[m]}$ by establishing its equivalence to $1$-FC ND. However, for $K \ge 2$, neither equivalence nor FDR control is guaranteed. 
We briefly explain the reason. 
In the standard proof of FDR control with conformal p-values (see e.g.,~\citet{bates2023testing,marandon2024adaptive}), 
the key step is to show that the p-values are PRDS, for which BH is known to control the FDR~\citep{benjamini2001control}.
When $K\ge 2$, however, the $K$-FC p-values \red{need not be PRDS in general}.
At a high level, consider two points, $Z_{n+1}$ and $Z_{n+2}$, which belong to different blocks. The trained scoring function used on $Z_{n+1}$ can depend on $Z_{n+2}$ in a different way than  on $(Z_1, \dots, Z_{n}, Z_{n+1})$, which it treats as unordered (and vice versa).

\red{As the $K$-FC p-values need not be PRDS, BH no longer
guarantees FDR control with these p-values. In practice,
one may expect BH to still robustly control the FDR
empirically while achieving reasonable statistical 
power---we find evidence of this in
our own experiments, as in Section~\ref{sec:exp_ms}.
However, as~\citet{dobriban2026benjamini, lei2026how} 
show in the simpler setting of  
dependent Gaussian p-values, the robustness of BH is 
more nuanced than previously thought. 
We likewise show an explicit FDR violation 
when $K = m$ in Appendix~\ref{appd:model_selection_fdr_inflation}.
}

\end{remark}

Remark \ref{remark:kfc_pval} underlines the lack of guarantees from the p-value perspective.  One can restore FDR control through corrective measures; an example is to randomly prune the rejection set~\citep{fithian2020conditional,jin2023model, bai2024optimized}. However, this lowers power and can re-introduce or exacerbate instability. 
These issues motivate the choice of e-values as the primary tool. As they give FDR control for free, the focus is then on designing valid, powerful e-values.
\vspace{-.75\baselineskip}
\subsection{A canonical instantiation: $K=1$}
\label{sec:default_setting}
Our first canonical instantiation of $K$-FC ND procedure is \red{obtained by choosing} $K=1$.
In this case, each reference and test unit is scored with a function $V(\cdot)$ trained over the \textit{entire dataset} $\Dref \cup \Dtest$, 
with the requirement that the training procedure is invariant to the order of the training datapoints.
Here, the training algorithm has access to both inliers and outliers \red{without knowing
whether they originate}  from $\Dref$ or $\Dtest$. 
When using ML models suitable for (unsupervised) \textit{outlier detection}---such as the Isolation Forest \citep{liu2012isolation} or the Local Outlier Factor \citep{breunig2000lof}---the resulting trained $V$ 
can effectively distinguish between the outliers and inliers of $\Dtest$.
Outlier detection models enable $1$-FC ND to only require a single training operation, making it highly scalable.

Alternatively, when using one-class classifiers (e.g., the one-class SVM~\citep{scholkopf1999support}) rather than outlier detection models, 
one might find that estimating the inlier distribution support over the outlier-contaminated $\Dref \cup \Dtest$ can overfit to these outliers. To make the model more robust to detecting each outlier, 
we can adopt the LOO-$1$-FC framework 
to ensure each test unit is scored by a model that does not see it during training. The resulting
procedure mitigates overfitting at the cost of $m+n$ training operations.

\paragraph{Connection to BH with FC p-values when $K=1$.}
As noted in Remark~\ref{remark:kfc_pval}, one can also construct full conformal p-values~\eqref{eq:kfc_pval}
and apply BH, though this is not provably valid for $K\ge 2$.
The $K=1$ case, by contrast, is special: existing results can be independently adapted to show that 
applying BH to the $1$-FC p-values {\em does} achieve theoretical FDR control. 
~\citet[Theorem 2]{yang2021bonus} essentially proves
this in the narrower context of testing parametric nulls.
In a different context,
\citet{bai2024optimized} use LOO-$1$-FC p-values
and show that BH achieves finite-sample FDR
control~\citep[Theorem 2]{bai2024optimized}.\footnote{\citet{bai2024optimized} consider testing random hypotheses, while we focus on deterministic ones.}

Interestingly, the FDR guarantee of BH---applied to either 
$1$-FC or LOO-$1$-FC p-values---is implied by 
Theorem~\ref{thm:kfc_valid}, which uses the e-BH perspective. 
Specifically, for both scoring frameworks, the BH rejection set 
$\cR^\bh_\alpha(p_1,\ldots,p_m)$ 
coincides with the e-BH rejection set $\cR_\alpha^\ebh(e_1,\red{\ldots},e_m)$ (where both procedures are \red{run} at $\alpha \in(0,1)$ and $\tilde \alpha = \alpha$ in the e-value construction),
as shown in~\citet[Proposition 8]{lee2024boosting}. 
As a result, the FDR control 
of $\cR^\ebh_\alpha(e_1,\ldots,e_m)$ carries over to  $\cR^\bh_\alpha(p_1,\ldots,p_m)$ for this canonical setting.
Moreover, the e-value perspective 
\red{enables a power-boosting mechanism via conditional calibration 
(to be introduced in Section~\ref{sec:ebhcc}),
which, for p-value-based procedures, has only been used to correct potential 
FDR violation;} see Appendix~\ref{appd:ebhcc_power_over_bh} for a concrete example.

\subsection{A canonical instantiation: $K=m$}
\label{sec:mfc}
The other canonical instantiation we consider is the case where $K=m$, where each  test unit $Z_{n+j}$ is its own block $B_j$.
Compared with $K=1$, this version allows for more flexible training of the scoring functions:
for each $j\in [m]$, we only require invariance to the order of $\Dref \cup B_j$.
This flexibility can be used to improve the quality of the scoring functions. %
Consider, for instance, a one-class classifier: by combining $m$-FC with the LOO framework, the scoring function for $Z_{n+j}$ can be trained using only  $\Dref \cup B_j \setminus \{Z_{n+j}\} = \Dref$, which is free from contamination.

The training of $V^{(j)}(\cdot)$ is also allowed to depend on $\cD_{-j} :=\Dtest \backslash B_j$ {\em with} the \red{elements'} identities. For example,
we can choose to train 
{\em PU (positive unlabeled) classifiers},
treating $\Dref \cup B_j$ as positive samples and $\cD_{-j}$ as unlabeled samples 
(a full conformal variation of~\citet{marandon2024adaptive}). 
\red{We can also utilize $\cD_{-j}$ for data-driven model selection specific to $Z_{n+j}$, as we detail below.}

\subsubsection{Model selection within the $m$-FC ND framework}
\label{sec:model_selection_mfc}

\red{In addition to cases where each test unit trains its own scorer,
the $m$-FC framework also accommodates  
data-driven model selection to use the best ML model for each 
test point without compromising FDR guarantees.}

Suppose we are given a suite of $L$ candidate ML models $f^{(1)}, \dots, f^{(L)}$,
chosen \textit{a priori}. 
For each $j\in[m]$, a natural question is to ask which of the $L$ models will give the \red{most powerful} score for $Z_{n+j}$. 
If, for each $\ell \in [L]$, 
we generate scores via the $m$-FC (or LOO-$m$-FC) framework, then we are left with a jointly exchangeable collection of $L$-tuples
$\{ (V_i^{(j, 1)}, \dots, V_i^{(j, L)}) \}_{i\in[n] \cup \{n+j\}}$
with $V_i^{(j, \ell)}$ denoting the $\ell$-th model's score for $Z_i$.
In fact, these $L$-tuples are exchangeable conditional on both $\cD_{-j}$ and $\Lbag \Dref\cup\{Z_{n+j}\} \Rbag$, where $\Lbag A \Rbag$ denotes the \red{multiset} $A$ without its order.
By the proof of Theorem \ref{thm:kfc_valid}, $e_j$ will be valid when its constituent scores satisfy such conditional exchangeability. Hence, when we have multiple such scores, we may use $\cD_{-j}$ and $\Lbag \Dref\cup\{Z_{n+j}\} \Rbag$ to select among them.

We formalize the idea as the following proposition, with proof deferred to Appendix~\ref{appd:proof_ms_validity}.
\begin{proposition}
    \label{prop:ms}
    For each $j \in \cH_0$, let $S_j := (\Lbag \Dref\cup\{Z_{n+j}\}\Rbag,\cD_{-j})$. Assume that 
    \begin{enumerate}[label=(\roman*),leftmargin=*,nosep]
    \item $\{ (V_i^{(j, 1)}, \dots,  V_i^{(j, L)}) \}_{i\in[n]\cup\{n+j\}}$
    are jointly exchangeable conditional on $S_j$, and
    \item $\{ (V_{n+r}^{(j, 1)}, \dots,  V_{n+r}^{(j, L)}) \}_{r \neq j}$ are $\sigma(S_j)$-measurable.
    \end{enumerate}
    For each $\ell \in [L]$, let $\Gamma_\ell$ be a $\sigma(S_j)$-measurable quality metric of $f^{(\ell)}$.
    Define {new} scores
    $V^{\mathrm{MS},(j)}_i = g((V_i^{(j, 1)}, \Gamma_1), \dots, (V_i^{(j, L)}, \Gamma_L))$
    for $i\in[n+m]$, using some combination function $g$.
    Then the e-value $e_j$  
    \eqref{eq:kfc_eval} using scores $\{V_i^{\mathrm{MS},(j)}\}_{i\in [n+m]}$ is valid.
    
\end{proposition}

\red{This strategy extends to
any $K  >1$ for block-specific model selection. We describe it in  
the $m$-FC setting, where the task of choosing the best model 
for each individual test unit is particularly intuitive.}
We conclude \red{with} two concrete instances of Proposition \ref{prop:ms}.
\paragraph{Best model selection.}

To find the most powerful model, we may want $\Gamma_\ell$ to be a proxy rejection count for the $\ell$-th model. For each $\ell$, we sort the values in $\Lbag V_1^{(j,\ell)},\dots,V_n^{(j,\ell)},V_{n+j}^{(j,\ell)}\Rbag$ and take its $n$ smallest elements as a proxy reference set, treat the remaining one together with $\{V_{n+r}^{(j,\ell)}\}_{r\neq j}$ as a proxy test set, and let $\Gamma_\ell$ be the number of rejections produced by $1$-FC ND on this proxy problem. We then select the model
$\hat \ell(j)\in \textstyle\operatorname{argmax}_{\ell\in[L]}\;\Gamma_\ell$
and define
$V_i^{\mathrm{MS},(j)} = 
 V_i^{(j,\hat \ell(j))}$ for $i\in[n+m]$. By  Proposition~\ref{prop:ms}, 
 \red{these lead to valid e-values. }

\paragraph{Top model ensembling.}

Rather than choosing the best model, we may rank the candidate models by $\Gamma_1,\dots,\Gamma_L$ (e.g., the proxy rejection counts above), retain a subset $M_j\subseteq [L]$ of top models, and ensemble their scores. 
For some choice of $\lambda > 0$, compute weights 
$w_{j,\ell} = \frac{\exp(\lambda \Gamma_\ell)}{\sum_{r\in M_j}\exp(\lambda \Gamma_r)}$
for $\ell\in M_j$  and define
$V_i^{\mathrm{MS},(j)}=\sum_{\ell\in M_j} w_{j,\ell}V_i^{(j,\ell)}$ for $i\in[n+m]$. As $\lambda \to +\infty$, this approaches best model selection. Note that when different models have different scales of scores, one can apply exchangeability-preserving common-scale transforms.

\red{
\begin{remark}[Model selection $m$-FC p-values with BH]
    From our discussion in   Remark~\ref{remark:kfc_pval} about $K$-FC 
    p-values, we know 
    that $m$-FC p-values with BH do not have FDR control guarantees and are not equivalent to the $m$-FC e-BH procedure. In simulations (Section~\ref{sec:experiments}), we find 
    the FDR to be controlled empirically, even when 
    the p-values are made with the model-selected scores.
    However, due to the lack of equivalence, we find that
    p-values are not necessarily more powerful 
    than our conditionally calibrated $m$-FC e-value 
    approach (Section~\ref{sec:ebhcc}); 
    see Figure~\ref{fig:mfc_bh_vs_ebhcc} in
    Appendix~\ref{appd:mfc_bh_vs_ebhcc} for an example.
\end{remark}
}

\subsection{Improving power through e-BH-CC}
\label{sec:ebhcc}
As our method uses e-BH as the selection procedure, 
we can boost its power using the e-BH-CC framework developed in~\citet{lee2024boosting}.
Central to the instantiation of e-BH-CC is to identify a ``sufficient statistic'' $S_j$ 
such that the distribution $(e_1,\ldots,e_m) \given S_j$ is known under $H_j$. 
\red{This knowledge is then used to ``boost'' the e-values 
to more powerful versions of themselves.}
For the $K$-FC e-values, a natural choice of $S_j$ is 
$(\Lbag \Dref \cup \{Z_{n+j}\}\Rbag, \cD_{-j})$;
conditional on $S_j = (\Lbag z_1,\ldots,z_n,z_{n+j}\Rbag,  
\{z_{n+\ell}\}_{\ell \neq j})$, 
\begin{equation} \label{eq:cond_dist}
(Z_1,\ldots,Z_n,Z_{n+j}, \{Z_{n+\ell}\}_{\ell \neq j})
\stackrel{H_j}{\sim} \sum_{\pi \in \Sym([n] \cup \{n+j\})} 
\frac{1}{(n+1)!}\cdot 
\delta_{(z_{\pi(1)},\ldots,z_{\pi(n)},z_{\pi(n+j)}, 
\{z_{n+\ell}\}_{\ell \neq j})},
\end{equation}
where 
$\delta_z$ denotes a point mass at $z$.
With~\eqref{eq:cond_dist},
we can in fact resample the data set $\Dref\cup \Dtest$, or any  
function of the data set, given $S_j$.

The boosting method then proceeds as follows. For each 
$j\in[m]$, define the following quantity for all $c \in [0,1]$:\footnote{ 
The range of $c$ is set to be $[0,1]$ without loss of generality by scaling.}
\begin{equation} \label{eq:cond_exp}
\phi_j(c;S_j) := \EE_{H_j}\bigg[\frac{m \cdot \ind\{e_j \ge \frac{m}{\alpha |\cR\cup \{j\}|} 
\text{ or } j \in \cA(c)\}}
{\alpha |\cR\cup \{j\}|} - e_j
\Biggiven S_j \bigg], 
\end{equation}
where $\cR = \cR^\ebh_\alpha(e_1,\ldots,e_m)$ and
$\cA(c) \subseteq [m]$ is an ``auxiliary rejection set''
such that \red{(1)}
it is
nondecreasing in $c$
\red{(i.e., $c_1 \le c_2 \implies \cA(c_1) \subseteq \cA(c_2)$);
and (2) $\cA(0) =\varnothing$. }
For example, we can let $\cA(c) = \{j\in[m]: p_j \le c\}$, 
where $p_j$ is the $K$-FC p-value defined in~\eqref{eq:kfc_pval}.

\red{In this setting, $\phi_j(c;S_j)$ is a known function in $c$, 
as the conditional distribution given $S_j$ is known (under $H_j$).
One can therefore 
check which values of $c \in[0,1]$  
give $\phi_j(c;S_j) \le 0$. The first term inside
the expectation in~\eqref{eq:cond_exp} at such values of $c$ has 
marginal expectation bounded by $\EE_{H_j}[e_j]$ by 
the tower property, implying it is 
a valid e-value at those $c$. 
}

\red{
Hence, if we denote $q_j = \min\{c\in[0,1] \colon j \in \cR \cup \cA(c)\}$
(and assume this minimum can indeed be attained), then we can define the
boosted e-values as 
\begin{equation}
    \label{eq:boosted_eval}
    e_j^\boost = \frac{m}{ \alpha | \cR \cup \{j\}|}\cdot 
    \ind\{\phi_j(q_j ;S_j) \le 0\}, \qquad j\in [m].
\end{equation}
At a high level, these are valid e-values because $q_j$ 
selects the smallest $c$ for which the first term above is nonzero for 
hypothesis $j$, while the condition $\phi_j(q_j;S_j)\le 0$ ensures that this 
nonzero value is retained only when its marginal expectation is bounded by 
that of $e_j$. We then apply e-BH to the boosted e-values 
$(e_1^\boost,\dots, e_m^\boost)$, obtaining a 
rejection set $\cR^{\textnormal{e-BH-CC}}$.\footnote{
    \red{See Appendix~\ref{appd:additional} for a more general 
    formulation of e-BH-CC 
    closely following~\citet{lee2024boosting}}.
}
}

The following proposition rigorously shows that the boosted 
e-values defined through~\eqref{eq:boosted_eval} are valid e-values 
and that the boosted e-values deterministically improve upon 
the original ones.
The proof of Proposition~\ref{prop:boost_eval} is delegated to
Appendix~\ref{appd:proof_boost_eval}. 

\begin{proposition}\label{prop:boost_eval}
The boosted e-values defined in~\eqref{eq:boosted_eval} are 
valid e-values, i.e., $\EE[e_j^\boost] \le 1$, for $j\in \cH_0$.
Moreover, $\cR^\ebh_\alpha(e_1,\ldots,e_m) \subseteq 
\cR^\ebh_\alpha(e_1^\boost,\ldots,e_m^\boost)$.\footnote{\red{
Certain alternatives for $|\cR \cup \{j\}|$ in $\phi_j$ preserve validity and uniform improvement for our e-BH-CC implementation; see Appendix~\ref{appd:ebhcc_specific_implementations} for the choice(s) used in our experiments.}} 
\end{proposition}

\begin{remark}[Comparison with~\citet{lee2024boosting}]
Our boosting scheme above is a strict generalization of the 
original e-BH-CC framework:
here, we boost via an auxiliary rejection set $\cA(c)$ which can be customized as long as it is nondecreasing in $c$.
In the original version, \red{a multiplicative factor} is 
applied to $e_j$ to enable more rejections\red{;} this can be \red{roughly} viewed as a special case 
of our \red{formulation} by letting $\cA(c) = \{j\in[m]: e_j \red{c/(1-c)}\ge m/(\alpha |\cR\cup \{j\}|)\}$. 
In our \red{unweighted} implementation, we choose $\cA(c) = \{j\in[m]: p_j \le c\}$, 
where $p_j$ is the $K$-FC p-value defined in~\eqref{eq:kfc_pval}.
This improvement gives two main benefits: it (1) allows for boosting even when $e_j=0$ (which, in the conformal setting, happens with positive probability),
and (2) leverages both the  e-values and p-values, achieving, to some extent, the best of both worlds.
\end{remark}

Constructing the boosted e-values in \eqref{eq:boosted_eval} is the
major computational cost of e-BH-CC, 
motivating us to implement several \red{improvements. Notably,}
evaluating $\phi_j(q_j;S_j)$ amounts to taking the average 
over $n+1$ values, as the quantity inside the 
expectation in~\eqref{eq:cond_exp} is supported on a 
set of at most $n+1$ elements given $S_j$. This implies that e-BH-CC can be implemented exactly, without Monte-Carlo estimation (see details in Appendix~\ref{appd:sampling_dist}).

\red{Typically}, a majority of these $n+1$ summands will be zero, \red{so} we can further reduce the computation 
by \red{evaluating only} the nonzero terms; see Appendix~\ref{appd:comp_shortcut}. 
As a result of these shortcuts, calculating the conditional expectation in~\eqref{eq:cond_exp} requires fitting far fewer than $O(nK)$ models.
Algorithm~\ref{alg:kfc_nd_cc} (Appendix \ref{appd:algos}) details the boosted $K$-FC ND 
procedure with the computational shortcuts.

%% file: tables-and-figures/kfc_schematic.tex
\begin{figure}[bt!]
\centering
\begin{tikzpicture}[
  selection/.style={draw=gray!72,fill=gray!13,line width=.85pt,rounded corners=2pt},
  dataset/.style={draw=black,dashed,line width=.8pt,rounded corners=2pt},
  process/.style={draw=gray!72,fill=gray!13,rounded corners=2pt,line width=.85pt,align=center,font=\scriptsize,inner sep=4pt},
  selectedflow/.style={-{Stealth[length=2.2mm,width=1.7mm]},draw=gray!72,line width=.85pt},
  tag/.style={font=\scriptsize,text=black!72},
]
  \draw[selection] (0.00,0.50) rectangle (7.20,1.70);

  \draw[dataset] (0.08,0.73) rectangle (5.32,1.47);
  \draw[dataset] (5.48,0.73) rectangle (12.52,1.47);

  \foreach \i/\lab in {0/Z_1,1/Z_2,2/\cdots,3/\cdots,4/Z_{n-1},5/Z_n}{
    \fill[blue!13] ({0.45+0.90*\i},1.10) circle (0.33);
    \draw[blue!65,line width=.8pt] ({0.45+0.90*\i},1.10) circle (0.33);
    \node[font=\tiny] at ({0.45+0.90*\i},1.10) {$\lab$};
  }
  \foreach \i/\lab in {6/Z_{n+1},7/Z_{n+2}}{
    \fill[blue!13] ({0.45+0.90*\i},1.10) circle (0.33);
    \draw[blue!65,line width=.8pt] ({0.45+0.90*\i},1.10) circle (0.33);
    \node[font=\tiny] at ({0.45+0.90*\i},1.10) {$\lab$};
  }
  \foreach \i/\lab in {8/Z_{n+3},9/Z_{n+4},10/\cdots,11/\cdots,12/\cdots,13/Z_{n+m}}{
    \fill[gray!20] ({0.45+0.90*\i},1.10) circle (0.33);
    \draw[gray!72,line width=.8pt] ({0.45+0.90*\i},1.10) circle (0.33);
    \node[font=\tiny] at ({0.45+0.90*\i},1.10) {$\lab$};
  }

  \node[font=\scriptsize,text=black,fill=gray!13,inner xsep=2pt,inner ysep=.4pt]
    at (2.70,1.47) {$\Dref$};
  \node[font=\scriptsize,text=black,fill=white,inner xsep=2pt,inner ysep=.4pt]
    at (9.00,1.47) {$\Dtest$};

  \draw[gray!72,line width=.75pt,decorate,decoration={brace,amplitude=4pt,mirror}]
    (5.55,0.42)--(7.05,0.42);
  \node[tag] at (6.30,0.03) {$B_1$};
  \draw[gray!72,line width=.75pt,decorate,decoration={brace,amplitude=4pt,mirror}]
    (7.35,0.42)--(8.85,0.42);
  \node[tag] at (8.10,0.03) {$B_2$};
  \draw[gray!72,line width=.75pt,decorate,decoration={brace,amplitude=4pt,mirror}]
    (9.15,0.42)--(10.65,0.42);
  \node[tag] at (9.90,0.03) {$\cdots$};
  \draw[gray!72,line width=.75pt,decorate,decoration={brace,amplitude=4pt,mirror}]
    (10.95,0.42)--(12.45,0.42);
  \node[tag] at (11.70,0.03) {$B_K$};

  \node[process,minimum width=2.70cm,minimum height=1.02cm] (train) at (14.15,1.10)
    {Train $V^{(k)}$ on\\$\Dref\cup B_k$\\and $\{B_\ell\}_{\ell\ne k}$};
  \draw[selectedflow] (3.60,1.70)--(3.60,2.08)-|(train.north);

  \draw[selection] (0.00,-1.76) rectangle (12.60,-0.84);
  \foreach \i/\lab in {0/V^{(k)}_1,1/V^{(k)}_2,2/\cdots,3/\cdots,4/V^{(k)}_{n-1},5/V^{(k)}_n,6/V^{(k)}_{n+1},7/V^{(k)}_{n+2},8/V^{(k)}_{n+3},9/V^{(k)}_{n+4},10/\cdots,11/\cdots,12/\cdots,13/V^{(k)}_{n+m}}{
    \fill[red!12] ({0.45+0.90*\i},-1.30) circle (0.33);
    \draw[red!65,line width=.8pt] ({0.45+0.90*\i},-1.30) circle (0.33);
    \node[font=\tiny] at ({0.45+0.90*\i},-1.30) {$\lab$};
  }
  \draw[selectedflow] (train.south)--(14.15,-0.40)-|(11.70,-0.84);

  \draw[selectedflow] (12.60,-1.30)--(13.50,-1.30);
  \draw[selection] (13.50,-1.76) rectangle (15.18,-0.84);
  \node[tag] at (14.34,-0.65) {Compute $\forall Z_{n+j}\in B_k$};
  \foreach \x/\lab in {13.92/e_1,14.76/e_2}{
    \fill[red!12] (\x,-1.30) circle (0.36);
    \draw[red!65,line width=.9pt] (\x,-1.30) circle (0.36);
    \node[font=\scriptsize] at (\x,-1.30) {$\lab$};
  }
\end{tikzpicture}
\caption{A schematic illustration of $K$-FC ND (Algorithm \ref{alg:kfc}). 
For each $Z_{n+j}$ contained in $B_k$, 
the e-value is constructed according to~\eqref{eq:kfc_eval} with $\{V^{(k)}_i\}_{i\in[n+m]}$. Note that, in the specific example above, since
$B_1$ only contains $Z_{n+1}$ and $Z_{n+2}$,
$V^{(1)}$ is only used to construct the e-values $e_1$ and $e_2$.
}
\label{fig:scheme}
\end{figure}

%% file: sections/shift.tex

Quite often, the assumption that the inlier distributions of $\Dref$ and $\Dtest$  are \red{identical} can be called into question. When the reference data is collected in a certain way to guarantee its status as an inlier, or when the inlier data can only be collected from a certain subset of the general population, then the reference inliers and test inliers would not necessarily be alike. For example, consider the problem of detecting which prompts are adversarial to LLMs (Section \ref{sec:intro}). If the reference data---benign prompts---were collected by sampling a subset of individuals from the population 
and collecting prompts from them, one can expect that these prompts may still look different than the benign prompts sent by a general user. The selection of prompts based on features of the prompter induces a 
\textit{distribution shift} between $\Dref$ and inliers of $\Dtest$, and the traditional methods of \citet{bates2023testing} as well as our proposal in Section \ref{sec:fc-ebh} are no longer \red{guaranteed to control the FDR} due to the violation of exchangeability assumptions. 

Fortunately, with knowledge of the distribution shift---feasible when collecting units based on their known covariates \citep{jin2023model}---conformal  methods can be adapted using  the assumption of \textit{weighted exchangeability} \citep{tibshirani2019conformal}. We show that 
by taking an e-value perspective, we can extend $K$-FC ND to cover this setting in a straightforward manner. Central to the power advantage of this method is the fact that the e-BH-CC framework (Section \ref{sec:ebhcc}) again improves the power of $K$-FC ND.

Formally, denote the two distributions $P$ and $Q$ such that $Z_i \sim P$ for $i\in[n]$, while the inliers in $\Dtest$ follow $Q$. The distribution shift can be expressed through the function $w(z) \coloneqq (\dif Q / \dif P)(z) $, the Radon-Nikodym derivative (assuming it exists). The multiple testing problem then attempts to reject a subset of the null hypotheses $H_j \colon Z_{n+j} \sim Q$. Note that when $w(z) \equiv 1$, the original problem setting is recovered.

\vspace{-0.4\baselineskip}
\subsection{Weighted $K$-block full conformal novelty detection}
\label{sec:weighted_kfc}

We define the weighted $K$-FC e-value underlying $K$-FC ND and state the validity of its construction.
Like in Section \ref{sec:kblocks}, we partition   $\Dtest$ into $K$ blocks $B_1 \cup \cdots \cup B_K$, 
and for each $k\in[K]$, we train a scoring function over the entire dataset in a manner that is invariant to the order of the samples in $\Dref \cup B_k$.
With the resulting collections of scores $\cV^{(k)} = \{V^{(k)}_i\}_{i\in[n+m]}\coloneqq  \{ V^{(k)}(Z_i)\}_{i\in[n+m]} $ (indexed by $k$) and the random weights $\{w_i\}_{i\in[n+m]}\coloneqq \{w(Z_i)\}_{i\in[n+m]}$, we can construct a weighted version of \eqref{eq:kfc_eval} as follows:
\begin{align}\label{eq:weighted_kfc_eval}
    \begin{split} 
        e_j &= \bigg( w_{n+j} + \sum_{i=1}^n w_i \bigg)\cdot \frac{ \ind\{V^{(k)}_{n+j} \ge T_j\}}{w_{n+j}+\sum_{i\in[n]}w_i\ind\{V_i^{(k)} \ge T_j\}},\\
        \text{and }T_j  &= \inf\left\{ t\in \cV^{(k)}   
        \colon \frac{m}{ w_{n+j} + \sum_{i=1}^n w_i} \cdot \frac{ w_{n+j}  + \sum_{i=1}^{ n}  w_i \indc{V^{(k)}_i \ge t} }{1\vee \sum_{\ell=1}^m \indc{V^{(k)}_{n+\ell} \ge t}   } \le \tilde \alpha  \right\} \\ 
    \end{split}
    \raisetag{4em}
\end{align}
\red{where $k \in [K]$ is such that $Z_{n+j} \in B_k$.}
The weighted exchangeability analogue of Theorem \ref{thm:kfc_valid} states the validity of these e-values.
\begin{theorem}
    \label{thm:weighted_eval_valid}
    Assume under the null $H_j$, $Z_1, \dots, Z_n, Z_{n+j}$ exhibit weighted exchangeability with weight function $w(\cdot)$ in the following manner: conditional on $\{Z_{n+\ell}\}_{\ell \neq j}$, the joint  density of  $Z_1, \dots, Z_n, Z_{n+j}$ satisfies
    \begin{equation}
        f(z_1, \dots, z_n, z_{n+j}) = w(z_{n+j}) \prod_{i \in [n]\cup\{n+j\}} p(z_i),
    \end{equation} 
    where $p$ is the density function of distribution $P$.  Further assume for each $k \in [K]$ that the scoring function $V^{(k)} (\cdot)$ was trained  in a way such that it is invariant to the order of $Z_1, \dots, Z_n, Z_{n+\ell}$  in the training set, for each $\ell \colon Z_{n+\ell} \in B_k$.
    \red{Furthermore, assume that the partition $\{B_k\}_{k\in[K]}$ is 
    either fixed or derived independently of the data}.
    Then\red{, for any fixed $\tilde\alpha\in(0,1)$,} the e-value $e_j$ constructed using \eqref{eq:weighted_kfc_eval} 
    is a strict e-value, i.e., $\EE[e_j] \le 1$, \red{for all $j\in\cH_0$}. Hence, $\fdr(\cR_\alpha^\ebh (e_1, \dots, e_m) ) \le \pi_0 \alpha $.
\end{theorem}
The proof of Theorem \ref{thm:weighted_eval_valid} is given in Appendix \ref{appd:proof_of_weighted_eval}.

\subsection{Applying e-BH-CC to weighted $K$-FC e-values}
\label{sec:weighted_ebhcc}

In the distribution-shift setting, weighted analogues of conformal p-values have been proposed \citep{tibshirani2019conformal, hu2023two, jin2023model}; our $K$-block full conformal setting would lead to the p-values
\begin{equation}
    \label{eq:weighted_kfc_pval}
    p_j = \frac{w_{n+j}+\sum^n_{i=1}w_i\ind\{V^{(k)}_i \ge V^{(k)}_{n+j}\} }{w_{n+j} + \sum_{i=1}^n w_i }\red{,}
\end{equation}
\red{where $k \in [K]$ is such that $Z_{n+j} \in B_k$.}
Although marginally superuniform, \red{the p-values $p_1, \dots, p_m$} are not necessarily PRDS, as shown by~\citet[Proposition 2.4]{jin2023model}. The BH procedure therefore does not guarantee FDR control (without corrections, such as random pruning of the rejection set \citep{jin2023model}), inspiring the weighted conformal e-value approach in the SC paradigm \citep{lee2024boosting}. However, they find that the weighted versions of SC BH and SC e-BH were no longer equivalent, and that the former often outperformed the latter in terms of power. Hence, the e-BH-CC framework was applied to close this power gap. We also find in our context that the p-value and e-value approaches are no longer equivalent---even when $K=1$---and similarly implement the e-BH-CC framework to boost the power of weighted $K$-FC e-BH. The resulting weighted $K$-FC e-BH-CC procedure \red{empirically improves on} $K$-FC e-BH; see Section~\ref{sec:exp_weighted} and Appendix~\ref{appd:more_experiments}.

The implementation of the e-BH-CC framework for the distribution shift setting is exactly the same as that of the default novelty detection problem---outlined in Section \ref{sec:ebhcc}---except we must account for the relaxation of the exchangeability property to weighted exchangeability.
Again take $S_j = \Lbag \Dref \cup \{Z_{n+j}\}\Rbag \cup \Dtest \backslash \{Z_{n+j}\}$. \red{If the partition was randomized independently of the data, we condition on its realized value as well.}
Conditional on $S_j = \Lbag z_1,\ldots,z_n,z_{n+j}\Rbag \cup  
\{z_{n+\ell}\}_{\ell \neq j}$,   we have
\begin{equation} \label{eq:weighted_exch_cond_dist}
\begin{aligned}
& (Z_1,\ldots,Z_n,Z_{n+j}, \{Z_{n+\ell}\}_{\ell \neq j}) \\
& \quad \stackrel{H_j}{\sim} \sum_{\pi \in \Sym([n] \cup \{n+j\})}
\frac{w(z_{\pi(n+j)})}{n! \cdot \sum_{i \in [n]\cup\{n+j\}} 
w(z_{\pi(i)}) } \cdot 
\delta_{(z_{\pi(1)},\ldots,z_{\pi(n)},z_{\pi(n+j)}, 
\{z_{n+\ell}\}_{\ell \neq j})}.
\end{aligned}
\end{equation}
\red{
    We expand upon weighted $K$-FC ND with e-BH-CC in Algorithm \ref{alg:weighted_kfc_nd_cc} (Appendix \ref{appd:algos}).
}

%% file: sections/experiments.tex
%

We compare \red{$K$-FC ND} to SC and AdaDetect (hereafter denoted AD) baselines through simulations and a real-data analysis. 
\red{We first simulate $K=1$, then showcase a model selection experiment 
using $K=m$.}
For sample-splitting competitors SC and AD, we use each split percentage $\rho \in \{25,50,75\}$, where $\rho$ denotes the percentage of $\Dref$ used for training. 
For a comparison of derandomized approaches, we use the derandomized SC procedure of \citet{bashari2024derandomized}, averaged over 20 random splits---these plots are deferred to
Appendix~\ref{appd:nonrandom_vs_derandomized}. Section~\ref{sec:realdata} contains the real-data analysis \red{with $m$-FC ND}. In the main text we show power plots only; the corresponding FDR plots are deferred to Appendix~\ref{appd:deferred_fdr}.

\subsection{Simulations for novelty detection}

\paragraph{Setting and hyperparameters.} To demonstrate the efficacy of $K$-FC ND, we borrow the simulation setup of \citet{bates2023testing, jin2023model, lee2024boosting}, which \red{generates outliers via a cluster-based process.} At the very beginning, we sample $n_W$ i.i.d.~draws from $\mathrm{Unif} ([-3,3]^{d_W})$ to attain an initial set of points in $\RR^{d_W}$; intuitively, these serve as our cluster centers. Denote this set as $\cW$; it will be fixed \red{across all replications within each simulation setting}. We then generate $Z_i$ for each $i \in [n+m]$ as follows. Sample $W_i \sim  \mathrm{Unif}(\cW)$ and $L_i \sim \cN_{d_W}(0, I)$. For each $i\in[n]$, construct $Z_i = L_i + W_i$. Finally, 
\begin{singlespace}
\begin{equation}
    Z_{n+j} =
    \begin{cases}
        L_{n+j} + W_{n+j} & \textnormal{if } j \in \cH_0; \\
        \sqrt a\, L_{n+j} + W_{n+j} & \textnormal{if } j \in \cH_1,
    \end{cases}
    \qquad\red{ \forall j\in [m].}
\end{equation}
\end{singlespace}
\noindent \red{Above, $a >  1$} is a hyperparameter governing the signal strength of the outliers. Hence, the inliers are standard multivariate Gaussian around the points in $\cW$, while the outliers are likewise multivariate Gaussian around these cluster centers but have larger spread. One can expect that for moderate to large values of $d_W$ and $n_W$, detecting outliers by understanding and estimating the inlier density (without any parametric knowledge) is difficult when the reference dataset is limited.

In our experiments for this setting, we fix $n_W = 30, d_W = 50, n = |\Dref| = 150$, take $m \in \{50,200\}$, vary the outlier proportion in $\Dtest$ as $\pi_1 \in \{0.2,0.4\}$, and target $\alpha = 0.1$.

\begin{figure}[!t]
    \centering
    \includegraphics[width=\textwidth]{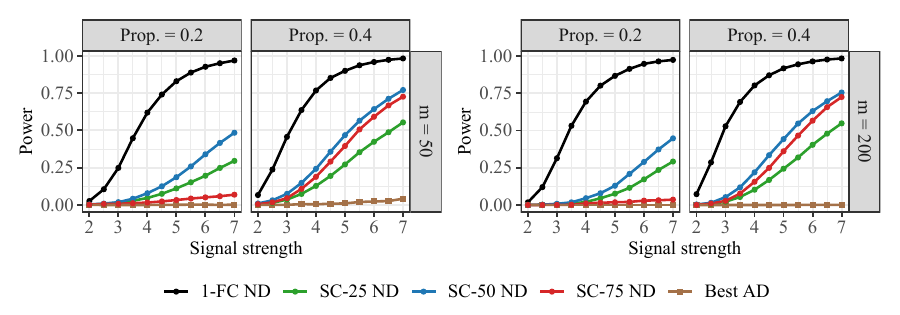}
    \caption{$1$-FC ND power results with Isolation Forest scores. We compare 1-FC ND, SC-$\rho$ ND, and the single ``Best AD'' baseline. The FDR target is $\alpha=0.1$. Each experiment uses 1,000 replications.}
    \label{draft:fig:n150_iso}
\end{figure}

\paragraph{Scoring models.}

For each approach, we consider two classes of scoring models: forest-based models and support vector-based models.
In particular, Isolation Forest~\citep{liu2012isolation} and the one-class SVM~\citep{scholkopf1999support}
are used as the scoring models in FC and SC ND; for AD, we use binary random forests
and support vector classifiers (SVC).
In all experiments, we use the Python package \texttt{scikit-learn} \citep{scikit-learn} for all model implementations.

The Isolation Forest experiments use 1-FC, while the SVM experiments use LOO-1-FC. In both cases we compare against SC-$\rho$ and AD-$\rho$ baselines. \red{When the distinction is immaterial, we use ``1-FC'' generically for $1$-FC ND, including its LOO variant.}

\subsubsection{$K=1$: 1-FC and LOO-1-FC}
\label{sec:exp_k1}

For $K=1$, we report the power of 1-FC ND, the three SC baselines, and a single ``Best AD'' curve, defined as the AD instance with the split proportion $\rho$ attaining the largest power over the signal grid within that experiment. As stated earlier, the figures comparing 1-FC ND to the derandomized SC methods are deferred to Appendix~\ref{appd:nonrandom_vs_derandomized}.

\begin{figure}[!t]
    \centering
    \includegraphics[width=\textwidth]{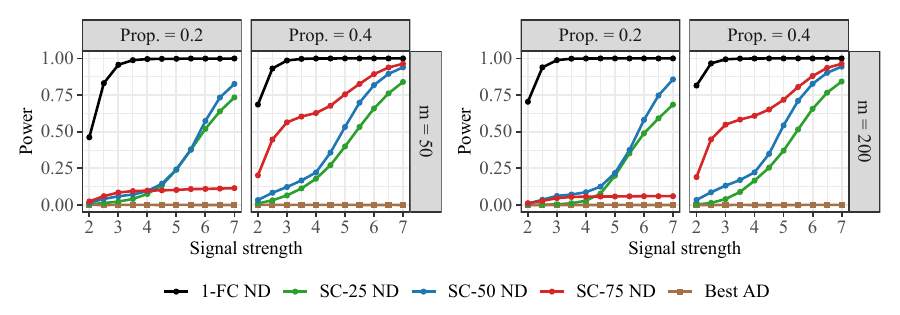}
    \caption{$1$-FC ND  power results with one-class SVM scores. Here the FC curve is the LOO-1-FC e-BH procedure. The FDR target is $\alpha=0.1$. Each experiment uses 1,000 replications.}
    \label{draft:fig:n150_svm}
\end{figure}

Figure~\ref{draft:fig:n150_iso} reports the Isolation Forest results, while Figure~\ref{draft:fig:n150_svm} reports the SVM results. In both cases the FC procedure is the most powerful competitor in every signal strength regime, with the gap especially pronounced for the sparser setting $\pi_1=0.2$. The SVM-based 1-FC ND performs especially well, which is within our expectations: the SVMs are fit on significantly more (high-dimensional) observations. 
We choose the model types for AD-$\rho$ to match those used for FC: random forests for Isolation Forest-based FC and SVC for the SVM-based FC. In all experiments, AD performs the poorest. \red{Note that 1-FC ND is 
equivalent  to using BH on the 1-FC p-values.}

\subsubsection{$K=m$: $m$-FC with model selection}
\label{sec:exp_ms}

For $K = m$---our other canonical choice of $K$---we will demonstrate its efficacy by implementing the model selection strategy outlined in   Section~\ref{sec:model_selection_mfc}.
We use the same data-generating mechanism and the same $(n,m,\pi_1,\alpha)$ choices as above. The candidate library contains ten models: 
four Isolation Forest models with numbers of trees in $\{25, 50, 100, 200\}$
and six one-class SVM models with $\nu\in\{0.004, 0.01, 0.1, 0.25\}$ (with \texttt{gamma} fixed to \texttt{scale} for the last two $\nu$ values and varying between \texttt{scale} and \texttt{auto} for the first two). 
These SVM candidates will again use the LOO framework. Following the ``top model ensembling'' strategy in Section~\ref{sec:model_selection_mfc}, we score the ten candidates by the proxy rejection criterion and ensemble the top three (after transforming the scores to a common scale).

\begin{figure}[tbp!]
    \centering
    \includegraphics[width=\textwidth]{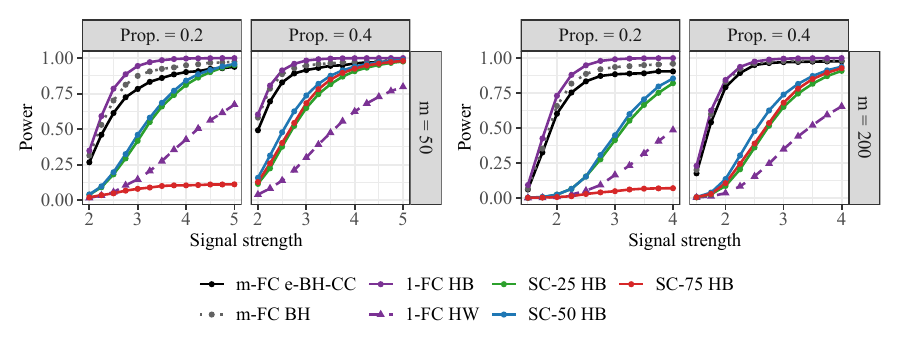}
    \caption{Power results for $m$-FC ND with model selection. $m$-FC e-BH-CC uses scores produced via ensembling the top models per test unit. HB and HW denote the hindsight-best and hindsight-worst single candidates within the same candidate library for that specific framework (per signal strength value $a$). 
    The FDR target is $\alpha=0.1$. Each experiment uses 1,000 replications.}
    \label{draft:fig:Km_mdlsel}
\end{figure}

In Figure~\ref{draft:fig:Km_mdlsel}, HB and HW denote the hindsight-best and hindsight-worst single candidate models within the same  ten models.
They are evaluated separately at each signal strength for each framework and should be interpreted as oracle benchmarks. \red{We compare our proposal, 
$m$-FC e-BH-CC, with the 1-FC HB baseline and 
BH on the $m$-FC p-values
(which has no FDR guarantee, as stated previously).
Overall, $m$-FC e-BH-CC closely tracks the performance of both 
across all settings, though the BH version slightly outperforms it
in the sparser setting (however, we show in 
Appendix~\ref{appd:mfc_bh_vs_ebhcc} that the opposite phenomenon can hold in 
ultra-sparse settings).}
$m$-FC e-BH-CC typically outperforms SC-$\rho$ HB across signal strengths.  We conclude that $m$-FC ND 
can adaptively learn which candidate ML models have the  highest detection 
power for the specific $\Dtest$. \red{All methods, guaranteed or not, exhibit 
empirical FDR control (Appendix~\ref{appd:deferred_fdr}).}

\subsection{Weighted novelty detection under distribution shift}
\label{sec:exp_weighted}

\red{We next introduce distribution shift between $\Dref$ and $\Dtest$, adapting the setup of \citet{jin2023model, lee2024boosting}. For $b\geq 1$, let $P_b$ denote the law of $\sqrt b\,L+W$ in the preceding construction and define $Q_b$ by
$\frac{\dif Q_b}{\dif P_b}(z)
    \propto
    {\sigma(z^\top\theta)} , 
$ where $\sigma(\cdot)$ is the standard logistic function.
We independently draw the reference inliers from $P_1$, the test inliers from $Q_1$, and the test outliers from $Q_a$, where $a > 1$.  
For each $r\in[d_W]$, we take $\theta_r = (4-r)/10 \cdot \indc{r \le 3}$ and keep the remaining setup from Section~\ref{sec:exp_k1} unchanged, except that $n_W = d_W = 30$. We report weighted 1-FC ND results for both Isolation Forest and SVM models and compare to the weighted SC e-BH-CC baselines implemented in~\citet{lee2024boosting}.}

\begin{figure}[!t]
    \centering
    \begin{subfigure}{\textwidth}
        \centering
        \includegraphics[width=0.92\textwidth]{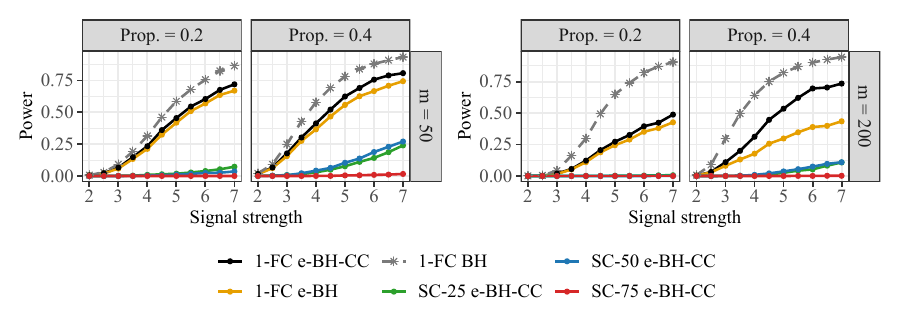}
        \caption{\red{Isolation Forest scores.}}
        \label{draft:fig:weighted_n150_if}
    \end{subfigure}
    \begin{subfigure}{\textwidth}
        \centering
        \includegraphics[width=0.92\textwidth]{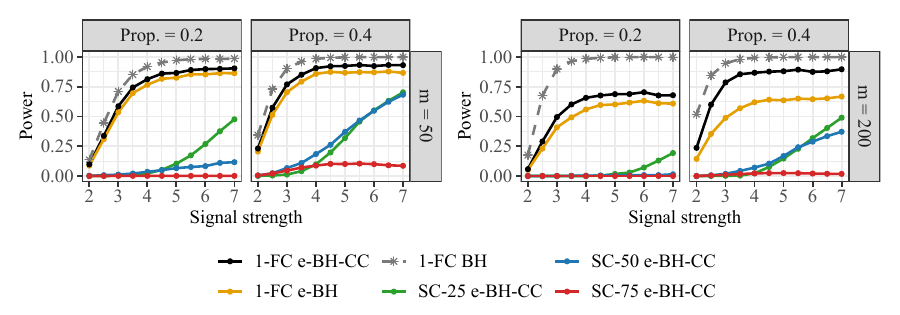}
        \caption{\red{One-class SVM scores.}}
        \label{draft:fig:weighted_n150_svm}
    \end{subfigure}
    \caption{\red{Weighted 1-FC ND power results under distribution shift.} We compare the \red{weighted variants of 1-FC e-BH, 1-FC BH, 1-FC e-BH-CC, and} the weighted SC-$\rho$ e-BH-CC baselines. The FDR target is $\alpha=0.1$. Each experiment uses 1,000 replications.}
    \label{draft:fig:weighted_n150}
\end{figure}

\red{Figure~\ref{draft:fig:weighted_n150} shows} a similar pattern to the unweighted experiments: our 1-FC e-BH-CC method performs substantially better than the weighted SC-$\rho$ e-BH-CC methods.
\red{The two panels} also illustrate the higher power provided by the e-BH-CC framework, which is especially significant for $\pi_1 = 0.4, m=200$. 
\red{We also compare with BH applied to 
weighted conformal p-values~\eqref{eq:weighted_kfc_pval}, which in certain settings outperforms our proposal but is not necessarily FDR-controlling~\citep{jin2023model}.}

%

\input{sections/experiments_real_data}

%% file: sections/experiments_real_data.tex
%

%
\subsection{Identification of malicious prompts in a curated dataset}
\label{sec:realdata}

In line with our LLM-focused motivations, we revisit the problem of  malicious prompt identification.~\citet{ayub2024embedding} use embedding models to pre-process benign and malicious LLM prompts and train random forest classifiers on these embeddings; the resulting fitted models can successfully predict new malicious prompts.

Using their provided labeled datasets of prompts and corresponding embeddings, we design a novelty detection setting in which we simulate the power of $m$-FC ND with model selection. For every replication, we draw $n=2,000$ benign prompts to be in $\Dref$ and $m=100$ prompts to be in $\Dtest$, with $80$ of them drawn from the benign prompts and 20 drawn from the malicious prompts (all uniformly randomly from their respective pools).

Our suite of models contains twelve candidates.~\citet{ayub2024embedding} supply embeddings from three different models (OpenAI, OctoAI, and MiniLM) and we compose the embedder with each of four Isolation Forest instances (50, 100, 200, and 400 trees), giving twelve possible scoring models.  Next, for the model selection $m$-FC procedure, we use the exact same implementation (top-3 model ensembling) as in  Section~\ref{sec:exp_ms}. For each target level $\alpha \in \{0.1,0.2,0.3,0.4\}$, we conduct 1,000 replications.

\begin{figure}[tbp!]
    \centering
    \includegraphics[width=\textwidth]{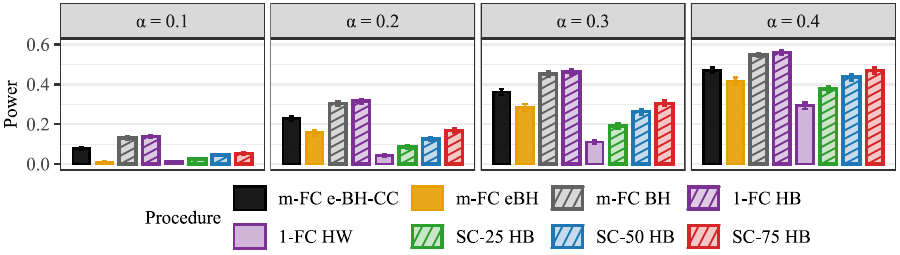}
    \caption{Power results for malicious prompt detection.  $m$-FC e-BH-CC uses scores produced from the same top-3 model selection rule as in Section~\ref{sec:exp_ms}. HB and HW denote the hindsight-best and hindsight-worst single candidates within the same candidate library, evaluated separately for each FC/SC framework and  FDR target $\alpha$. 
    \red{Hatching marks the $m$-FC BH bar and the hindsight-best oracle benchmark bars.}
    Each experiment uses 1,000 replications.}
    \label{draft:fig:real_data_power}
\end{figure}

Figure~\ref{draft:fig:real_data_power} shows that our (boosted) $m$-FC 
e-BH-CC approach performs close to \red{BH with $m$-FC p-values}
and  the hindsight-best candidate for 1-FC ND  
while 
\red{generally outperforming the 1-FC hindsight-worst candidate and the 
SC-$\rho$ hindsight-best baselines.} 

%% file: sections/discussion.tex
In this paper, we introduce a methodology for 
\red{nonparametric} novelty detection based on the full conformal 
inference framework.
\red{To enable } flexible, data-driven novelty detection 
\red{with} finite-sample FDR control, we construct conformal e-values
\red{rather than} conformal p-values.
We conclude with a discussion of future inquiries.

\paragraph{Using p-values versus e-values.}

The choice between p- and e-values is too nuanced to fully discuss in this paper.
\red{
For example, the e-BH-CC version of $K$-FC ND neither 
dominates nor is dominated by the BH approach 
(see Appendix~\ref{appd:ebhcc_power_over_bh} and~\ref{appd:mfc_bh_vs_ebhcc}).
}
For $K\ge2$, FDR control via BH is not ensured, \red{but can be 
reestablished by~\citet{fithian2020conditional} through their conditional calibration approach (albeit by re-introducing randomness). 
Generally, the design space for algorithms which use either conformal p- 
or e-values remains very open.}

\paragraph{Leveraging structure in $\Dtest$.} 

Although we focus mainly on the \red{practical recommendations $K\in\{1,m\}$}, our conclusions hold for any $K\in[m]$. For $K \notin \{1, m\}$, it is  unclear how to best split $\Dtest$ into $K$ blocks; \red{Theorem~\ref{thm:kfc_valid}, as stated, only
assumes the partition is independently, and potentially randomly, derived.}
However, one can imagine some underlying structure among the test units such that grouping them together would fit a higher-quality model. \red{The extension of $K$-FC ND to 
settings where such clustered structure exists and can be learned via data-driven 
techniques  could prove to be a fruitful research direction.}

%% file: supplement.tex
\spacingset{1}
\appendix

\renewcommand{\theequation}{S\arabic{equation}}
\renewcommand{\thefigure}{S\arabic{figure}}
\renewcommand{\thetable}{S\arabic{table}}
\renewcommand{\thealgocf}{S\arabic{algocf}}
\renewcommand{\thedefinition}{S\arabic{definition}}
\renewcommand{\thetheorem}{S\arabic{theorem}}
\renewcommand{\thelemma}{S\arabic{lemma}}
\renewcommand{\thetrick}{S\arabic{trick}}
\renewcommand{\theproperty}{S\arabic{property}}
\renewcommand{\theexample}{S\arabic{example}}
\renewcommand{\thecorollary}{S\arabic{corollary}}
\renewcommand{\theremark}{S\arabic{remark}}
\renewcommand{\theassumption}{S\arabic{assumption}}
\renewcommand{\theproposition}{S\arabic{proposition}}

\begin{center}
{\Large\bfseries Full conformal novelty detection}\\[2ex]
{\large\bfseries SUPPLEMENTARY MATERIAL}
\end{center}
\input{./sections/appendix}

\clearpage
\bibliographystyle{apalike}
\bibliography{ref}

%% file: sections/appendix.tex
\input{./sections/appendix_index}

%% file: sections/appendix_index.tex

\begingroup
\ifshowrevision\color{red!60}\fi
\section{Comparison of proposed methods}
\label{appd:method_comparison}
\input{./sections/appendix/method_comparison}
\clearpage
\endgroup

\section{Technical proofs}
\label{appd:proofs}
\input{./sections/appendix/technical_proofs}

\clearpage
\section{Theoretical improvements to e-BH-CC }
\label{appd:additional}
\input{./sections/appendix/ebhcc_results}

\clearpage
\section{Null proportion correction for higher power}
\label{sec:null_prop}
\input{./sections/appendix/null_proportion_correction}

\clearpage
\section{Theoretical examples of e-value-based improvements}
\label{appd:theory_examples}
\input{./sections/appendix/theoretical_examples_of_evalue_improvements}

\clearpage
\section{Omitted simulation details}
\label{appd:omitted_sim_details}
\input{./sections/appendix/omitted_simulation_details}

\clearpage
\section{Deferred simulations}
\label{appd:more_experiments}
\input{./sections/appendix/deferred_simulations}
\clearpage
\section{Algorithms}
\label{appd:algos}
\input{./sections/appendix/algorithms}

%% file: sections/appendix/method_comparison.tex
\subsubsection*{Methodological comparison}

\begin{table}[!ht]
\ifshowrevision\captionsetup{font={footnotesize,color=red!60}}\fi
\scriptsize
\centering
\setlength{\tabcolsep}{2pt}
\renewcommand{\arraystretch}{1.18}
\begin{tabular}{l|cccccc}
\hline
\multicolumn{1}{c|}{\rule[-2ex]{0pt}{6.4ex}\raisebox{0.4ex}{\makecell[c]{Method}}}
& \raisebox{0.4ex}{\makecell[c]{Finite-sample\\FDR control}}
& \raisebox{0.4ex}{\makecell[c]{Data-driven\\model selection}}
& \raisebox{0.4ex}{\makecell[c]{Non-random}}
& \raisebox{0.4ex}{\makecell[c]{FDR control extends to\\distribution shift}}
& \raisebox{0.4ex}{\makecell[c]{Allows for\\LOO scores}}
& \raisebox{0.4ex}{\makecell[c]{Boosted power\\via CC}}\\
\hline
$K$-FC e-BH
& \cmark & -- & -- & \cmark & \cmark & \cmark\\
\quad \rotatebox[origin=c]{180}{$\Lsh$} 1-FC e-BH
& \cmark & \xmark & \cmark & \cmark & \cmark & \cmark\\
\quad \rotatebox[origin=c]{180}{$\Lsh$} $m$-FC e-BH
& \cmark & \cmark & \cmark & \cmark & \cmark & \cmark\\
\hline
$K$-FC BH
& -- & -- & -- & \xmark & \cmark & \xmark\\
\quad \rotatebox[origin=c]{180}{$\Lsh$} 1-FC BH
& \cmark & \xmark & \cmark & \xmark & \cmark & \xmark\\
\quad \rotatebox[origin=c]{180}{$\Lsh$} $m$-FC BH
& \xmark & \cmark & \cmark & \xmark & \cmark & \xmark\\
\hline
\end{tabular}
\caption{A comparison among the proposed full conformal novelty detection
methods. Rows marked with \rotatebox[origin=c]{180}{$\Lsh$} 
are canonical special instantiations of the
preceding $K$-FC method family. 
A \cmark\, indicates that the property is enabled
(or, for FDR control, guaranteed), 
a \xmark\, indicates that it is not, and ``--''
indicates that it depends on $K$, the partition, or the choice 
of p-value versus e-value.
In particular, $K$-FC BH has finite-sample FDR control guarantees 
when $K=1$, but not in
general for $K\geq2$; data-driven model selection is available for $K>1$; and
the procedure is non-random for $K\in\{1,m\}$ or when an intermediate
partition is fixed deterministically.}
\label{tab:method_comparison}
\end{table}

\subsubsection*{Computational cost comparison}

\begin{table}[!ht]
\ifshowrevision\captionsetup{font={footnotesize,color=red!60}}\fi
\scriptsize
\centering
\setlength{\tabcolsep}{6pt}
\renewcommand{\arraystretch}{1.18}
\begin{tabular}{l|c}
\hline
\multicolumn{1}{c|}{\rule[-2ex]{0pt}{6.4ex}\raisebox{0.4ex}{\makecell[c]{Method}}}
& \raisebox{0.4ex}{\makecell[c]{Number of model fits\\(per learner)}}\\
\hline
$K$-FC ND & $K$\\
\quad \rotatebox[origin=c]{180}{$\Lsh$} 1-FC ND & $1$\\
\quad \rotatebox[origin=c]{180}{$\Lsh$} $m$-FC ND & $m$\\
\hline
\rule[-1.6ex]{0pt}{5.2ex}LOO-$K$-FC ND
& \makecell[c]{
$nK+m+K$, if $K>1$}\\
\quad \rotatebox[origin=c]{180}{$\Lsh$} LOO-1-FC ND & $n+m$\\
\quad \rotatebox[origin=c]{180}{$\Lsh$} LOO-$m$-FC ND & $m(n+2)$\\
\hline
\end{tabular}
\caption{The exact number of model-fitting operations per learner, 
corresponding to Algorithm~\ref{alg:kfc}. 
For LOO-$K$-FC ND with $K>1$, the additional $K$ model fits correspond to one
common model per block, reused to score the out-of-block test units. (When
$K=1$, there are no out-of-block test units, so this common model is not
trained.)}
\label{tab:method_cost}
\end{table}

%% file: sections/appendix/technical_proofs.tex
\subsection{Proof of Theorem~\ref{thm:kfc_valid}}
\label{appd:proof_eval_valid}

We first prove a more general score-level result, which contains both the
$K$-FC and LOO-$K$-FC constructions as special cases.

\begin{theorem}
\label{thm:kfc_valid_general}
Fix any $j\in\cH_0$ and $\tilde \alpha \in (0,1)$.
Let $\cW=\{W_i\}_{i\in[n+m]}$ be any collection of scores, and define
\begin{align}
    \begin{split}
        \label{eq:general_eval}
        e_j &= (n+1)\cdot \frac{\ind\{W_{n+j} \ge T\}}{1+\sum_{i\in[n]}\ind\{W_i \ge T\}},\\
        \textnormal{where } T
        &= \inf\left\{ t\in \cW \colon
        \frac{m}{ n+1} \cdot \frac{1 + \sum_{i=1}^{ n} \ind\{W_i \ge t\} }{1\vee \sum_{\ell=1}^m \ind\{W_{n+\ell} \ge t\}   } \le \tilde \alpha  \right\}.
    \end{split}
\end{align}
Assume that under $H_j$, there exists a collection of random  variables $S_j$ such that,
conditionally on $S_j$, the full score vector $\{W_i\}_{i\in[n+m]}$ is
invariant to permutations of $[n]\cup\{n+j\}$; namely, for any permutation
$\pi$ which acts as a bijection on $[n]\cup\{n+j\}$ and keeps fixed 
$\pi(n+\ell)=n+\ell$ for all $\ell\in [m]\setminus \{j\}$,
\[
(W_1,\ldots,W_{n+m}) \mid S_j
\stackrel{d}{=}
(W_{\pi(1)},\ldots,W_{\pi(n+m)}) \mid S_j.
\]
Then $\EE[e_j]\le 1$.
\end{theorem}

\begin{proof}
Define a hypothesis-specific stopping time
\[
\tilde T_j := \inf\bigg\{t\in\cW:
\frac{m}{n+1} \cdot \frac{\ind\{W_{n+j} \ge t\} + \sum_{i=1}^n \ind\{W_i \ge t\}
}{1 + \sum_{\ell \neq j} \ind\{W_{n+\ell} \ge t\}} \le \tilde \alpha \bigg\}.
\]
By construction, $\tilde T_j$ is invariant to the permutation of
$W_1,\ldots,W_n,W_{n+j}$ and $\tilde T_j \le T$.
Next, on the event $\{W_{n+j} \ge T\}$, we have that
$W_{n+j} \ge \tilde T_j$ and
\[
\frac{m}{n+1} \cdot \frac{1 + \sum_{i=1}^n \ind\{W_i \ge \tilde T_j\}}
{\sum_{\ell = 1}^m \ind\{W_{n+\ell} \ge \tilde T_j\}}
=
\frac{m}{n+1} \cdot \frac{\ind\{W_{n+j} \ge \tilde T_j\} + \sum_{i=1}^n \ind\{W_i \ge \tilde T_j\}}
{1 + \sum_{\ell \neq j} \ind\{W_{n+\ell} \ge \tilde T_j\}} \le \tilde \alpha,
\]
with the last step following from the definition of $\tilde T_j$.
The above also implies that $T \le \tilde T_j$. Consequently, we have
$T = \tilde T_j$ on the event $\{W_{n+j} \ge T\}$.
With this fact, we proceed to show that
\begin{align}
    \begin{split}
        \label{eq:general_thm_proof}
\EE[e_j]
&= (n+1)\cdot \EE\bigg[\frac{\ind\{W_{n+j} \ge T\}}{1+\sum_{i=1}^n \ind\{W_i \ge T\}}\bigg]\\
&\le (n+1)\cdot \EE\bigg[\frac{\ind\{W_{n+j} \ge \tilde T_j\}}{1+\sum_{i=1}^n \ind\{W_i \ge \tilde T_j\}}\bigg]\\
&= (n+1)\cdot \EE\bigg[\frac{\ind\{W_{n+j} \ge \tilde T_j\}}{\big(\ind\{W_{n+j} \ge \tilde T_j\}+\sum_{i=1}^n \ind\{W_i \ge \tilde T_j\} \big) \vee 1}\bigg].
    \end{split}
\end{align}
By assumption, under $H_j$, conditional on $S_j$, the score vector
$\{W_i\}_{i\in[n+m]}$ is invariant to the permutation of
$[n]\cup\{n+j\}$. Since $\tilde T_j$ is also invariant to the permutation
of $W_1,\ldots,W_n,W_{n+j}$, we have
\[
\begin{aligned}
&(n+1)\cdot \EE\Bigg[\frac{\ind\{W_{n+j} \ge \tilde T_j\}}
{\big(\ind\{W_{n+j} \ge \tilde T_j\}
+\sum_{\ell=1}^{n}
\ind\{W_\ell \ge \tilde T_j\}\big)\vee 1}\Bigg]\\
= ~& \EE\Bigg[\sum_{i \in [n] \cup \{n+j\}}
\EE\Bigg[\frac{\ind\{W_{n+j} \ge \tilde T_j\}}
{\red{\big(\ind\{W_{n+j} \ge \tilde T_j\}
+\sum_{\ell\in ([n]\cup\{n+j\})\setminus\{n+j\}}
\ind\{W_\ell \ge \tilde T_j\}\big)\vee 1}}
\Biggiven S_j \Bigg]\Bigg]\\
\stepa{=} ~& \EE\Bigg[\sum_{i \in [n] \cup \{n+j\}}
\EE\Bigg[\frac{\ind\{W_i \ge \tilde T_j\}}
{\red{\big(\ind\{W_i \ge \tilde T_j\}
+\sum_{\ell\in ([n]\cup\{n+j\})\setminus\{i\}}
\ind\{W_\ell \ge \tilde T_j\}\big)\vee 1}}
\Biggiven S_j \Bigg]\Bigg]\\
= ~& \EE\Bigg[
\frac{\sum_{i\in[n]\cup\{n+j\}}\ind\{W_i\ge\tilde T_j\}}
{\red{\big(\sum_{\ell\in[n]\cup\{n+j\}}\ind\{W_\ell\ge\tilde T_j\}\big)\vee1}}
\Bigg]\le 1.
\end{aligned}
\]
\red{
Step (a) is due to the fact that, conditional on
$S_j$, swapping $W_i$ and
$W_{n+j}$ does not change the distribution of the ratio 
within the expectation in the topmost line above. 
The bound then follows from the fact that 
the denominator of the ratio is
invariant to the swapping.}
Combining this with \eqref{eq:general_thm_proof} proves the claim.
\end{proof}

\begin{proof}[Proof of Theorem~\ref{thm:kfc_valid}]
\red{
If the partition $\cup_{k=1}^K B_k$ is randomized independently of the data, we first condition on its realized value and average over this randomness at the end. 
}
Fix any $j\in\cH_0$, and let $Z_{n+j}\in B_k$.
Take $S_j := \{Z_{n+\ell}\}_{\ell\in[m]\backslash\{j\}}$ and
$W_i := V_i^{(k)}$ for $i\in[n+m]$.
Under $H_j$, the distribution of $(Z_1,\ldots,Z_n,Z_{n+j})$ is
exchangeable conditional on $S_j$.
Moreover, by the definition of the $K$-FC construction,
$V^{(k)}(\cdot)$ is trained in a way that is invariant to the ordering of
the samples in $\Dref\cup B_k$.
Therefore, conditional on $S_j$, the full score vector
$\{V_i^{(k)}\}_{i\in[n+m]}$ is invariant to permutations of
$[n]\cup\{n+j\}$.
Hence the condition of Theorem~\ref{thm:kfc_valid_general} is satisfied,
which implies that $\EE[e_j]\le 1$.
The FDR control follows directly from the property of the e-BH
procedure~\citep{wang2022false}.
\end{proof}

\begin{corollary}
\label{cor:loo_kfc_valid}
The same conclusion as in Theorem~\ref{thm:kfc_valid} holds for the
LOO-$K$-FC e-values obtained by using in \eqref{eq:kfc_eval} the scores
$V_i^{(k)} = V^{(k,-i)}(Z_i)$ for $Z_i\in\Dref\cup B_k$ and
$V_i^{(k)} = V^{(k)}(Z_i)$ for $Z_i\in\Dtest\setminus B_k$.
\end{corollary}

\begin{proof}
Fix any $j\in\cH_0$, and let $Z_{n+j}\in B_k$.
Take $S_j := \{Z_{n+\ell}\}_{\ell\in[m]\backslash\{j\}}$ and
write $C_k:=\Dref\cup B_k$. Define
\[
W_i :=
\begin{cases}
V^{(k,-i)}(Z_i), & Z_i\in C_k,\\
V^{(k)}(Z_i), & Z_i\in\Dtest\setminus B_k.
\end{cases}
\]
Under $H_j$, the distribution of $(Z_1,\ldots,Z_n,Z_{n+j})$ is
exchangeable conditional on $S_j$.
For each $i\in[n]\cup\{n+j\}$, the scorer $V^{(k,-i)}(\cdot)$ is trained
in a way that is invariant to the ordering of
$C_k\backslash\{Z_i\}$.
Hence, conditional on $S_j$, permuting the indices in
$[n]\cup\{n+j\}$ sends the score $W_i$ to the score of the permuted
index.
For $\ell\neq j$, the score $W_{n+\ell}$ is unchanged under the same
permutation: its evaluation point is fixed, and its training bag is either
$C_k\backslash\{Z_{n+\ell}\}$ when $Z_{n+\ell}\in B_k$ or $C_k$ when
$Z_{n+\ell}\notin B_k$, both of which remain unchanged as unordered bags.
Therefore, conditional on $S_j$, the full score vector
$\{W_i\}_{i\in[n+m]}$ is invariant to permutations of
$[n]\cup\{n+j\}$.
The condition of Theorem~\ref{thm:kfc_valid_general} is thus satisfied,
and the result follows.
\end{proof}

\subsection{Proof of Proposition~\ref{prop:ms}}
\label{appd:proof_ms_validity}
\begin{proof}
Fix any $j\in\cH_0$, and abbreviate
$V_i^{\mathrm{MS}} := V_i^{\mathrm{MS},(j)}$ for $i\in[n+m]$.
Let $T_j^{\mathrm{MS}}$ denote the threshold in \eqref{eq:kfc_eval}
computed from the scores $\{V_i^{\mathrm{MS}}\}_{i\in[n+m]}$.
Define
\[
\tilde T_j^{\mathrm{MS}}
:=
\inf\bigg\{
t\in \{V_i^{\mathrm{MS}}\}_{i\in[n+m]}:
\frac{m}{n+1}\cdot
\frac{\ind\{V_{n+j}^{\mathrm{MS}} \ge t\}+\sum_{i=1}^n \ind\{V_i^{\mathrm{MS}} \ge t\}}
{1+\sum_{r \neq j}\ind\{V_{n+r}^{\mathrm{MS}} \ge t\}}
\le \tilde \alpha
\bigg\}.
\]
By construction, $\tilde T_j^{\mathrm{MS}} \le T_j^{\mathrm{MS}}$.
Next, on the event $\{V_{n+j}^{\mathrm{MS}} \ge T_j^{\mathrm{MS}}\}$, we have
$V_{n+j}^{\mathrm{MS}} \ge \tilde T_j^{\mathrm{MS}}$ and
\[
\frac{m}{n+1} \cdot
\frac{1 + \sum_{i=1}^n \ind\{V_i^{\mathrm{MS}} \ge \tilde T_j^{\mathrm{MS}}\}}
{\sum_{r=1}^m \ind\{V_{n+r}^{\mathrm{MS}} \ge \tilde T_j^{\mathrm{MS}}\}}
=
\frac{m}{n+1} \cdot
\frac{\ind\{V_{n+j}^{\mathrm{MS}} \ge \tilde T_j^{\mathrm{MS}}\}
+ \sum_{i=1}^n \ind\{V_i^{\mathrm{MS}} \ge \tilde T_j^{\mathrm{MS}}\}}
{1 + \sum_{r\neq j} \ind\{V_{n+r}^{\mathrm{MS}} \ge \tilde T_j^{\mathrm{MS}}\}}
\le \tilde \alpha,
\]
with the last step following from the definition of
$\tilde T_j^{\mathrm{MS}}$.
The above also implies that $T_j^{\mathrm{MS}} \le \tilde T_j^{\mathrm{MS}}$.
Consequently, we have
$T_j^{\mathrm{MS}} = \tilde T_j^{\mathrm{MS}}$
on the event $\{V_{n+j}^{\mathrm{MS}} \ge T_j^{\mathrm{MS}}\}$.
Hence,
\begin{align}
    \begin{split}
\EE[e_j]
&= (n+1)\cdot \EE\bigg[
\frac{\ind\{V_{n+j}^{\mathrm{MS}} \ge T_j^{\mathrm{MS}}\}}
{1+\sum_{i=1}^n \ind\{V_i^{\mathrm{MS}} \ge T_j^{\mathrm{MS}}\}}
\bigg]\\
&\le (n+1)\cdot \EE\bigg[
\frac{\ind\{V_{n+j}^{\mathrm{MS}} \ge \tilde T_j^{\mathrm{MS}}\}}
{1+\sum_{i=1}^n \ind\{V_i^{\mathrm{MS}} \ge \tilde T_j^{\mathrm{MS}}\}}
\bigg]\\
&= (n+1)\cdot \EE\bigg[
\frac{\ind\{V_{n+j}^{\mathrm{MS}} \ge \tilde T_j^{\mathrm{MS}}\}}
{\big(\ind\{V_{n+j}^{\mathrm{MS}} \ge \tilde T_j^{\mathrm{MS}}\}
+\sum_{i=1}^n \ind\{V_i^{\mathrm{MS}} \ge \tilde T_j^{\mathrm{MS}}\}\big)\vee 1}
\bigg].
    \end{split}
    \label{eq:prop_ms_proof}
\end{align}

Now condition on $S_j$.
Since $\Gamma_1,\dots,\Gamma_L$ are $\sigma(S_j)$-measurable, conditional on $S_j$
the map
\[
h_j(x_1,\dots,x_L)
:=
g\big((x_1,\Gamma_1),\dots,(x_L,\Gamma_L)\big)
\]
is deterministic (here, we implicitly assume measurability of the function $g(\cdot)$).
By assumption (i), the collection
$\{(V_i^{(j,1)},\dots,V_i^{(j,L)})\}_{i\in[n]\cup\{n+j\}}$
is jointly exchangeable conditional on $S_j$.
Therefore,
\[
\{V_i^{\mathrm{MS}}\}_{i\in[n]\cup\{n+j\}}
=
\{h_j(V_i^{(j,1)},\dots,V_i^{(j,L)})\}_{i\in[n]\cup\{n+j\}}
\]
is also jointly exchangeable conditional on $S_j$.
Moreover, by assumption (ii) and the $\sigma(S_j)$-measurability of
$\Gamma_1,\dots,\Gamma_L$, the scores
$\{V_{n+r}^{\mathrm{MS}}\}_{r\neq j}$ are $\sigma(S_j)$-measurable.
It follows that $\tilde T_j^{\mathrm{MS}}$ is invariant to the permutation of
$V_1^{\mathrm{MS}},\dots,V_n^{\mathrm{MS}},V_{n+j}^{\mathrm{MS}}$
conditional on $S_j$.

We may therefore repeat the exchangeability argument from the proof of
Theorem~\ref{thm:kfc_valid}:
\[
\begin{aligned}
&(n+1)\cdot \EE\bigg[
\frac{\ind\{V^{\mathrm{MS}}_{n+j} \ge \tilde T^{\mathrm{MS}}_j\}}
{\big(\ind\{V^{\mathrm{MS}}_{n+j} \ge \tilde T^{\mathrm{MS}}_j\}
+\sum_{i=1}^n \ind\{V^{\mathrm{MS}}_i \ge \tilde T^{\mathrm{MS}}_j\}\big)\vee 1}
\bigg]\\
=~&
\EE\Bigg[
\sum_{i \in [n] \cup \{n+j\}}
\EE\bigg[
\frac{\ind\{V^{\mathrm{MS}}_{n+j} \ge \tilde T^{\mathrm{MS}}_j\}}
{\red{\big(\ind\{V^{\mathrm{MS}}_{n+j} \ge \tilde T^{\mathrm{MS}}_j\}
+\sum_{\ell\in([n]\cup\{n+j\})\setminus\{n+j\}}
\ind\{V^{\mathrm{MS}}_\ell \ge \tilde T^{\mathrm{MS}}_j\}\big)\vee 1}}
\bigggiven S_j \bigg]
\Bigg]\\
\stepa{=}~&
\EE\Bigg[
\sum_{i \in [n] \cup \{n+j\}}
\EE\bigg[
\frac{\ind\{V^{\mathrm{MS}}_{i} \ge \tilde T^{\mathrm{MS}}_j\}}
{\red{\big(\ind\{V^{\mathrm{MS}}_{i} \ge \tilde T^{\mathrm{MS}}_j\}
+\sum_{\ell\in([n]\cup\{n+j\})\setminus\{i\}}
\ind\{V^{\mathrm{MS}}_\ell \ge \tilde T^{\mathrm{MS}}_j\}\big)\vee 1}}
\bigggiven S_j \bigg]
\Bigg]\\
=~&
\EE\Bigg[
\frac{\sum_{i\in[n]\cup\{n+j\}}\ind\{V^{\mathrm{MS}}_i\ge\tilde T_j^{\mathrm{MS}}\}}
{\red{\big(\sum_{\ell\in[n]\cup\{n+j\}}\ind\{V^{\mathrm{MS}}_\ell\ge\tilde T_j^{\mathrm{MS}}\}\big)\vee1}}
\Bigg]
\le 1,
\end{aligned}
\]
where step (a) is because, conditional on $S_j$, swapping
$V_i^{\mathrm{MS}}$ and $V_{n+j}^{\mathrm{MS}}$ does not change the
distribution.
Combining the above with \eqref{eq:prop_ms_proof}, we conclude that
$\EE[e_j]\le 1$. Therefore, $e_j$ is a valid e-value.
\end{proof}

\subsection{Proof of Proposition~\ref{prop:boost_eval}}
\label{appd:proof_boost_eval}
\red{By Lemma~\ref{lem:equiv_comp}, the boosted e-values defined
in~\eqref{eq:boosted_eval} agree with the general critical-value construction
in~\eqref{eq:boosted_eval_general}. We therefore prove the result using the
general construction.}
Fix any $j\in \cH_0$. If $\phi_j(\hat c_j;S_j)\le 0$, then
by the definition of $e_j^\boost$, 
\$ 
\EE[e_j^\boost \given S_j] = 
\EE\bigg[\frac{m \ind\{e_j \ge \frac{m}{\alpha |\cR \cup \{j\}|} 
\text{ or } j \in \cA(\hat c_j)\}}{\alpha |\cR \cup \{j\}|}
\Biggiven S_j\bigg] = \phi_j(\hat c_j;S_j) + \EE[e_j \given S_j]
\le \EE[e_j\given S_j].
\$
Taking expectation over $S_j$ on both sides leads to $\EE[e_j^\boost] \le 
\EE[e_j] \le 1$.

If, on the other hand, $\phi_j(\hat c_j;S_j) >0$, then 
\$ 
\EE[e_j^\boost \given S_j] 
& = \EE\bigg[\lim_{\ell \rightarrow \infty}
\frac{m \ind\{e_j \ge \frac{m}{\alpha |\cR \cup \{j\}|} 
\text{ or } j \in \cA(\hat c_{j,\ell})\}}{\alpha |\cR \cup \{j\}|}
\Biggiven S_j\bigg] \\
& = \lim_{\ell \rightarrow \infty}\EE\bigg[
\frac{m \ind\{e_j \ge \frac{m}{\alpha |\cR \cup \{j\}|} 
\text{ or } j \in \cA(\hat c_{j,\ell})\}}{\alpha |\cR \cup \{j\}|}
\Biggiven S_j\bigg] \\
& = \lim_{\ell \rightarrow \infty} \phi_j(\hat c_{j,\ell};S_j) + 
\EE[e_j\given S_j] \le \EE[e_j \given S_j],
\$
where the second equality is by the monotone convergence theorem,
and the last inequality is by the choice of $\hat c_{j}$.
Again taking expectation over $S_j$, we have in this case that 
$\EE[e_j^\boost] \le \EE[e_j] \le 1$.

To see that $\cR^\ebh_\alpha(e_1,\ldots,e_m) \subseteq 
\cR^\ebh_\alpha(e_1^\boost,\ldots,e_m^\boost)$, we leverage the fact 
shown in~\citet{lee2024boosting} that 
$\cR^{\ebh}_\alpha(e_1,\ldots,e_m) = \cR^\ebh_\alpha(\tilde e_1,
\ldots,\tilde e_m)$, where 
\$ 
\tilde e_j = \frac{m \ind\{e_j \ge \frac{m}{\alpha |\cR \cup \{j\}|} 
\}}{\alpha |\cR \cup \{j\}|}, ~\forall j\in [m].
\$
Since $\tilde e_j \le e_j^\boost$ deterministically, we complete the proof. 

\subsection{Proof of Theorem \ref{thm:weighted_eval_valid}}
\label{appd:proof_of_weighted_eval}

\begin{proof}
\red{
    If the partition $\cup_{k=1}^K B_k$ is randomized independently of the data, we first condition on its realized value and average over this randomness at the end. We then fix $j$ and the block $k$ to which $Z_{n+j}$ belongs.} Define the modified threshold 
    \begin{multline}\label{eq:weighted_full_conf_modified_threshold}
        \hat T_j = \inf\Bigg\{ t \in  \cV^{(k)} \colon    \frac{m}{ w_{n+j} + \sum_{i=1}^n w_i } \cdot \frac{w_{n+j}\indc{V_{n+j}^{(k)} \ge t} + \sum_{i=1}^n w_i\indc{V_i^{(k)} \ge t}}{1 + \sum_{\ell \in [m]\setminus\{j\}} \indc{V_{n+\ell}^{(k)} \ge t}} \le \tilde\alpha  \Bigg\}.
    \end{multline}
    Similar to the proof of Theorem \ref{thm:kfc_valid}, we can use the fact that on the event $\{V_{n+j}^{(k)} \ge T_j\}$, we have $T_j = \hat T_j$.
    Let $\cE_j = \Lbag \{Z_1, \dots, Z_n, Z_{n+j}\}\Rbag$ be the unordered elements of $\Dref\cup \{Z_{n+j}\}$, and define
    \begin{align*}
        \cF_j &= \sigma\left(\cE_j,\{Z_{n+\ell}\}_{\ell\in[m]\setminus\{j\}}\right),\\
        W &= w_{n+j}+\sum_{i=1}^n w_i,\\
        D &= w_{n+j}\indc{V_{n+j}^{(k)}\geq \hat T_j}+\sum_{i=1}^n w_i\indc{V_i^{(k)}\geq \hat T_j}.
    \end{align*}
    Now, under $H_j$, write
    \begin{align}
        \begin{split}
            \EE[e_j] &=  \EE\left[ \frac{ \big( w_{n+j} + \sum_{i=1}^n w_i \big)\indc{V_{n+j}^{(k)} \ge T_j}}{w_{n+j}  + \sum_{i=1}^n w_i \indc{V_i^{(k)} \ge T_j}} \right]\\
             &=  \EE\left[ \frac{ \big( w_{n+j} + \sum_{i=1}^n w_i \big)\indc{V_{n+j}^{(k)} \ge T_j}}{w_{n+j}\indc{V_{n+j}^{(k)} \ge T_j}  + \sum_{i=1}^n w_i \indc{V_i^{(k)} \ge T_j}} \right] \quad (\textnormal{using the notation of }0/0=0)\\
             &\stepa{\le}  \EE\left[\indc{D>0}\frac{W\indc{V_{n+j}^{(k)} \ge \hat T_j}}{D} \right] \\ 
            &\stepb{=} \EE\left[\indc{D>0}\frac{W}{D}
                \PP\left(V_{n+j}^{(k)}\geq\hat T_j\given\cF_j\right)\right]
        \end{split}
    \end{align}
    
    Step (a) is allowed due to considering the expectation on and off the event $\{V_{n+j}^{(k)} \ge T_j\}$, similar to the analogous step in Appendix \ref{appd:proof_eval_valid}, equation \eqref{eq:general_thm_proof}.
    Regarding step (b), the contribution on $\{D=0\}$ is zero under the convention $0/0=0$. On $\{D>0\}$, the tower property applies because $W$, $D$, and $\hat T_j$ are $\cF_j$-measurable.

    Hence, we must evaluate
    $$
        \EE\left[ \indc{V_{n+j}^{(k)} \ge \hat T_j}  \given \cF_j \right] = \PP\big(V_{n+j}^{(k)} \ge \hat T_j \given \cF_j\big).
    $$
    For $z=\Lbag z_1,\dots,z_n,z_{n+j}\Rbag$, define $v_\ell=V^{(k)}(z_\ell)$. We then have
    \begin{equation}\label{eq:weighted_unif_scores}
        V_{n+j}^{(k)} \given \{ \cE_j = z\}, \{Z_{n+\ell}\}_{\ell \neq j} \sim \sum_{\ell  \in [n] \cup \{n+j\}} \frac{w(z_\ell)}{w(z_{n+j}) + \sum_{i=1}^n w(z_i) }\delta_{v_\ell}
    \end{equation}
    because 1) the elements in $\cE_j$ are assumed to be weighted exchangeable and 2) the scoring function $V^{(k)}(\cdot)$ was trained invariant to the order of elements inside $\cE_j$.
    Thus, we can conclude that
    $$
        \PP(V_{n+j}^{(k)} \ge t \given \cF_j ) = \sum_{\ell \in [n] \cup \{n+j\}} \frac{w_\ell\indc{V_\ell^{(k)} \ge t}}{W}
    $$
    which directly implies that
    \begin{equation}
        \PP( V_{n+j}^{(k)} \ge \hat T_j \given \cF_j ) = \frac{D}{W}.
    \end{equation}
    Here we used that $\hat T_j$ is $\cF_j$-measurable. Hence, substituting this back into our computation above gives
    \[
        \EE\left[\indc{D>0}\frac{W}{D}
        \PP\left(V_{n+j}^{(k)}\geq\hat T_j\given\cF_j\right)\right]
        =\EE\left[\indc{D>0}\right]\leq 1,
    \]
    so $\EE[e_j] \le 1$ under $H_j$.

\end{proof}

\subsection{Proof of Proposition~\ref{prop:nullest}}
\label{appd:proof_nullest}
If $\cH_0=\varnothing$, the claim is immediate. Otherwise, fix
$j \in \cH_0$ and let $B_k$ denote the block to which $Z_{n+j}$ belongs.
We use the convention $\inf\varnothing=+\infty$. By construction, we have
\$
\EE\bigg[\frac{e_j}{\hat \pi_0^{(j)}}\bigg]
= (n+1) \cdot \EE\bigg[\frac{\ind\{V_{n+j}^{(k)} \ge T_k\}}
{1 + \sum_{i\in [n]}\ind\{V_{i}^{(k)} \ge T_k\}} 
\cdot \frac{1}{\hat \pi_0^{(j)}}\bigg]. 
\$
Define an auxiliary threshold:
\$
\hat T_j = \inf\bigg\{t \in \cV^{(k)}:
\frac{m}{n+1}\cdot
\frac{\ind\{V_{n+j}^{(k)}\ge t\}
+\sum_{i\in[n]} \ind\{V_i^{(k)} \ge t\}}{1 + \sum_{\ell \in [m]\backslash \{j\}} \ind\{V_{n+\ell}^{(k)} \ge t\}} \le \tilde \alpha\bigg\}.
\$
Per their definitions, it is straightforward to see that 
$\hat{T}_j \le T_k$. On the event $\{V_{n+j}^{(k)} \ge T_k\}$,
we also have $V_{n+j}^{(k)} \ge \hat T_j$ and
\$
\frac{m}{n+1}\cdot
\frac{1 + \sum_{i\in[n]} \ind\{V_i^{(k)} \ge \hat T_j\}}
{\sum_{\ell \in [m]} \ind\{V_{n+\ell}^{(k)} \ge \hat T_j\}}
= \frac{m}{n+1}\cdot
\frac{\ind\{V_{n+j}^{(k)} \ge \hat T_j\} + \sum_{i\in[n]} \ind\{V_i^{(k)} \ge \hat T_j\}}
{1 + \sum_{\ell \in [m]\backslash\{j\}} \ind\{V_{n+\ell}^{(k)} \ge \hat T_j\}} \le \tilde \alpha,
\$
where the last inequality follows from the definition of $\hat T_j$.
Recalling the definition of $T_k$, we have $T_k \le \hat T_j$.
Putting everything together, we have $T_k = \hat T_j$ on the
event $\{V_{n+j}^{(k)} \ge T_k\}$. Then,
\$
\EE\bigg[\frac{e_j}{\hat \pi_0^{(j)}}\bigg] 
& \le (n+1) \cdot \EE\bigg[\frac{\ind\{V_{n+j}^{(k)} \ge \hat T_j\}}{1 + \sum_{i\in [n]}\ind\{V_{i}^{(k)} \ge \hat T_j\}}
\cdot \frac{1}{\hat \pi_0^{(j)}}\bigg]\\
& = (n+1) \cdot \EE\bigg[\frac{\ind\{V_{n+j}^{(k)} \ge \hat T_j\}}
{1 \vee \sum_{i\in [n] \cup \{n+j\}}\ind\{V_{i}^{(k)} \ge \hat T_j\}}
\cdot \frac{1}{\hat \pi_0^{(j)}}\bigg].
\$
Note that $\hat T_j$ and $\hat \pi_0^{(j)}$ are both invariant 
to permutations of $\Dref\cup \{Z_{n+j}\}$, and that 
$Z_1,\ldots,Z_n,Z_{n+j}$ are exchangeable given
$\{Z_{n+\ell}\}_{\ell\in[m]\backslash\{j\}}$ (since $j\in \cH_0$).
Averaging the final expectation in the preceding display over these
$n+1$ exchangeable indices replaces its numerator by $1/(n+1)$ times
the sum of all corresponding indicators. This sum divided by its maximum
with $1$ is at most $1$. Therefore,
$\EE[e_j /\hat \pi_0^{(j)}]\le \EE[1/\hat \pi_0^{(j)}]$.
We proceed to bound the right-hand side:
\$
\EE\bigg[\frac{1}{\hat \pi_0^{(j)}}\bigg]
& = \frac{m}{n+1}\EE\Big[\frac{\sum_{i \in [n]\cup \{n+j\}}
\ind\{S_{i} \le S_{(\tau_j)}\}}{1 + \sum_{\ell \in [m]\backslash \{j\}} 
\ind\{S_{n+\ell} \le S_{(\tau_j)}\}}\Big]\\
& = \frac{m}{n+1} \sum_{i \in [n]\cup \{n+j\}} 
\EE\Big[\frac{
\ind\{S_{i} \le S_{(\tau_j)}\}}{1 + \sum_{\ell \in [m]\backslash \{j\}} 
\ind\{S_{n+\ell} \le S_{(\tau_j)}\}}\Big]\\
& \le \frac{m}{n+1} \sum_{i \in [n]\cup \{n+j\}} 
\EE\Big[\frac{
\ind\{S_{i} \le S_{(\tau_j)}\}}{1 + \sum_{\ell \in \cH_0\backslash \{j\}} 
\ind\{S_{n+\ell} \le S_{(\tau_j)}\}}\Big].
\$
For any $i \in [n]\cup \{n+j\}$, 
\@\label{eq:nullest}
\EE\Big[\frac{
\ind\{S_{i} \le S_{(\tau_j)}\}}{1 + \sum_{\ell \in \cH_0\backslash \{j\}} 
\ind\{S_{n+\ell} \le S_{(\tau_j)}\}}\Big]
= \EE\Big[\frac{
\ind\{S_{i} \le S_{(\tau_j)}\}}{1 \vee(\ind\{S_i \le S_{(\tau_j)}\} + \sum_{\ell \in \cH_0\backslash \{j\}} 
\ind\{S_{n+\ell} \le S_{(\tau_j)}\})}\Big].
\@  
Recall that $\tau_j$ is a stopping time with respect to the 
backwards filtration $\{\cG_k\}_{k\in[n+m]}$. Therefore, for any $k \in [n+m]$, 
we can find some function $f_k$ such that 
\$
& \ind\{\tau_j \ge k\} = f_k(\{A_\ell\}_{\ell \ge k}, \{B_\ell\}_{\ell \ge k}), \text{ where}\\
&A_\ell = 1+ \sum_{\ell' \in [m]\backslash\{j\}}
\ind\{S_{n+\ell'} \le S_{(\ell)}\} 
\text{ and }B_\ell = \ind\{S_{n+j} \le S_{(\ell)}\}
+ \sum_{\ell'\in[n] } \ind\{S_{\ell'} \le S_{(\ell)}\}.
\$
To proceed, define
\$
&\tilde A^{i,j}_\ell = \ind\{S_i \le S_{(\ell)}\} + \sum_{\ell' \in [m]\backslash \{j\}} 
\ind\{S_{n+\ell'} \le S_{(\ell)}\} 
\text{ and }\tilde B^{i,j}_\ell = 1
+ \sum_{r\in([n]\cup\{n+j\})\backslash\{i\}}
\ind\{S_r \le S_{(\ell)}\},\\
&\tilde\tau_{j,i}:=
\max\!\left(\{1\}\cup
\left\{k\in[n+m]:
f_k(\{\tilde A^{i,j}_\ell\}_{\ell\ge k},
\{\tilde B^{i,j}_\ell\}_{\ell\ge k})=1\right\}\right).
\$
We shall show that $\tau_j = \tilde \tau_{j,i}$ on the event
$\{S_i\le S_{(\tau_j)}\}$. On this event, $S_i\le S_{(\ell)}$ for every
$\ell\ge\tau_j$, so the original and modified arrays agree over these
indices. Hence the defining function equals
$1$ at $k=\tau_j$ and equals $0$ for every $k>\tau_j$. It follows directly
that $\tilde\tau_{j,i}=\tau_j$.

Returning to~\eqref{eq:nullest}, we have 
\$
\eqref{eq:nullest} & \le 
\EE\bigg[\frac{\ind\{S_i \le S_{(\tilde \tau_{j,i})}\}}
{1 \vee (\ind\{S_i \le S_{(\tilde \tau_{j,i})}\}+ \sum_{\ell \in \cH_0\backslash \{j\}} 
\ind\{S_{n+\ell} \le S_{(\tilde \tau_{j,i})}\})}\bigg]\\
& = \frac{1}{|\cH_0|}
\EE\bigg[\frac{\ind\{S_i \le S_{(\tilde \tau_{j,i})}\} + 
\sum_{\ell' \in \cH_0 \backslash\{j\}}\ind\{S_{n+\ell'} \le S_{(\tilde \tau_{j,i})}\}}
{1 \vee (\ind\{S_i \le S_{(\tilde \tau_{j,i})}\}+ \sum_{\ell \in \cH_0\backslash \{j\}} 
\ind\{S_{n+\ell} \le S_{(\tilde \tau_{j,i})}\})}\bigg]\\
& \le 1/|\cH_0|,
\$
where the equality follows from the joint exchangeability of
$\{S_i\}\cup\{S_{n+\ell}\}_{\ell \in \cH_0\backslash \{j\}}$ under the
independent i.i.d. setup and the permutation-invariant score construction,
and from the invariance of $\tilde \tau_{j,i}$ under permutations of the
same collection. As a result,
$\EE[e_j/\hat \pi_0^{(j)}] \le m/|\cH_0|$. Summing this bound over
$j\in\cH_0$ gives $m$, completing the proof.

\begin{remark}
Recall that the level $\tilde\alpha$ used to construct the threshold also
determines the e-value; write this dependence as $T_k(a)$ and $e_j(a)$.
The same proof applies with the data-dependent level
$a=\tilde\alpha/\hat\pi_0^{(j)}$. Indeed, $\hat\pi_0^{(j)}$ is invariant
under permutations of the underlying units in
$\Dref\cup\{Z_{n+j}\}$, so this level, $T_k(a)$, and the auxiliary
threshold $\hat T_j(a)$ remain invariant. The threshold argument is unchanged for any
nonnegative $a$, using the convention $\inf\varnothing=+\infty$.
Consequently,
$\sum_{j\in\cH_0}\EE[e_j(\tilde\alpha/\hat\pi_0^{(j)})/
\hat\pi_0^{(j)}]\le m$.
\end{remark}

%% file: sections/appendix/ebhcc_results.tex
\red{
In Section~\ref{sec:ebhcc}, we define the boosted e-values directly through
$q_j$, which is the form used in our implementation. For completeness, we
first give the more general formulation of e-BH-CC 
(based on~\citet{lee2024boosting}), which
remains well-defined even when the minimum defining $q_j$ is not attained. Since
$\cA(c)$ is nondecreasing in $c$, the function $\phi_j(c;S_j)$ defined
in~\eqref{eq:cond_exp} is also nondecreasing in $c$. Assuming that the set
below is nonempty, define
\begin{equation}
    \hat c_j = \sup\{c\in[0,1]:\phi_j(c;S_j)\le 0\}.
\end{equation}
For the auxiliary rejection sets considered in the main text, nonemptiness
can be assumed because $\cA(0)=\varnothing$ and the 
first term inside the expectation
in~\eqref{eq:cond_exp} at $c=0$ is bounded above by $e_j$. Thus,
$\hat c_j$ is the largest auxiliary level, up to a possible boundary point,
at which the candidate boosted value has marginal expectation bounded by
that of $e_j$.
}

\red{
The general boosted e-value is defined as
\begin{equation}
    \label{eq:boosted_eval_general}
e_j^\boost =
\begin{cases}
\displaystyle
\frac{m}{\alpha |\cR \cup \{j\}|}
\ind\bigg\{e_j \ge \frac{m}{\alpha |\cR \cup \{j\}|}
\text{ or } j \in \cA(\hat c_j)\bigg\},
& \text{if }\phi_j(\hat c_j;S_j)\le 0,\\[2ex]
\displaystyle
\lim_{\ell\rightarrow\infty}
\frac{m}{\alpha |\cR \cup \{j\}|}
\ind\bigg\{e_j \ge \frac{m}{\alpha |\cR \cup \{j\}|}
\text{ or } j \in \cA(\hat c_{j,\ell})\bigg\},
& \text{if }\phi_j(\hat c_j;S_j)>0,
\end{cases}
\end{equation}
where, in the second case, $\{\hat c_{j,\ell}\}_{\ell\ge1}$ is an
increasing sequence satisfying $\phi_j(\hat c_{j,\ell};S_j)\le0$ and
$\hat c_{j,\ell}\rightarrow\hat c_j$. The first case applies when the
boundary $\hat c_j$ itself leads to $\phi_j$ evaluating to a 
nonpositive value. Otherwise, the second case defines the boosted 
e-value using values which approach the
boundary value $\hat c_j$ from below. The limit is well-defined
because the expression inside the limit is nondecreasing in
$\hat c_{j,\ell}$. The next subsection shows that, whenever the minimum
defining $q_j$ is attained, this general 
construction~\eqref{eq:boosted_eval_general} agrees with the direct
definition in~\eqref{eq:boosted_eval}.
}

\subsection{\red{Equivalence between the two definitions of $e_j^\boost$}}
\label{appd:equiv_comp}
\begin{lemma}\label{lem:equiv_comp}
For any $j\in [m]$, let $q_j = \inf\{c\in [0,1]: j \in \cR \cup \cA(c)\}$. 
If $j \in \cR\cup \cA(q_j)$, then the general boosted e-value
in~\eqref{eq:boosted_eval_general} satisfies
\$
e_j^\boost = \frac{m\ind\{\phi_j(q_j;S_j)\le 0\}}{\alpha |\cR \cup \{j\}|}. 
\$
Consequently, it agrees with the main-text definition
in~\eqref{eq:boosted_eval}.
\end{lemma}
\begin{proof}
Fix $j\in[m]$. We consider the case of  
$\phi_j(\hat c_j;S_j)\le 0$ and $\phi_j(\hat c_j;S_j) > 0$ in
\eqref{eq:boosted_eval_general}, respectively.

\red{By the  definition of the e-BH rejection set in
\citet{wang2022false}, for every $j\in[m]$,
\[
e_j \ge \frac{m}{\alpha|\cR\cup\{j\}|}
\quad\Longleftrightarrow\quad
j\in\cR.
\]
Thus, for any $c\in[0,1]$, the indicator appearing
in~\eqref{eq:boosted_eval_general} equals
$\ind\{j\in\cR\cup\cA(c)\}$.}

\begin{itemize}
\item \textbf{Case 1:} $\phi_j(\hat c_j;S_j) \le 0$.
It suffices to show that 
$\ind\{\phi_j(q_j;S_j)\le 0\} = \ind\{j \in \cR \cup \cA(\hat c_j)\}$.
If $\phi_j(q_j;S_j)\le 0$, then $q_j \le \hat c_j$
by definition. Then by assumption and 
the monotonicity of $\cA(\cdot)$, 
$j \in \cR \cup \cA(q_j) \subseteq \cR \cup \cA(\hat c_j)$.
Conversely, if $j \in \cR \cup \cA(\hat c_j)$, then by definition of $q_j$,
$\hat c_j \ge q_j$, and therefore
$0 \ge \phi_j(\hat c_j;S_j) \ge \phi_j(q_j;S_j)$. 
Combining the above, there is
$\ind\{\phi_j(q_j;S_j)\le 0\} = \ind\{j\in \cR\cup \cA(\hat c_j) \}$.
\item \textbf{Case 2:} $\phi_j(\hat c_j;S_j)>0$.
First, suppose that $\phi_j(q_j;S_j) \le 0$. Then by the choice of $\hat c_j$, 
there is $q_j < \hat c_j$. 
Since $\lim_{\ell \rightarrow \infty} \hat c_{j,\ell}
= \hat c_j$, there exists $\ell_0$, such that 
$\hat c_{j,\ell_0} \ge q_j$. We then have 
$j \in \cR \cup \cA(q_j) \subseteq \cR \cup \cA(\hat c_{j,\ell_0})$, 
and therefore $\lim_{\ell \rightarrow \infty} \ind\{j \in \cR \cup \cA(\hat c_{j,\ell})\}=1$. 
If instead $\phi_j(q_j;S_j) > 0$, then $q_j \ge \hat c_j$.
For any $\ell \ge 1$, $\hat c_{j,\ell} < q_j$, and therefore by the definition of $q_j$, 
$j \notin \cR \cup \cA(\hat c_{j,\ell})$. Consequently, $\lim_{\ell \rightarrow \infty} 
\ind\{j \in \cR \cup \cA(\hat c_{j,\ell})\} = 0$.
In this case, we also conclude that $\ind\{\phi_j(q_j;S_j)\le 0\}
= \lim_{\ell \rightarrow \infty} \ind\{j \in \cR \cup \cA(\hat c_{j,\ell})\}$.
\end{itemize}
Putting everything together, the general construction
in~\eqref{eq:boosted_eval_general} agrees with the definition
in~\eqref{eq:boosted_eval}.

\end{proof}
\subsection{Sampling distributions}
\label{appd:sampling_dist}
We focus on the case $\cA(c) = \{j \in [m] : p_j \le c\}$. 
For $j\in[m]$, let $k$ denote the block to which $Z_{n+j}$ belongs. 
Define the function
\$
G_j(c) := \frac{m \ind\{e_j \ge \frac{m}{\alpha |\cR \cup \{j\}|} \text{ or } p_j \le c\}}
{\alpha |\cR \cup \{j\}|}
\$
and note that $\phi_j(c;S_j) = \EE[G_j(c)-e_j \given S_j]$. 
Since $G_j(c)$ is a function of the $K$-FC e-values and p-values, it 
is fully determined given 
\$
\cV_j := 
\begin{bmatrix}
\Lbag \{V^{(1)}_i\}_{i\in[n]}\Rbag \cup \{V^{(1)}_{n+\ell}\}_{\ell\in[m]}\\
\vdots \\
\Lbag \{V^{(K)}_i\}_{i\in[n]}\Rbag \cup \{V^{(K)}_{n+\ell}\}_{\ell\in[m]}
\end{bmatrix},
\$
where we recall that $\Lbag A \Rbag$ refers to the multiset $A$ with its ordering removed.

For each $i\in[n]\cup\{n+j\}$, let $\bZ^{(i)}$ be $\bZ$ with
$Z_i$ and $Z_{n+j}$ exchanged, and let $Z_\ell^{(i)}$ denote its
$\ell$th component. Define
\[
\widetilde\Dref^{(i)}
=(\Dref\cup\{Z_{n+j}\})\setminus\{Z_i\},
\qquad
\widetilde\Dtest^{(i)}
=(\Dtest\setminus\{Z_{n+j}\})\cup\{Z_i\}.
\]
For each $b\in[K]$, define
\[
\widetilde B_b^{(i)}
=
\begin{cases}
(B_k\setminus\{Z_{n+j}\})\cup\{Z_i\}, & b=k,\\
B_b, & b\ne k,
\end{cases}
\]
and train
\[
\widetilde V^{(b,i)}(\cdot)
=
f\!\left(
\widetilde\Dref^{(i)}\cup\widetilde B_b^{(i)};
\widetilde\Dtest^{(i)}\setminus\widetilde B_b^{(i)}
\right).
\]
Set $\widetilde V_\ell^{(b,i)}
=\widetilde V^{(b,i)}(Z_\ell^{(i)})$ for $\ell\in[n+m]$, and let
$\cV_j^{(i)}$ denote the array in the preceding display defining $\cV_j$,
with each $V_\ell^{(b)}$ replaced by $\widetilde V_\ell^{(b,i)}$.
Then, under $H_j$,
\[
\cV_j\given S_j
\sim
\sum_{i\in[n]\cup\{n+j\}}
\frac{1}{n+1}\,\delta_{\cV_j^{(i)}}.
\]
\subsection{A provably valid computational shortcut}
\label{appd:comp_shortcut}

For the choice of $\cA(c) = \{j\in[m]: p_j \le c\}$,
we provide a computational shortcut inspired by the techniques used in~\citet{luo2022improving}.
Recall that for each $j$, the computational task amounts to evaluating
\$ 
\phi_j(q_j;S_j) = \EE_{H_j}[f_j \given S_j], \text{ where }
f_j := \frac{m \ind\{e_j' \ge \frac{m}{\alpha |\cR(\be') \cup \{j\}|} 
\text{ or } p_j' \le q_j \}}
{\alpha |\cR(\be')\cup \{j\}|} - e_j',
\$
where we write the e-BH rejection set as $\cR(\be')$ to emphasize its dependence on the e-values 
$\be' := (e_1',\ldots,e_m')$.
Note that the expectation above is taken over the distribution of 
$\bZ' := (Z_1',\ldots,Z_{n+m}') \given S_j$ under $H_j$, where the $K$-FC e-values $\be'$
and p-values $\bp'$ are functions of $\bZ'$, while $q_j$ is a function of the original data $\bZ$---it is to be distinguished from $\bZ'$---and is 
fixed when taking the expectation.

Next, we write $f_j$ as the difference of two terms: $f_j = f_j^{(1)} - f_j^{(2)}$, where 
\$ 
f_j^{(1)} := \frac{m \ind\{p_j' \le q_j\}}
{\alpha |\cR(\be')\cup \{j\}|}, \text{ and }
f_j^{(2)} := e_j' - \frac{m \ind\{p_j' > q_j, e_j' \ge \frac{m}{\alpha |\cR(\be') \cup \{j\}|}\}}
{\alpha |\cR(\be')\cup \{j\}|}.
\$
It is straightforward to see that $f_j^{(1)} \ge 0$; for the second term, we can 
check that 
\$ 
f_j^{(2)} \ge e_j' - \frac{m \ind\{e_j' \ge \frac{m}{\alpha |\cR(\be') \cup \{j\}|}\}}
{\alpha |\cR(\be')\cup \{j\}|}
\ge 0,
\$
where the last inequality uses $\ind\{e_j' \ge \frac{m}{\alpha |\cR(\be') \cup \{j\}|}\} \le 
\frac{e_j'\alpha |\cR(\be')\cup\{j\}|}{m}$.

Both $f_j^{(1)}$ and $f_j^{(2)}$ depend on the original data $\bZ$
and the newly generated data $\bZ'$, but since the expectation is taken
over the distribution of $\bZ'$, we can drop the notational dependence on $\bZ$ and treat $f_j^{(1)}$ and $f_j^{(2)}$ 
as functions of $\bZ'$ only. 
Let $\Omega_j^{(1)}:=\{\bZ': f_j^{(1)}(\bZ')>0\}$ and $\Omega_j^{(2)}:= \{\bZ': f_j^{(2)}(\bZ')>0\}$ 
denote the support of $f^{(1)}_j$ and $f^{(2)}_j$, respectively. 
Our first observation is that we can simplify $\Omega_j^{(1)}$ as 
\begin{equation}
\begin{split}    
\label{eq:fj1_support} 
p_j' \le q_j & \iff 
\frac{1 + \sum^n_{i=1} \ind\{V^{(k)}_i(\bZ') \ge V^{(k)}_{n+j}(\bZ')\}}{n+1} \le q_j \\
& \iff
1 + \sum^n_{i=1} \ind\{V^{(k)}_i(\bZ') \ge V^{(k)}_{n+j}(\bZ')\}
\le \lfloor q_j(n+1)\rfloor,
\end{split}
\end{equation}
where $k$ is the block to which $Z_{n+j}'$ belongs.
Note that we only need to compute $f_j^{(1)}$ for the samples 
satisfying the condition in~\eqref{eq:fj1_support}, 
since the integrand is zero otherwise.
Moreover, given $S_j$, the score function $V^{(k)}(\cdot)$ is fixed, 
so verifying~\eqref{eq:fj1_support} only requires sorting the scores,
without retraining the score function.

The computation of $f_j^{(2)}$ is more involved, as its support cannot be simplified to a
condition as in~\eqref{eq:fj1_support}. However, we can approximate $\Omega_j^{(2)}$ with 
a computationally simpler set $\tilde \Omega_j^{(2)}$ without losing the validity of the e-values.
This is achieved by the following lemma.
\begin{lemma}\label{lem:shortcut}
For any set $\tilde{\Omega}$ which is a subset of the space of $\bZ'$, we have 
\$ 
\EE[f_j(\bZ') \given S_j] \le \EE\big[\ind\{\bZ'\in \Omega_j^{(1)} \cup \tilde \Omega\} \cdot f_j(\bZ') \given S_j\big].
\$
\end{lemma}
\begin{proof}
By the decomposition, we have 
\$
& \EE[f_j(\bZ') \given S_j] \\ %
= ~& \EE\big[f_j^{(1)}(\bZ') \given S_j\big] - \EE\big[f_j^{(2)}(\bZ') \given S_j\big]\\
= ~&\EE\big[\ind\{\bZ'\in \Omega_j^{(1)}\} \cdot f_j^{(1)}(\bZ') \given S_j\big] 
- \EE\big[\ind\{\bZ'\in \Omega_j^{(2)}\} \cdot f_j^{(2)}(\bZ') \given S_j\big]\\
= ~&\EE\big[\ind\{\bZ'\in \Omega_j^{(1)}\} \cdot f_j^{(1)}(\bZ') \given S_j\big] 
- \EE\big[\ind\{\bZ'\in \Omega_j^{(2)}\} \cdot f_j^{(2)}(\bZ') \given S_j\big]\\ 
& \qquad \qquad \qquad \qquad \qquad \qquad \qquad  
- \underbrace{\EE\big[\ind\{\bZ'\in (\Omega_j^{(1)} \cup \tilde{\Omega}) \backslash \Omega_j^{(2)}\} \cdot f_j^{(2)} (\bZ')\given S_j\big]}_{=0}\\
\le~& \EE\big[\ind\{\bZ'\in \Omega_j^{(1)}\} \cdot f_j^{(1)}(\bZ') \given S_j\big] 
- \EE\big[\ind\{\bZ'\in \Omega_j^{(2)} \cap (\tilde{\Omega} \cup \Omega_j^{(1)})\} \cdot f_j^{(2)}(\bZ') \given S_j\big]\\ 
& \qquad \qquad \qquad \qquad \qquad \qquad \qquad - 
\EE\big[\ind\{\bZ'\in (\Omega_j^{(1)} \cup \tilde{\Omega}) \backslash \Omega_j^{(2)}\} \cdot f_j^{(2)}(\bZ') \given S_j\big]\\
= ~ &\EE\big[\ind\{\bZ'\in \Omega_j^{(1)} \cup \tilde \Omega \} \cdot f_j^{(1)}(\bZ') \given S_j\big] 
- \EE\big[\ind\{\bZ'\in \Omega_j^{(1)} \cup \tilde{\Omega}\} \cdot f_j^{(2)}(\bZ') \given S_j\big]\\
=~& \EE\big[\ind\{\bZ'\in \Omega_j^{(1)} \cup \tilde{\Omega}\} \cdot f_j(\bZ') \given S_j\big],
\$
where the inequality uses $f_j^{(2)} \ge 0$.
\end{proof}
Using the above lemma, define
$
\tilde {\phi}_j(q_j;S_j) = \EE\big[\ind\{\bZ' \in \tilde \Omega \cup \Omega_j^{(1)}\}\cdot f_j(\bZ') \given S_j\big]$.
Since $\tilde \phi_j(q_j;S_j)$ is an upper bound for $\phi_j(q_j;S_j)$, for any $j\in[m]$,
\$ 
\tilde e_j^\boost = \frac{m\ind\{\tilde \phi_j(q_j;S_j)\le 0\}}{\alpha |\cR \cup \{j\}|}, 
\$ 
is a valid e-value. As a result, we can replace $\Omega_j^{(2)}$ with a computationally 
simplified set $\tilde \Omega_j^{(2)}$ to obtain a (slightly) more conservative e-value. 
In our implementation, we choose
$ \tilde \Omega_j^{(2)}  := \big\{\bZ': V^{(k)}_{n+j}(\bZ') \ge T_k\big\}$,
where $T_k$ is the original threshold for block $k$, computed from $\bZ$.
To summarize, we only compute the integrand $f_j(\bZ')$ for the samples
satisfying
\@
\label{eq:est_support}
\begin{aligned}
\bZ' \in \Omega_j^{(1)} \cup \tilde \Omega_j^{(2)}
&= \bigg\{\bZ':\\[-0.2em]
&\quad 1 + \sum_{i=1}^n \ind\{V_i^{(k)}(\bZ')\ge V_{n+j}^{(k)}(\bZ')\}
\le \lfloor q_j(n+1)\rfloor\\[-0.2em]
&\quad \text{or } V^{(k)}_{n+j}(\bZ') \ge T_k\bigg\}.
\end{aligned}
\@

\clearpage

%% file: sections/appendix/null_proportion_correction.tex
%
%
%
%
As an artifact of e-BH, the proposed $K$-FC ND procedure controls 
the FDR by $\pi_0 \alpha$, where we recall that $\pi_0$ is the fraction of 
inliers in $\Dtest$. Since our FDR target is $\alpha$, controlling the FDR 
at a lower level means we are ``wasting'' some of our FDR budget, 
especially when $\pi_0$ is small.

To fully unleash the power of $K$-FC ND, we introduce an estimator 
for $\pi_0$ based on~\citet{gao2023adaptive} and use it to modify the e-BH-CC procedure, increasing its 
FDR control level to $\alpha$.
To construct such an estimator, we work with a potentially different score function
$\cS(\cdot)$ trained over $\Dref \cup \Dtest$, where the training 
procedure is invariant to the ordering of input samples.
We then assign $S_i = \cS(Z_i)$ for any $i\in [n+m]$.
Let $S_{(1)} \le \cdots \le S_{(n+m)}$ be the order statistics of 
$\{S_i\}_{i\in [n+m]}$ and define for any $t \in \RR$ that
\@
\begin{split}
\label{eq:pi0}
& \hat \pi^{(j)}_0(t) = \frac{n+1}{m}\frac{1 + \sum_{\ell \in [m]\backslash \{j\}}
\ind\{S_{n+\ell}\le t\}}
{\ind\{S_{n+j} \le t\} +  \sum_{i\in[n]} \ind\{S_i \le t\}}, \quad \forall 
j\in[m].
\end{split}
\@
Here and below, a positive quantity divided by zero is interpreted as
$+\infty$, and $1/(+\infty)=0$.
The null proportion estimator is constructed \textit{per hypothesis} $j$:
$
\hat \pi^{(j)}_0 = \hat \pi^{(j)}_0(S_{(\tau_j)}),
$
where $\tau_j$ is a stopping time of the backwards filtration 
$\{\cG_k\}_{k\in[n+m]}$, i.e., $\{\tau_j \ge k\} \in \cG_k$, where
\$
\cG_k \coloneqq \sigma\Bigg(
\begin{aligned}
&\left\{\sum_{\ell \in [m]\backslash \{j\}}
\indc{S_{n+\ell} \le S_{(k')}}: k'\in\{k,\ldots,n+m\}\right\}\\[-0.2em]
&{}\cup\left\{\ind\{S_{n+j}\le S_{(k')}\}+
\sum_{i=1}^n \indc{S_i\le S_{(k')}}: k'\in\{k,\ldots,n+m\}\right\}
\end{aligned}
\Bigg),
\quad k \in [n+m].
\$
For example, following the principle considered in~\citet{gao2023adaptive}, the stopping time 
can be taken as
\$
\tau_j  = 1\vee\sup\big\{k \in [n+m-1]: \hat \pi_0^{(j)}(S_{(k)}) \ge \hat \pi^{(j)}_0(S_{(k+1)})\big\},
\$
where the maximum with $1$ makes the stopping time well defined if the displayed set is empty.

As pointed out by~\citet{wang2022false}, a sufficient condition 
for e-BH to control the FDR is that $\sum_{j \in \cH_0} \EE[e_j] \le m$, 
and a set of e-values satisfying this condition are called \textit{compound} e-values by~\citet{ignatiadis2024compound}. 
When $e_1,\ldots,e_m$ are all valid, strict e-values, we can see that 
$\sum_{j \in \cH_0} \EE[e_j] \le m\pi_0$, allowing for further tightening of FDR control. The following proposition, however, shows that 
$\{e_j/\hat \pi^{(j)}_0\}_{j\in[m]}$ are compound e-values.

\begin{proposition}
\label{prop:nullest}
\emergencystretch=2em
Consider the $K$-FC e-values defined in~\eqref{eq:kfc_eval} and the 
null proportion estimators defined in~\eqref{eq:pi0}. \red{Assume that the conditions of
Theorem~\ref{thm:kfc_valid} hold.}
Then $\sum_{j\in \cH_0} \EE[e_j /\hat \pi^{(j)}_0] \le m$.
\end{proposition} 
The proof of Proposition~\ref{prop:nullest} is deferred to 
Appendix~\ref{appd:proof_nullest}. As a consequence of Proposition~\ref{prop:nullest}, 
applying e-BH to $\{e_j/\hat \pi_0^{(j)}\}$ at level $\alpha$ controls FDR at level $\alpha$. 
Moreover, following almost the same proof as Proposition~\ref{prop:boost_eval}, applying the boosting steps to $\{e_j/\hat \pi_0^{(j)}\}_{j\in[m]}$ discussed in
Section~\ref{sec:ebhcc} controls FDR at level $\alpha$.

%% file: sections/appendix/theoretical_examples_of_evalue_improvements.tex
\subsection{Exhibit A: e-BH-CC improves power over BH}
\label{appd:ebhcc_power_over_bh}

We give an example of a realistic low-signal setting where e-BH-CC achieves a number of rejections while BH remains powerless. 

\paragraph{The setting.} Let $n = 30, m=10$, and the FDR target $\alpha = \frac{10}{31}$. Each unit is a two-dimensional vector $Z = (X,Y)$. The inliers are sampled i.i.d.~from $P  = \frac{9}{10}\delta_{(0,0)} + \frac {1}{10} \delta_{(2,0)}$, and we specify three outliers in $\Dtest$:
\begin{align}
    \begin{split}
        \Dref &= \{Z_i \sim P\colon i\in[30]\} \\
        \Dtest &= \{Z_{n+j} = (2, 1)\colon j=1,2,3\} \cup \{Z_{n+j} \sim P \colon j = 4, \dots, 10\}.
    \end{split}
\end{align}

For ease of representation, we fix the non-conformity score used in conformal p-value and e-value: $V(Z) = V(X,Y) = X$. The auxiliary score, used for boosting via CC, will similarly be fixed as $W(Z) = W(X,Y) = Y$. That is, the main score used to assign nonconformity scores is the projection to the first coordinate, while the auxiliary score contains new information not used pre-boosting which is found in the second coordinate.

\paragraph{An early-stopping conformal e-value.} 

We give a variant of the 1-FC conformal e-value which prevents $T$ from being $+\infty$ in the ``hopeless case'' (i.e., no rejections made), as in \citet{ren2024derandomised}. Note that we use $V(\cdot)$ as the score, with the notation $V_{i} = V(Z_i) = X_i$ for $i\in[n+m]$.

\begin{align}\label{eq:kfc_eval_early}
    \begin{split} 
        e_j &= (n+1)\cdot \frac{ \ind\{V _{n+j} \ge T^{\text{early}}\}}{1+\sum_{i\in[n]}\ind\{V_i \ge T^{\text{early}}\}},\\
        \text{where }T^{\text{early}}  &= \inf\left\{ t\in \cV   \colon \frac{m}{ n+1} \cdot \frac{1 + \sum_{i=1}^{ n} \indc{V _i \ge t} }{1\vee \sum_{j=1}^m \indc{V _{n+j} \ge t}   } \le  \tilde \alpha 
        \text{ OR } 
        \sum_{j=1}^m \ind\{V_{n+j} \ge t\} < \frac{1}{\tilde \alpha}
        \right\} 
    \end{split}
\end{align}
The hopeless case is important, as this is where the main benefit of boosting with CC actualizes. With the standard stopping time $T$, e-BH-CC is unable to make any improvement upon e-BH and BH. For our example, we will take $\tilde \alpha = \alpha = \frac{10}{31}$.

\subsubsection{e-BH-CC improves over BH with positive probability}
\label{appd:cc_better_than_bh}
With a positive probability (where the randomness stems from the draw of inliers populating $\Dref$ and $\Dtest$), e-BH-CC is able to make rejections in a setting where BH is  powerless. 
We show this phenomenon by splitting the possible outcomes into cases based on $L_1 = |\{ i\in[30]\colon Z_i = (2,0) \}|$ and $L_2 = |\{ j\in[10]\colon Z_{n+j} = (2,0) \}|$. $L_1$ and $L_2$ represent the random number of inliers that are sampled from the second component of the mixture $P$ in $\Dref$ and $\Dtest$, respectively.

\paragraph{The initial rejections by BH and e-BH.}

To begin, first note that the BH rejection mechanism induces a relationship between $L_1$ and $L_2$. Since the p-values are initially based on $V$, the p-values for any test unit with $X_{n+j} = 2$ are $\frac{1 + L_1}{31}$, while the other units with $X_{n+j} = 0$ have p-value 1. There are $L_2 +3$ p-values of value $\frac{1 + L_1}{31}$, which means via the BH threshold that these p-values are all simultaneously rejected when
\begin{equation}
    \frac{1 + L_1}{31} \le \frac{(10/31) \cdot (L_2 + 3)}{10} = \frac{L_2 + 3}{31 } \iff L_1 \le L_2 + 2,
\end{equation}
which means that BH is powerless when $L_1 \ge L_2 + 3$. 

We can also check e-BH with $T^{\text{early}}$ can be characterized with the exact same inequality. Note that  $T^{\text{early}}\in\{0,2, +\infty\}$. However, it can never take the value of 0, as the ratio inside its definition is 1 and $\sum _{j=1}^{10} \ind\{ V_{n+j} \ge 0 \} = 10 \ge 1/\tilde\alpha = 3.1$. $T^{\text{early}} = 2$ in two ways: either $3 + L_2 < 3.1$ (activating the hopeless case) or $L_1 \le L_2 + 2$ (making the ratio less than $\frac{10}{31}$). Otherwise, $T^{\text{early}} = +\infty$, meaning even our hopeless case adjustment was not enough to prevent an infinite stopping time. The only relevant case for making a nonzero number of rejections is $T^{\text{early}} = 2$, where the e-values are 
$$
e_j = \frac{31}{1 + L_1} \ind\{X_{n+j} = 2\}.
$$
A quick check of the e-BH threshold at $\alpha = 10/31$ shows that we need $L_1 \le L_2 + 2$ to pass the threshold, meaning that $\cR^{\ebh}  = \cR^{\bh}$ w.p. 1 in our setting. Again, we reiterate that in this setting, irrespective of the early stopping time $T^{\mathrm{early}}$, \emph{BH and e-BH are still equivalent}. The key difference is the threshold $T^{\text{early}}$ having a non-infinite value when $3 + L_2 < 3.1$, i.e., exactly when $L_2 =0$.

\paragraph{The auxiliary statistic and CC mechanism.} Recall the main component of conditional calibration is the conditional expectation function
\begin{equation}
    \label{eq:cond_exp_appendix}
    \phi_j(c;S_j) := \EE_{H_j}\bigg[\frac{m \cdot \ind\{e_j \ge \frac{m}{\alpha |\cR\cup \{j\}|} 
    \text{ or } j \in \cA(c)\}}
    {\alpha |\cR\cup \{j\}|} - e_j
    \Biggiven S_j \bigg], 
\end{equation}
as written in \eqref{eq:cond_exp}. We define the auxiliary statistic in the form of a p-value as below
\begin{equation}
    p_j^{\mathrm{aux}} = \frac{1 + \sum_{i=1}^n \ind\{ W_i \ge W_{n+j} \}}{n+1}
\end{equation}
and define the auxiliary set as $\cA(c) = \{j \colon p_j^{\mathrm{aux}}\le c\}$. The boosting decision depends on the critical value $\hat c_j \coloneqq \sup\{c \in [0,1] \colon \phi_j(c;S_j) \le  0\}$. Note that the conditional distribution of the e-values given $S_j$ is supported on $n+1$ units, depending on which unit is ``switched'' with $Z_{n+j}$---this holds for each $j$ (see Appendix~\ref{appd:sampling_dist}). We will explicitly write the conditional expectation as a sum on a case-by-case basis to find $\hat c_j$ and show that the resulting boosting does have an effect.

\paragraph{Case I: $L_2 = 0, L_1 \ge 3$.} As detailed earlier, this is the case where BH and e-BH are both powerless. That is, the realized $\cR^\ebh = \cR^\bh = \{\}$. 

Fix $j$ such that it corresponds to a $(2,1)$ unit (i.e., $Z_{n+1}, Z_{n+2}, Z_{n+3}$). 
By definition of $L_1$ and $L_2$, we know that the $\cE_j$ in our conditioning statistic $S_j = (\cE_j, \{Z_{n+k}\}_{k\neq j})$ can be written as 
\begin{equation}
    \cE_j =  [ \underbrace{\{(2, 0), \dots, (2,0)\}}_{L_1 \text{ copies}} \cup  \underbrace{\{(0, 0), \dots, (0,0)\}}_{30-L_1 \text{ copies}} \cup \{(2,1)\} ].
\end{equation}
Furthermore, $L_1$ and $L_2$ are deterministic conditional on $S_j$. 
Now we do casework on which unit is ``switched'' into the $Z_{n+j}$ case.
\begin{itemize}
    \item $\tilde Z_{n+j} = (2,0)$. Since the initial e-value depends only on the first component of the unit, $\tilde e_j = e_j = \frac{31}{1  + L_1}$. Since the number of units in $\cE_j$ with first component equal to 2 is \textit{unchanged} under resampling, the other e-values are also unchanged. Hence, $\tilde {\cR}^\ebh =  {\cR}^\ebh = \{\}$ and $\tilde e_j \not \geq \frac{m}{\alpha |\tilde {\cR}^\ebh \cup\{j\}|} = \frac{m}{\alpha} = 31$ (since $L_1 > 0$). Lastly, we can check that under this resampling case, $\widetilde p_j^{\mathrm{aux}} = 1$. This case happens with probability $\frac{L_1}{31}$.

    \item $\tilde Z_{n+j} = (0,0)$. In this case, $\tilde e_j = 0$. From the exact reasoning as the previous subcase, $\tilde {\cR}^\ebh = \{\}$. $\widetilde p_j^{\mathrm{aux}} = 1$ in this case. This case happens with probability $\frac{30 - L_1}{31}$.

    \item $\tilde Z_{n+j} = (2,1)$. This is the case where the resampled and the original dataset are identical. $\tilde e_j = \frac{31}{1 + L_1}$, while $\tilde {\cR}^\ebh = \{\}$ (again) and $\widetilde p_j^{\mathrm{aux}} = \frac{1}{n+1} = \frac{1}{31}$. This case happens with probability $\frac{1}{31}$.
\end{itemize}
Hence, we can write
\begin{align}
    \begin{split}
        \phi_j(c; S_j) &= \frac{1}{31}\bigg( 31\cdot \ind\bigg\{c \ge \frac{1}{31}\bigg\}  - \frac{31}{1 + L_1}\bigg)\\ &\quad  + \frac{L_1}{31}\bigg( 31\cdot \ind\{ c \ge 1 \} - \frac{31}{1 + L_1}\bigg) + \frac{30 - L_1}{31} ( 31\cdot \ind\{ c \ge 1 \} - 0)\\
        &= \ind\bigg\{ c\ge \frac{1}{31} \bigg \}  +30 \cdot \ind\{c\ge 1\} - 1
    \end{split}
\end{align}  
which is positive when $c =1$ and nonpositive anywhere below. Hence, $\hat c_j = 1$, but since $\phi_j(1;S_j) >0$, we use the second branch of \eqref{eq:boosted_eval_general}. Choosing $\hat c_{j, \ell} = \frac{1}{31} + (1 - \frac{1}{\ell})\frac{30}{31}$, we see that $j \in \cA(\hat c_{j, \ell})$ for all $\ell \ge 1$ by virtue of $p_j^{\mathrm{aux}} = \frac{1}{31}$. Thus, we conclude that $e_j^\boost = \frac{m}{\alpha} = 31$ when $Z_{n+j} = (2,1)$ and $0$ otherwise. The result is that three e-values are boosted from $\frac{31}{1 +L_1}$ to $31$, meaning that we pass the e-BH threshold at $\frac{m}{3\alpha} = 31/3$, giving us $\cR^\ebh(e_1^\boost, \dots, e_{10}^\boost) = \{1,2,3\} \supsetneq \cR^\ebh = \cR^\bh = \{\}$ and achieving power 1.

Under the model $P$, this event happens with probability 
$$\PP(L_2 = 0 )(1 - \PP(L_1 \le 2))  = 0.9^{7} \cdot \bigg(1 - \bigg( 0.9^{30} + \binom{30}{1}\cdot 0.1\cdot  0.9^{29} + \binom{30}{2} \cdot 0.1^2 \cdot 0.9^{28} \bigg ) \bigg) $$
as $L_1\sim \Bin(30, \frac{1}{10} )$ and $L_2 \sim \Bin(7, \frac{1}{10} )$.
This comes out to approximately $0.282$.

\paragraph{Case II: $L_2 = 0, L_1 \le 2$.} 

In this case, BH and e-BH both have power 1, giving the rejection set $\{1,2,3\}$. Every other unit in $\Dtest$ has both coordinates 0, and a quick calculation shows that these units can never be boosted. However, by our definition of e-BH-CC, its boosted rejection set can never be less powerful than that of e-BH. Thus, $\cR^\ebh(e_1^\boost, \dots, e_{10}^\boost)=\cR^\ebh = \cR^\bh= \{1,2,3\} $ in this case.

\paragraph{Case III: $L_2 \ge 1$.}

In this case, e-BH-CC cannot improve the power of e-BH. If $L_1\le L_2+2$, BH and e-BH already have power 1. Now suppose $L_1\ge L_2+3$, where both are powerless, and fix $Z_{n+j}=(2,1)$. Under any resample whose test unit has first component 2, the ratio and early-stopping branches require, respectively, $L_1\le L_2+2$ and $L_2+3<3.1$; both fail, so the resampled $T^{\text{early}}=+\infty$ and $\tilde e_j=0$. If the resampled test unit has first component 0, the ratio branch requires $L_1\le L_2$ and fails; whether or not the early-stopping branch holds, $\tilde e_j=0$. Thus, $\tilde e_j=0$ for every resample. However, $q_j=p_j^{\mathrm{aux}}=1/31$ because $\cR^\ebh=\varnothing$, and the identity resample occurs with probability $1/31$ and has $\widetilde p_j^{\mathrm{aux}}=1/31\le q_j$, so $\phi_j(q_j;S_j)>0$. Hence, none of the three $(2,1)$ units can be boosted, and e-BH-CC makes no power improvement.

We conclude that $\PP(\text{e-BH-CC improves over BH})\approx 0.282$ in this stylized example.

\subsection{Exhibit B: model selection leads to FDR inflation via BH}
\label{appd:model_selection_fdr_inflation}

In Section \ref{sec:mfc}, we show how our e-value-based methodology allows for choosing between multiple model candidates in a data-driven manner without sacrificing FDR control. Here, we will show that attempting to do the same using conformal p-values and BH will incur selection bias and lead to a violation of FDR control. 

\paragraph{The setting.} Let $n = 3, m=2$, and the FDR target $\alpha = \frac 12$.  We will be in the global null setting: $Z_1, \dots, Z_5 \sim \cN(0,1)$ (though the choice of the Gaussian as the inlier distribution is WLOG; the argument works for any absolutely continuous distribution $P$). Hence,
\begin{align}
    \begin{split}
        \Dref &= \{Z_1, Z_2, Z_3\} \\
        \Dtest &= \{Z_4, Z_5\}.
    \end{split}
\end{align}

\paragraph{Model selection: the candidate models.} 
We consider two candidate learners, i.e., $L = 2$. For each $j\in\{1,2\}$ and $\ell\in \{ 1,2\}$, we train the nonconformity scores
\begin{equation}
    V^{(j, \ell)}(\cdot) = f^{(\ell)}(\Lbag \Dref \cup \{Z_{n+j}\}\Rbag;\{Z_{n+(3-j)}\})
\end{equation}
where we explicitly write that $f^{(\ell)}$ only depends on $\Dref \cup \{Z_{n+j}\}$ in a permutation agnostic manner. In our construction, we will ignore the second argument $Z_{n+(3-j)}$ and only train the score using $\Lbag \Dref \cup \{Z_{n+j}\}\Rbag$. 

We now explicitly define our two learner candidates. Define for each $j\in [2]$ the insertion rank relative to $\cE_j \coloneqq \Lbag \Dref \cup \{Z_{n+j}\}\Rbag$ as follows:
\begin{equation}
    r_j(z) = 1 + \sum_{x \in \cE_j }\ind\{x  > z\}.
\end{equation}
The summation above is invariant to the ordering of elements in $\cE_j$, so $r_j$ is as well. We then define
\begin{align}
    \begin{split}
        V^{(j,1)}(z) &= 5 - r_j(z)\\
        V^{(j,2)}(z) &= (5 - r_j(z) ) - 2 \cdot \ind\{ r_j(z) = 2 \} + 2\cdot \ind\{r_j(z)\ge 4\}.
    \end{split}
\end{align}
The first learner can be seen as an ``upper-tail'' score, which assigns a higher score to a unit the more values it exceeds in $\cE_j$. The second learner is a perturbation of the first.

\paragraph{Conformal p-values and ``proxy'' p-values for model selection.} 
For each $j\in[2], \ell\in[2]$, define the conformal p-value
\begin{equation}
    p_j^{(\ell)} = \frac{1 + \sum_{i=1}^n \ind\{ V^{(j,\ell)} (Z_i) \ge V^{(j,\ell)} (Z_{n+j})\} }{n+1}.
\end{equation}
Note that this is a slightly different proxy conformal p-value, since we use the entire bag $\Lbag \Dref \cup \{Z_{n+j}\}\Rbag$ as the proxy reference set, rather than the $n$ lowest scores. 
For each $j$, $Z_1,Z_2,Z_3,Z_{n+j}$ form a 4-tuple of i.i.d. continuous random variables, so $p_j^{(\ell)}\sim \frac 14 \mathrm{Unif}\{1,2,3,4\}$, which is superuniform.

For each $j$, we want to choose some $\hat\ell_j$ in a data-driven manner as described in Section \ref{sec:mfc}, and use $p_j^{(\hat\ell_j)}$. As specified in the main text, we compute ``proxy'' p-values by pretending that the reference set is $\Dref \cup \{Z_{n+j}\}$ and that the test set is $\Dtest \setminus \{Z_{n+j}\}$. Any decision made using these p-values is agnostic to the ordering of $\Dref \cup \{Z_{n+j}\}$, which directly implies that $p_j^{(\hat\ell_j)}$ is superuniform as well.

Since $|\Dtest| = 2$, the proxy test set for each $j$ simply contains the other test unit. Hence, writing $\tilde{p}^{(\ell)}_{2\given 1}$ as the proxy p-value of $Z_{n+2}$ for selecting $Z_{n+1}$'s model (and vice versa), we construct
\begin{align}
    \begin{split}
        \tilde{p}^{(\ell)}_{2\given 1}  & = \frac{1 +  \sum_{Z \in \cE_1} \ind\{ V^{(1, \ell)} (Z)\ge V^{(1,\ell)} (Z_{n+2}) \} }{n + 2}\\
        \tilde{p}^{(\ell)}_{1\given 2}  & = \frac{1 +  \sum_{Z \in \cE_2} \ind\{ V^{(2, \ell)} (Z)\ge V^{(2,\ell)} (Z_{n+1}) \} }{n + 2}.
    \end{split}
\end{align}
Although the main text specifies to use the learner which maximizes the length of the proxy rejection set $\cR^\bh (\{\tilde {p}^{(\ell)}_{k\given j }\colon k\neq j\})$, such a rule would be less useful when we only have a single proxy p-value per $j$. Instead, we will choose $\ell$ to minimize the p-values:
\begin{equation}
    \hat\ell(j)\coloneqq \argmin{\ell \in [2]} \;\tilde{p} ^{(\ell)}_{ (3-j) \given j }
\end{equation}
with a tie between the proxy p-values resulting in choosing $\ell=1$ by default.

\subsubsection{FDR is violated by using BH on $p_j^{(\hat\ell(j))}$}
\label{appd:fdr_violation_ms}
We now calculate the FDR violation, showing that $\fdr[\cR^\bh(p_1^{(\hat\ell(1))}, p_2^{(\hat\ell(2))})] = \frac 35 > \frac 12 $. Note that in our global null setting, the FDR is identical to the probability of making any rejection, so we will be calculating $\PP(|\cR^\bh(p_1^{(\hat\ell(1))}, p_2^{(\hat\ell(2))})| \ge 1)$.

Denote $K_1$ and $K_2$ to be the global ranks of $Z_{n+1}$ and $Z_{n+2}$, respectively, in $\Dref\cup\Dtest$. Since the 5 points are i.i.d. and continuous, there are 20 equally likely pairs for $(K_1, K_2)$ (since $K_1\neq K_2$). We can check that $V^{(1,1)}(Z_{n+2}) = 5 - K_2$ and $V^{(1,2)} (Z_{n+2})> V^{(1,1)}(Z_{n+2})$ if and only if $K_2 \in \{4,5\}$. Furthermore, both $\{V^{(1,1)} (Z_{i}) \colon i\in[4]\}$ and $\{V^{(1,2)} (Z_{i}) \colon i\in[4]\}$ are the set $\{1,2,3,4\}$. Thus, 
\begin{equation}
    \tilde p^{(1)}_{2 \given 1} = \frac{1  + \sum_{x=1}^4 \ind\{x \ge  V^{(1,1)}(Z_{n+2}) \}}{n+2 };\quad \tilde p^{(2)}_{2 \given 1} = \frac{1  + \sum_{x=1}^4 \ind\{x \ge  V^{(1,2)}(Z_{n+2}) \}}{n+2 }
\end{equation}
and, for $K_2\in\{4,5\}$, the respective pairs $\bigl(V^{(1,1)}(Z_{n+2}),V^{(1,2)}(Z_{n+2})\bigr)$ are $(1,3)$ and $(0,2)$, so $V^{(1,2)} (Z_{n+2})> V^{(1,1)}(Z_{n+2}) \implies \tilde p^{(2)}_{2 \given 1} < \tilde p^{(1)}_{2 \given 1}$. This means 
\begin{equation}
    \label{eq:ell_1_characterization}
    K_2 \in \{4,5\} \iff \hat\ell(1) = 2
\end{equation}
and by symmetry
\begin{equation}
    \label{eq:ell_2_characterization}
    K_1 \in \{4,5\} \iff \hat\ell(2) = 2.
\end{equation}
After characterizing the learner decision, we must now analyze what happens to the p-value under either choice of the learner. Again, by symmetry, fix $j$ and note that the rank of $Z_{n+j}$ in $\cE_j$ (in descending order) is $R_j \coloneqq K_j - \ind\{K_{3-j} < K_j\}$. Under learner 1, the p-value directly uses this rank: $p^{(1)}_j = \frac{R_j}{4}$. Under learner 2, the p-value is perturbed:
\begin{equation}
    p_j^{(2)} = \begin{cases}
        1/4 & \text{if } R_j = 1\\
        1 & \text{if } R_j = 2\\
        3/4 & \text{if } R_j = 3\\
        2/4 & \text{if } R_j = 4.
    \end{cases}
\end{equation}

Let us now examine the BH rejection event for $\alpha = \frac 12 $. Since $m=2$, we either reject when either p-value is at most $\frac{\alpha}{2} = \frac 14 $ or if both p-values are at most $\frac{\alpha}{2}\cdot 2 = \frac 12 $. We can rewrite the event as a union of disjoint events: $\{\min(p_1^{\hat\ell(1)}, p_2^{\hat\ell(2)}) = \frac 14 \} \cup \{p_1^{\hat\ell(1)} =p_2^{\hat\ell(2)} =\frac 12\}$. 

\paragraph{Probability of the first event.} From \eqref{eq:ell_1_characterization} and \eqref{eq:ell_2_characterization}, we know that under either learner, the p-value for $H_j$ will be $\frac14$ if and only if the rank of $Z_{n+j}$ in $\cE_j$ is 1. Hence, the first event happens when either $R_1=1$ or $R_2=1$. We can enumerate over all possible pairs of $(K_1, K_2)$ for which this happens, giving 8 pairs. Hence, 
\begin{equation}
    \PP\bigg(\min(p_1^{\hat\ell(1)}, p_2^{\hat\ell(2)}) = \frac 14 \bigg) = \frac{8}{20} = \frac{2}{5}.
\end{equation}

\paragraph{Probability of the second event.} There are two ways to obtain $p_j = \frac 12$: either we choose learner 1, and $R_j = 2$, or we choose learner 2, and $R_j = 4$. Enumerating over all 12 other pairs, we find that only $(K_1, K_2) \in \{ (2,3), (3,2), (4,5), (5,4) \}$ fulfill this condition for both $j\in[2]$. The first two pairs lead to both hypotheses selecting learner 1, and they both have rank $2$ in their respective $\cE_j$. The last two pairs lead to both hypotheses selecting learner 2, and both have rank $4$ in their respective $\cE_j$. Thus,
\begin{equation}
    \PP\bigg( \red{p_1^{\hat\ell(1)} = p_2^{\hat\ell(2)} = \frac 12 }\bigg ) = \frac{4}{20} = \frac{1}{5}.
\end{equation}

We conclude by calculating
\begin{multline}
    \PP\bigg(|\cR^\bh(p_1^{(\hat\ell(1))}, p_2^{(\hat\ell(2))})| \ge 1\bigg ) \\= \PP\bigg(\min(p_1^{\hat\ell(1)}, p_2^{\hat\ell(2)}) = \frac 14  \text{ OR } p_1^{\hat\ell(1)} =p_2^{\hat\ell(2)} =\frac 12\bigg) = \frac{2}{5} + \frac 1 5 = \frac 35 > \frac 12
\end{multline}
which is a violation of supposed FDR control. Meanwhile, the e-value-based conformal model selection mechanism provably controls the FDR, showing that e-values provide the airtight FDR guarantee that the unadjusted selected-p-value/BH construction considered here lacks.

%% file: sections/appendix/omitted_simulation_details.tex
\subsection{Specific implementations of e-BH-CC in the numerical experiments}
\label{appd:ebhcc_specific_implementations}

Among the experiments, conditional calibration is used only in the $m$-FC model-selection simulations of Section~\ref{sec:exp_ms}, the weighted $K=1$ simulations of Section~\ref{sec:exp_weighted}, and the $m$-FC malicious-prompt analysis of Section~\ref{sec:realdata}. The default unweighted $K=1$ simulations of Section~\ref{sec:exp_k1} do not use e-BH-CC. Whenever conditional calibration is used, we instantiate the framework of Section~\ref{sec:ebhcc} through the representation in \eqref{eq:boosted_eval}. The common denominator in the boosted threshold is the hybrid quantity
\[
\hat R_j  =
\begin{cases}
|\cR^{\ebh}|, & \text{if } j \in \cR^{\ebh} , \\
s_j , & \text{if }  j \notin \cR^{\ebh} \text{ and either } |\cR^{\mathrm{BH}} | = 0 \text{ or }  s_j   < |\cR^{\mathrm{BH}}|\\
|\cR^{\mathrm{BH}}  \cup \{j\}|, & \text{otherwise},
\end{cases}
\]
where $\cR^{\mathrm{BH}} $ is the BH rejection set formed from the resampled conformal p-values and $s_j $ is the number of resampled p-values tied with or smaller than $p_j $. Thus, the main experiment-specific distinction lies in the choice of auxiliary rejection family $\cA(c)$, equivalently in the scalar statistic $c_j $ entering \eqref{eq:cond_exp}.

Note that using this choice of $\hat R_j$ over the default $|\cR^{\ebh} \cup \{j\}|$ still guarantees uniform improvement of the boosted rejection set over the original e-BH rejection set, i.e., $\cR^\ebh(e^\boost_1, \dots, e^\boost_m) \supseteq \cR^\ebh(e_1, \dots, e_m)$. The exact argument used to prove Theorem 2 in~\citet{lee2024boosting} (their uniform improvement theorem) holds with our choice of $\hat R_j$.

\red{
Lastly, to prevent long computation times during the e-BH-CC procedure, 
we can stop the per-hypothesis boosting early without violating 
the FDR control. We use the method provided in~\citet[Section 3.2.3]{lee2024boosting}: we order the hypotheses based on the plausibility of 
them being rejected via the conditionally calibrated boost. Then, 
we go through in descending order of likelihood, and stop this process 
once we have failed to boost enough hypotheses in a row. 
}

\subsubsection{Unweighted $K=1$ simulations}

The default $K=1$ simulations of Section~\ref{sec:exp_k1} do not use conditional calibration. There we report the unboosted $1$-FC e-values from \eqref{eq:kfc_eval} and apply e-BH directly, alongside the SC and AdaDetect baselines. Accordingly, no  choice of $\cA(c)$   is required for these experiments.

\subsubsection{$m$-FC model-selection simulations}

For the model-selection experiments, the auxiliary family in Section~\ref{sec:ebhcc} is taken to be the pointwise threshold family
\[
\cA(c) = \{\ell \in [m] : p_\ell \le c\},
\]
where $p_\ell$ denotes the model-selected conformal p-value. Equivalently, the statistic entering \eqref{eq:cond_exp} is simply
\[
c_j = p_j.
\]
The denominator is the shared hybrid quantity $\hat R_j$ displayed above. Conditional calibration is applied after model selection, so the resampling step recomputes the selected score vector under each admissible swap determined by $S_j$. In the manuscript experiments, this is paired with the p-value-based top-three smoothed selector described in Section~\ref{sec:exp_ms}; the corresponding candidate library and hyperparameter choices are summarized in Section~\ref{appd:model_selection_settings}.

\subsubsection{Weighted $K=1$ simulations}

For the weighted $K=1$ experiments, the same hybrid denominator $\hat R_j$ is used, but the auxiliary family is based on BH applied to the weighted conformal p-values from \eqref{eq:weighted_kfc_pval}.
\red{For a weighted p-value vector $\bp=(p_1,\dots,p_m)$, let
$\cR_\alpha^{\mathrm{BH}}(\bp)$ denote its BH rejection set at the fixed target level $\alpha$, and define
\[
c_j(\bp)=\frac{p_j}{|\cR_\alpha^{\mathrm{BH}}(\bp)\cup\{j\}|},
\qquad
\cA_{\bp}(c)=\{\ell\in[m]:c_\ell(\bp)\le c\},
\quad c\in[0,1].
\]
Under each conditional resample, both the p-value vector and the BH rejection set in this definition are recomputed. The implementation multiplies both the observed statistic $c_j(\bp)$ and its resampled counterpart by the common positive factor $m/\alpha$. This rescaling does not change any comparison and therefore produces exactly the same boosted e-values and rejection set.}
This choice is used for both weighted $1$-FC and weighted SC. The distinction between the two procedures lies in the underlying score construction, and for SC additionally in the preliminary split of the reference data into training and calibration subsets, rather than in the conditional-calibration rule itself. The conditional law underlying the expectation in \eqref{eq:cond_exp} is the weighted-exchangeable law given in \eqref{eq:weighted_exch_cond_dist}.

\subsection{Specific implementations of model selection }
\label{appd:model_selection_settings}

Across the $m$-FC numerical experiments of Section~\ref{sec:exp_ms} and the malicious-prompt analysis of Section~\ref{sec:realdata}, we use the same ``top model ensembling'' model-selection mechanism from Section~\ref{sec:model_selection_mfc}. Since $K=m$ in both settings, each block contains a single test point. For each $j\in[m]$ and candidate learner $\ell$, we form the score vector $\{V_i^{(j,\ell)}\}_{i\in[n+m]}$ and compute the proxy p-values described in Section~\ref{sec:model_selection_mfc}: we sort $\Lbag V_1^{(j,\ell)},\dots,V_n^{(j,\ell)},V_{n+j}^{(j,\ell)}\Rbag$, take the $n$ smallest values as a proxy reference set, treat the remaining value together with $\{V_{n+r}^{(j,\ell)}\}_{r\neq j}$ as a proxy test set, and apply BH at level $\alpha$. Writing these proxy p-values as $\tilde p^{(\ell)}_{1 \given j},\dots,\tilde p^{(\ell)}_{m \given j}$, we define
\[
R^{\mathrm{proxy}}_{j,\ell}
=
\left|\cR^\bh_\alpha\!\left(\tilde p^{(\ell)}_{1 \given j},\dots,\tilde p^{(\ell)}_{m \given j}\right)\right|,
\qquad
B^{\mathrm{proxy}}_{j,\ell}
=
\sum_{r=1}^m \ind\!\left\{\tilde p^{(\ell)}_{r \given j} < \alpha\right\},
\]
and then set
\[
\Gamma_{j,\ell}
=
\frac{R^{\mathrm{proxy}}_{j,\ell}}{m}
+
10^{-6}\frac{B^{\mathrm{proxy}}_{j,\ell}}{m}.
\]
This is the proxy utility used to rank the candidate learners. Thus, the quantity denoted abstractly by $\Gamma_\ell$ in Section~\ref{sec:model_selection_mfc} is instantiated here as the block-specific utility $\Gamma_{j,\ell}$.

In the manuscript experiments, we do not use hard top-$1$ selection, but rather top-3 ensembling. For each $j$, we retain the three candidates with the highest proxy utility, map each candidate score row to a common scale via the empirical-rank transform
\[
\red{s_q \mapsto -\log\!\left(
\frac{2+\#\{r\in[n+m]\setminus\{q\}:s_r\ge s_q\}}{n+m+1}
\right),\qquad q\in[n+m],}
\]
applied within that candidate row, and then average the transformed rows with normalized  weights proportional to $\exp\{12\,\Gamma_{j,\ell}\}$. The resulting ensembled score collection is used to form both the model-selected conformal p-values (which are used in the e-BH-CC implementation; see Appendix \ref{appd:ebhcc_specific_implementations}) and the model-selected e-values. The hindsight-best and hindsight-worst curves reported in Sections~\ref{sec:exp_ms} and~\ref{sec:realdata} are oracle summaries computed post hoc from the same candidate library; they are not part of the data-driven selection rule itself.

For the numerical $m$-FC simulations, the candidate library is the ten-model collection used in Section~\ref{sec:exp_ms}. It consists of four Isolation Forest candidates with $25$, $50$, $100$, and $200$ trees, together with six one-class SVM candidates having parameter pairs
\[
\begin{aligned}
(\nu,\gamma)\in \{&(0.004,\texttt{auto}),\ (0.01,\texttt{auto}),\ (0.004,\texttt{scale}),\\
&(0.01,\texttt{scale}),\ (0.1,\texttt{scale}),\ (0.25,\texttt{scale})\}.
\end{aligned}
\]
The Isolation Forest candidates use the in-sample score construction, while the SVM candidates use the leave-one-out construction.

For the malicious-prompt analysis, the same exact selection mechanism (with the same choices and hyperparameters) is applied to a different candidate library. There, the candidate set contains twelve rows obtained by pairing the three embedding models supplied by~\citet{ayub2024embedding} (OpenAI, OctoAI, and MiniLM) with four Isolation Forest candidates having $50$, $100$, $200$, and $400$ trees. In the manuscript run, each embedding is first reduced to dimension $40$ by truncated SVD before the scores are formed. For each test unit, we then measure each scorer and construct the e-values with the top-3 ensembled scores as described above. The SC benchmark curves in Section~\ref{sec:realdata} are evaluated over the same twelve candidates after splitting the reference prompts according to $\rho\in\{25,50,75\}$, with the first $\lfloor (\rho/100)n\rfloor$ prompts used for training and the remainder used for calibration.

Additionally, note that in the malicious prompt detection application we construct a slightly more powerful variation of our e-value, outlined in \eqref{eq:kfc_eval_early}---such a choice does not affect validity of the downstream FDR control.

\paragraph{A computational shortcut for $m$-FC ND with model selection.}
When each candidate model is trained under either the 1-FC or LOO-1-FC framework, there is a further implementation shortcut. In that case, switching unit $Z_i$ with $Z_{n+j}$ (as in the conditional calibration computation step) does not change the trained model used to score \emph{any other unit using any learner}; the only change this leads to is switching the scores for $V_i^{(\ell)}$ and $V_{n+j}^{(\ell)}$ for all learners $\ell$.
Namely, we do not have to retrain any scores used in $e_{j'}$ for $j' \in [m] \setminus \{j\}$, meaning we only have to do one model fit per learner (or $n+m$ if we use the LOO framework). 
The $m$-FC structure is then used only in the proxy evaluation step. This shortcut is used for all learners in all model-selection experiments.

\paragraph{Relationship to existing conformal model selection approaches.}
Model selection within a conformal framework has been somewhat studied, due to the richness in conditional information (in the form of invariance to the ordering of certain units).~\citet{liang2024conformal} proposes a method for the related problem of conformal prediction. They want to choose the model which gives the shortest conformal prediction interval length; hence, they train each model over both the reference data and (single) test point in a permutation-invariant manner, then find which model achieves the shortest lengths for each of the $n+1$ units. The conditional information being used here is the unordered bag $\Lbag \Dref\cup \{Z_{n+1}\}$. 

However, as we argue, the multiple testing problem gives you additional information for conditioning (which means it can be used for model measurement and selection) in the form of all test units other than the current $\{Z_{n+j}\}$. This information is important as the BH rejection procedure is a function of these other test units, and having access to their values would allow us to measure a proxy for the power of BH for each score. 

The closest approach in the existing conformal literature is~\citet{bai2024optimized}, which is concurrent and independent work. Even their method, however, is a less general instance of our strategy, in two high-level ways.  First, they measure the BH rejection set via proxy p-values which use $\Dref \cup \{Z_{n+j}\}$ as the proxy reference set. Although this is valid under the null $H_j$, the corresponding alternative would induce a proxy reference set with an additional large non-conformity score, which would affect each proxy p-value and lead to BH acting differently (and potentially much less powerfully) on these proxy p-values and the true p-values. Meanwhile, our construction is robust to when $Z_{n+j}$ is a true outlier, by only using the $n$ smallest scores in the reference set (and the remaining largest score as the new test point for $j$), maintaining permutation invariance under the null but simulating the alternative better. Second, they only propose a  top-1 model selection, which we show can be improved to a general model ensembling strategy. We also note that their p-value approach will similarly lead to FDR violations due to the dependences induced by model selection, for which they have to adjust via a random pruning step on the BH rejection set.

\citet{marandon2024adaptive} also consider model selection on top of the conformal novelty detection problem. However, their method involves holding out part of their training fold (already split off from the reference dataset), which compounds the existing issues of stability and data efficiency.

%% file: sections/appendix/deferred_simulations.tex

\subsection{Deferred FDR plots from Section~\ref{sec:experiments}}
\label{appd:deferred_fdr}

The main text focuses on power. Here we collect the corresponding FDR plots for the primary experiments in Section~\ref{sec:experiments}.

\begin{figure}[H]
    \centering
    \includegraphics[width=0.64\textwidth]{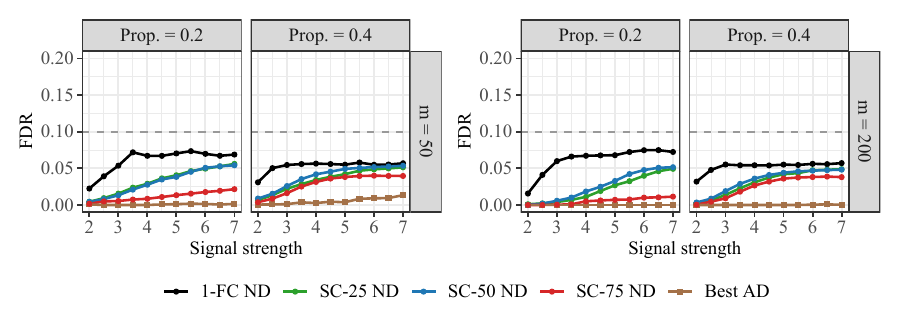}\par\vspace{0.08em}
    \includegraphics[width=0.64\textwidth]{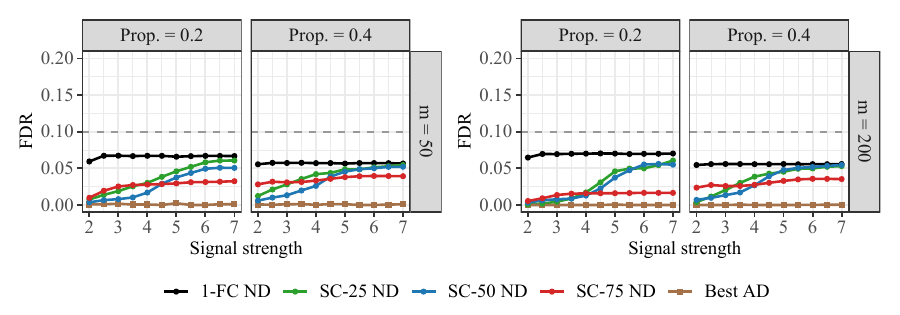}\par\vspace{0.08em}
    \includegraphics[width=0.64\textwidth]{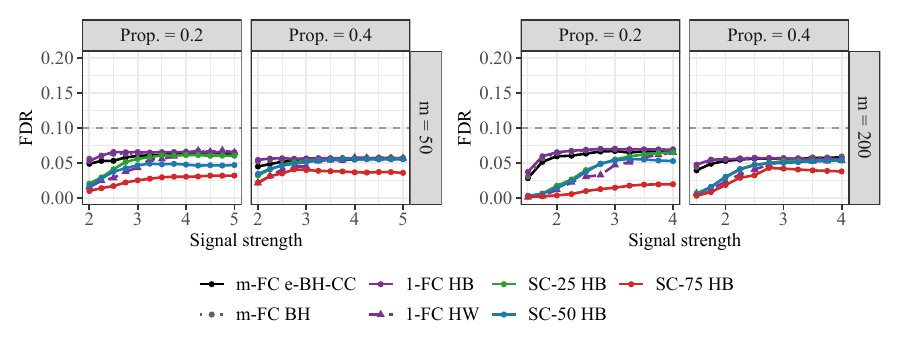}\par\vspace{0.08em}
    \includegraphics[width=0.64\textwidth]{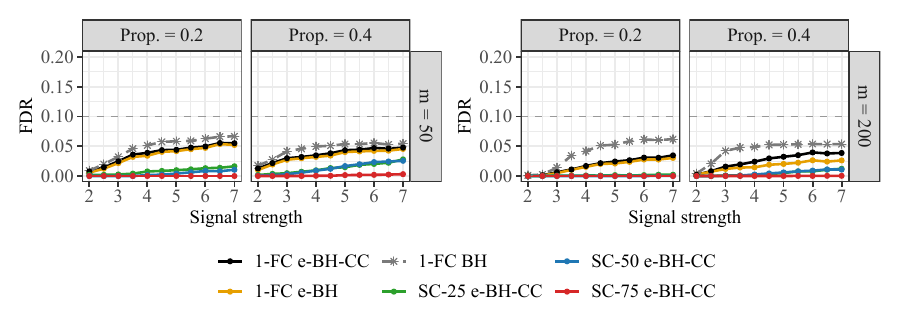}\par\vspace{0.08em}
    \includegraphics[width=0.64\textwidth]{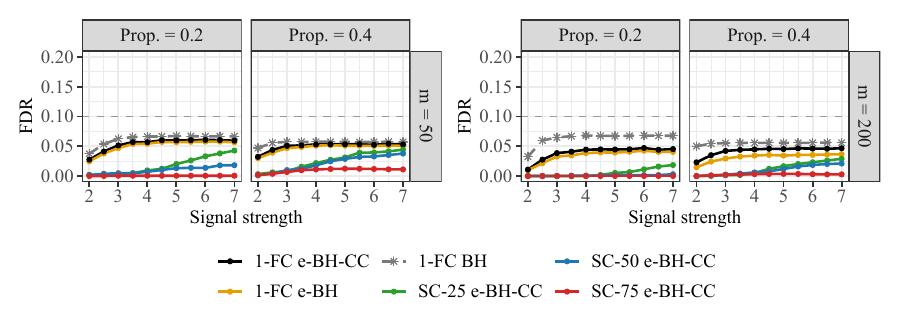}\par\vspace{0.08em}
    \caption{From top to bottom, these are the FDR plots for Fig.~\ref{draft:fig:n150_iso}, Fig.~\ref{draft:fig:n150_svm}, Fig.~\ref{draft:fig:Km_mdlsel}, Fig.~\ref{draft:fig:weighted_n150_if}, and Fig.~\ref{draft:fig:weighted_n150_svm}.}
    \label{appd:fig:deferred_fdr_stack}
\end{figure}

\begin{figure}[H]
    \centering
    \includegraphics[width=0.9\textwidth]{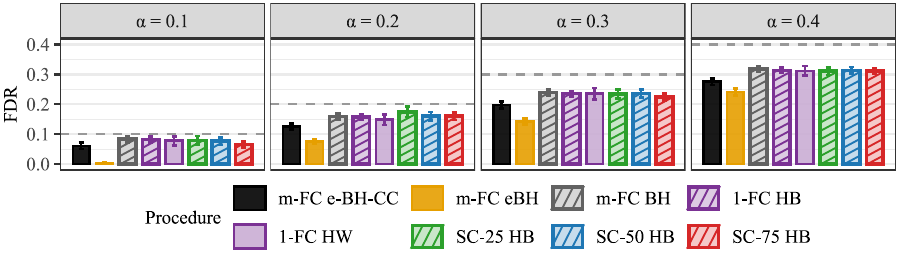}
    \caption{FDR results corresponding to Fig.~\ref{draft:fig:real_data_power} for the malicious-prompt experiment.}
    \label{appd:fig:real_data_fdr}
\end{figure}

\FloatBarrier
\subsection{HB/HW model-selection benchmarks}
\label{appd:model_selection_hbhw}

To complement the model-selection experiments in Sections~\ref{sec:exp_ms} and~\ref{sec:realdata}, we collect here the corresponding HB/HW comparisons with both power and FDR shown together. In each case, HB and HW are oracle benchmarks over the candidate library used in the corresponding experiment: for each framework and configuration, HB denotes the single candidate with the largest empirical power, while HW denotes the single candidate with the smallest empirical power (averaged over all replications). These curves therefore quantify the spread induced purely by model choice within the library.

\subsubsection{Numerical study}

The figure shows that the proposed $m$-FC e-BH-CC procedure lies much closer to the HB curves than to the HW curves, especially when $m=200$. This indicates that the   top-model ensembling strategy is effectively recovering strong candidates without access to oracle information. The gap between HB and HW is also substantial for the SC baselines, so the model-selection problem is nontrivial. At the same time, the lower panel shows empirical FDR remaining near the target level across the displayed settings, supporting the interpretation that the adaptive selection step improves power without introducing visible FDR inflation in these experiments.

\begin{figure}[H]
    \centering
    \includegraphics[width=\textwidth]{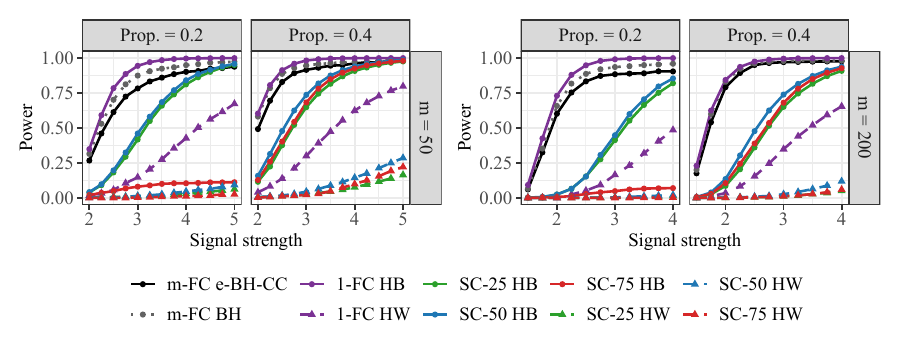}\par\vspace{0.12em}
    \includegraphics[width=\textwidth]{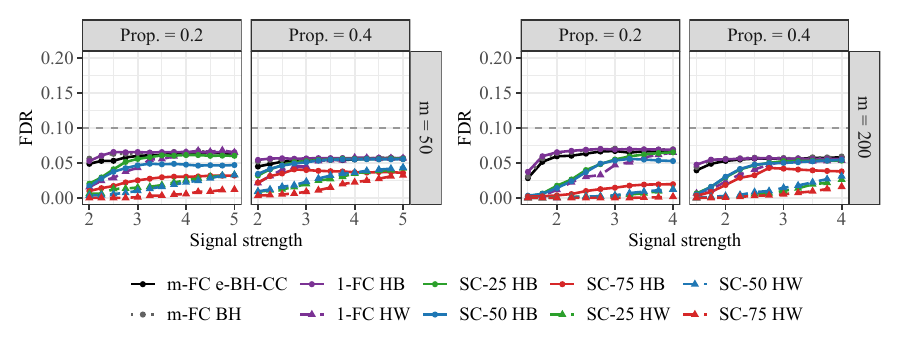}
\caption{From top to bottom, the power and FDR plots for the $m$-FC model-selection experiment when both hindsight-best and hindsight-worst oracle baselines are displayed. HB and HW are computed pointwise over the same ten-model candidate library used in Section~\ref{sec:exp_ms}. Each experiment uses 1,000 replications.}
    \label{appd:fig:model_selection_hbhw}
\end{figure}

\FloatBarrier
\subsubsection{Real-data application}

Figure~\ref{appd:fig:real_data_hbhw} gives the analogous HB/HW comparison for the malicious-prompt study from Section~\ref{sec:realdata}. Here the candidate library is still the twelve-model suite, and the $m$-FC procedure again applies the same top-3 model ensembling rule (1,000 observed replications).

The same qualitative pattern persists in the application setting. The boosted $m$-FC curve remains much closer to 1-FC HB than to 1-FC HW across the displayed FDR targets, while also improving over the unboosted $m$-FC e-BH curve. Relative to the SC oracle baselines, $m$-FC e-BH-CC is competitive with the strongest HB curves and clearly separated from the HW curves, indicating that the adaptive model ensembling is recovering useful candidates in the model library. The lower panel shows empirical FDR staying below the target levels throughout the displayed range.

\begin{figure}[H]
    \centering
    \includegraphics[width=\textwidth]{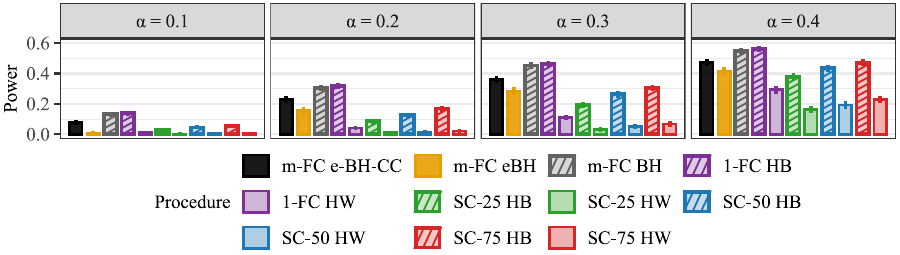}\par\vspace{0.12em}
    \includegraphics[width=\textwidth]{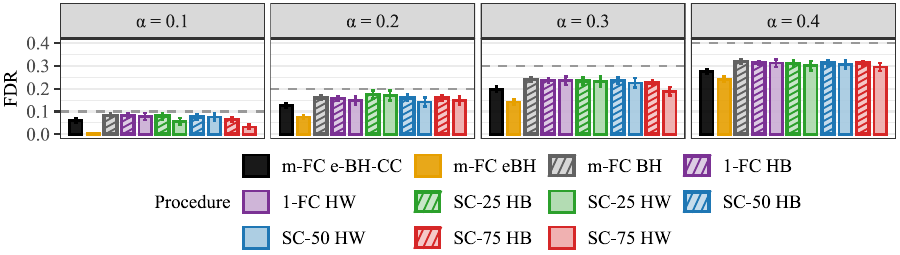}
\caption{From top to bottom, the power and FDR plots for the malicious-prompt experiment when both hindsight-best and hindsight-worst oracle baselines are displayed. HB and HW are computed pointwise over the same twelve-model candidate library used in Section~\ref{sec:realdata}. Each experiment uses 1,000 replications.}
    \label{appd:fig:real_data_hbhw}
\end{figure}

\FloatBarrier
\subsection{Large sample experiments}
\label{appd:large_sample_experiments}

We also repeated the unweighted setting experiments in a substantially larger regime with $n=1000$ and $m=1000$. Figure~\ref{appd:fig:large_sample_maintext_style} reports the corresponding power and FDR plots, comparing $1$-FC ND, the SC baselines, and the strongest AdaDetect curve. The corresponding derandomized comparison is deferred to Section~\ref{appd:nonrandom_vs_derandomized}; see Figure~\ref{appd:fig:large_sample_derandomized}. 

The SC and AD baselines perform much better in the large sample setting than in the limited reference setting; however, 1-FC ND still outperforms them all. We interpret this setting as a scenario where the split-based methods get enough training data for the corresponding power curves to look reasonable. The full-data efficient procedure still attains higher power.

\begin{figure}[H]
    \centering
    \includegraphics[width=\textwidth]{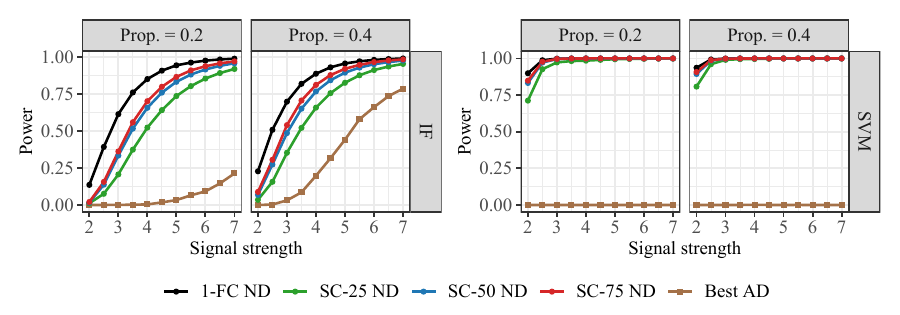}\par\vspace{0.12em}
    \includegraphics[width=\textwidth]{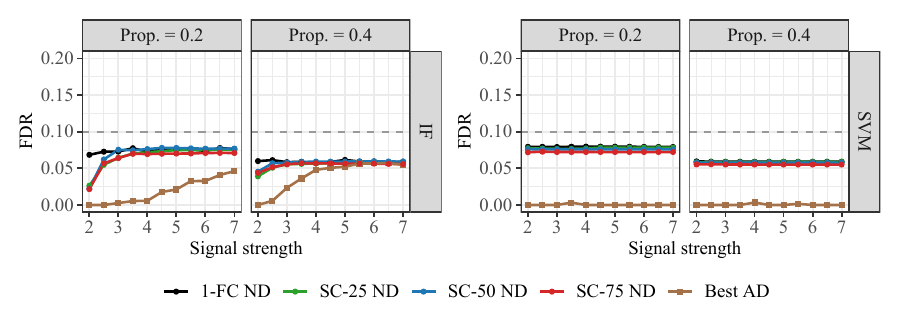}
\caption{From top to bottom, the power and FDR plots for the large-sample unweighted setting experiments with $n=1000$ and $m=1000$. The comparison includes $1$-FC ND, the SC baselines, and the strongest AdaDetect curve, with Isolation Forest panels on the left and one-class SVM panels on the right.}
    \label{appd:fig:large_sample_maintext_style}
\end{figure}

\FloatBarrier
\begingroup
\ifshowrevision\color{red!60}\fi
\subsection{$m$-FC BH versus e-BH-CC experiments}
\label{appd:mfc_bh_vs_ebhcc}

We use the data-generating mechanism from Section~\ref{sec:experiments} with
$n=150$, $m=15$, and $|\cH_1|=1$. We set $\alpha=0.1$ and vary the signal
strength over $a\in\{1.5,1.75,\ldots,5\}$, using 500 replications at each
value of $a$. The model collection and top-three ensembling rule are
identical to those in Section~\ref{sec:exp_ms}.

For each $j\in[m]$ and $\ell\in[10]$, define the candidate-wise conformal
p-value and its harmonic mean by
\[
p_j^{(\ell)}
=
\frac{1}{n+1}
\left(1+\sum_{i=1}^n
\ind\!\left\{V_i^{(j,\ell)}\geq V_{n+j}^{(j,\ell)}\right\}\right),
\qquad
\bar p_j=\frac{10}{\sum_{\ell=1}^{10}\bigl(p_j^{(\ell)}\bigr)^{-1}}.
\]
For e-BH-CC, we take $\tilde\alpha=0.1$, use the early-stopping threshold in
\eqref{eq:kfc_eval_early}, and set
$\cA(c)=\{j\in[m]:\bar p_j\leq c\}$. The statistic $\bar p_j$ is used to define the auxiliary rejection set for the
conditional calibration mechanism. We then apply the
e-BH-CC construction of Section~\ref{sec:ebhcc} with
$\hat R_j=|\cR^{\ebh}\cup\{j\}|$.
For comparison, we also implement 
the BH procedure applied to the 
model-selected conformal
p-values formed from the same ensembled scores
(that is, for each test point, we choose the p-value 
corresponding to ensembled scores of the top three 
models). 
Figure~\ref{fig:mfc_bh_vs_ebhcc} shows that e-BH-CC improves power over
$m$-FC BH in the moderate-to-high signal regime.
Both e-BH-CC and BH empirically control the FDR, 
as expected (though only the former has the guarantee).
We also compare to the hindsight best and worst models 
with 1-FC ND, as in Section~\ref{sec:exp_ms}.

\begin{figure}[H]
    \centering
    \includegraphics[width=0.65\textwidth]{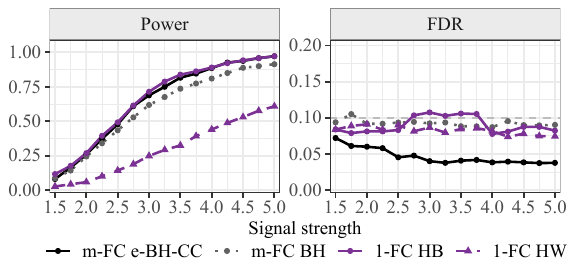} 
    \caption{\red{Power (left) and FDR (right) for $m$-FC BH and $m$-FC e-BH-CC in
    the sparse model-selection experiment with $n=150$, $m=15$, $|\cH_1| = 1$,
    and $\alpha=0.1$. 
    The simulation setting mostly follows that of 
    Section~\ref{sec:exp_ms}.
    Each experiment is comprised of  500
    replications.}}
    \label{fig:mfc_bh_vs_ebhcc}
\end{figure}

\FloatBarrier
\endgroup
\subsection{Non-random versus derandomized plots}
\label{appd:nonrandom_vs_derandomized}

Although the experiments in Section~\ref{sec:experiments} primarily compare FC-based methods to the standard random splitting baselines, one may also wish to compare against derandomized SC procedures in the sense of \citet{bashari2024derandomized}. Briefly, we choose $D$ different random splits of $\Dref$ into train and calibration sets, produce split conformal e-values $\{e_j^{(d)}\}_{j\in[m]}$ \eqref{eq:sc_eval} on each split, and average them:
\[
    \bar e_j \coloneqq \frac{1}{D}\sum_{d=1}^D e_j^{(d)}.
\]
We then apply e-BH to $\{\bar e_j\}_{j\in[m]}$. Throughout this subsection, we take $D=20$.

\Needspace{0.72\textheight}
\subsubsection{Comparison of $1$-FC and LOO-$1$-FC ND to derandomized SC baselines}
Figures~\ref{fig:small_sample_power_plot_nonrandom} and~\ref{fig:unweighted_svm_small_sample_power_plot_nonrandom} compare the $1$-FC and LOO-$1$-FC ND experiments from Section~\ref{sec:experiments} to derandomized SC-$\rho$ methods, using Isolation Forest and one-class SVM scores, respectively. As derandomization typically lowers the power of the e-BH rejection set (see, e.g., \cite{ren2024derandomised, lee2024boosting}), the FC procedures remain favorable in these comparisons. Figure~\ref{appd:fig:large_sample_derandomized} reports the corresponding large-sample derandomized comparison.

\begin{figure}[H]
    \centering
    \includegraphics[width=0.90\textwidth]{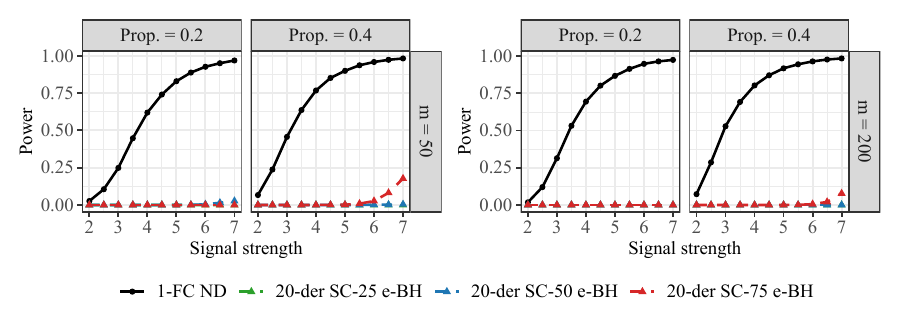}\par\vspace{0.12em}
    \includegraphics[width=0.90\textwidth]{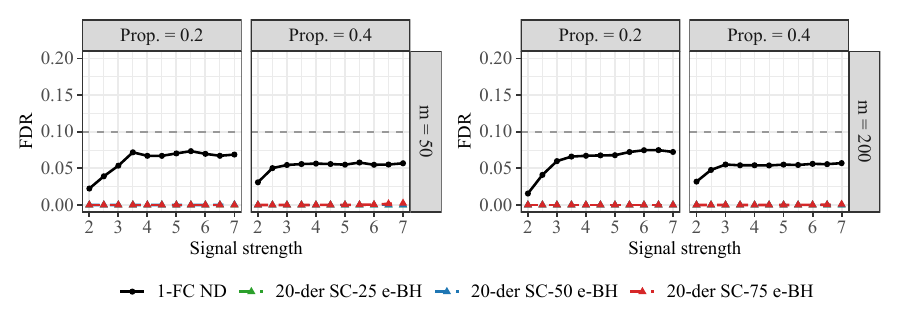}
    \caption{From top to bottom, the power and FDR plots corresponding to Fig.~\ref{draft:fig:n150_iso}, comparing $1$-FC ND to derandomized SC-based methods for the unweighted setting experiment with Isolation Forest scores.}
    \label{fig:small_sample_power_plot_nonrandom}
\end{figure}

\begin{figure}[H]
    \centering
    \includegraphics[width=0.86\textwidth]{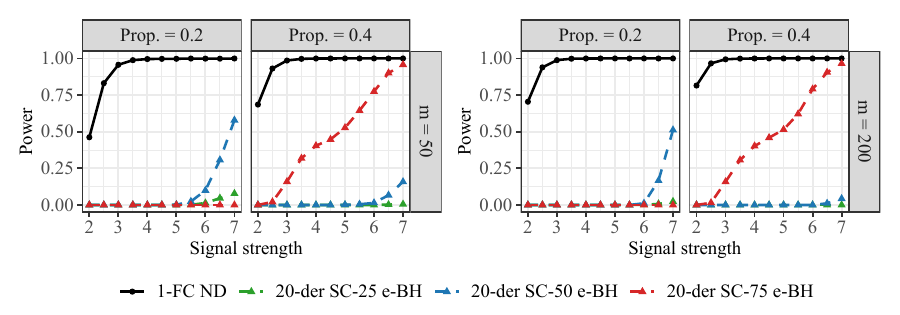}\par\vspace{0.12em}
    \includegraphics[width=0.86\textwidth]{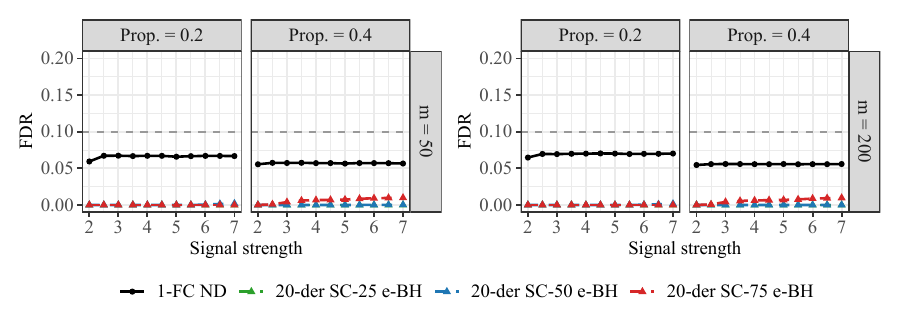}
    \caption{From top to bottom, the power and FDR plots corresponding to Fig.~\ref{draft:fig:n150_svm}, comparing LOO-$1$-FC ND to derandomized SC-based methods for the unweighted setting experiment with one-class SVM scores.}
    \label{fig:unweighted_svm_small_sample_power_plot_nonrandom}
\end{figure}

\begin{figure}[H]
    \centering
    \includegraphics[width=0.86\textwidth]{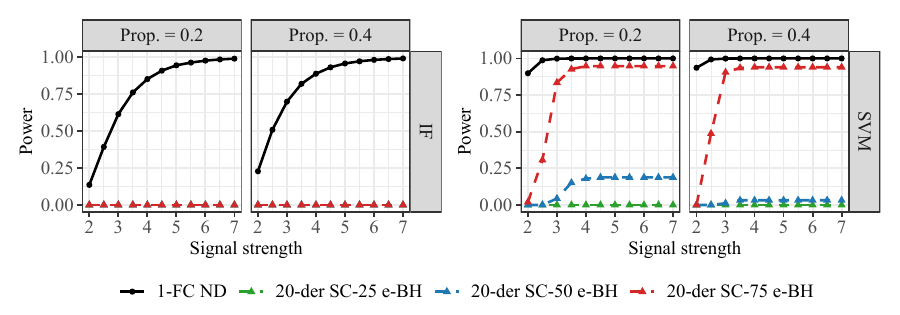}\par\vspace{0.12em}
    \includegraphics[width=0.86\textwidth]{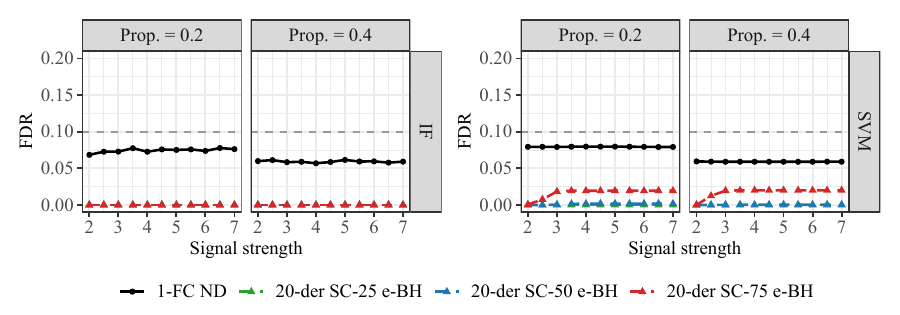}
\caption{From top to bottom, the power and FDR plots comparing $1$-FC ND to derandomized SC baselines for the large-sample experiment with $n=1000$ and $m=1000$. Isolation Forest panels appear on the left and one-class SVM panels on the right.}
    \label{appd:fig:large_sample_derandomized}
\end{figure}

\Needspace{0.76\textheight}
\subsection{Null proportion correction experiments}
\label{sec:null_prop_experiments}

Under the same simulation setting as in Section~\ref{sec:exp_k1}, we implement the null proportion correction for $1$-FC e-BH detailed in Appendix~\ref{sec:null_prop}. The specific stopping time $\tau_j$ is chosen as
\[
\tau_j^{(10)} \coloneqq 1\vee\sup \left\{ k \in [n+m-1] \colon
\sum_{l=k}^{n+m-1}\ind\!\left\{\hat{\pi}_0^{(j)}(S_{(l)})
\ge \hat{\pi}_0^{(j)}(S_{(l+1)})\right\}=10 \right\}.
\]
Intuitively, as we travel backward from $k=n+m-1,\dots,1$, we stop at the tenth index where the event $\hat{\pi}_0^{(j)}(S_{(k)}) \ge \hat{\pi}_0^{(j)}(S_{(k+1)})$ occurs, defaulting to $1$ if fewer than ten such indices exist. Figure~\ref{fig:null_prop} compares the uncorrected $1$-FC e-BH procedure with its $\pi_0$-corrected counterpart. The correction increases power, especially in the weaker-signal regime, while shifting the empirical FDR behavior from $\pi_0 \alpha$ toward the target level $\alpha$.

\begin{figure}[H]
    \centering
    \includegraphics[width= \textwidth]{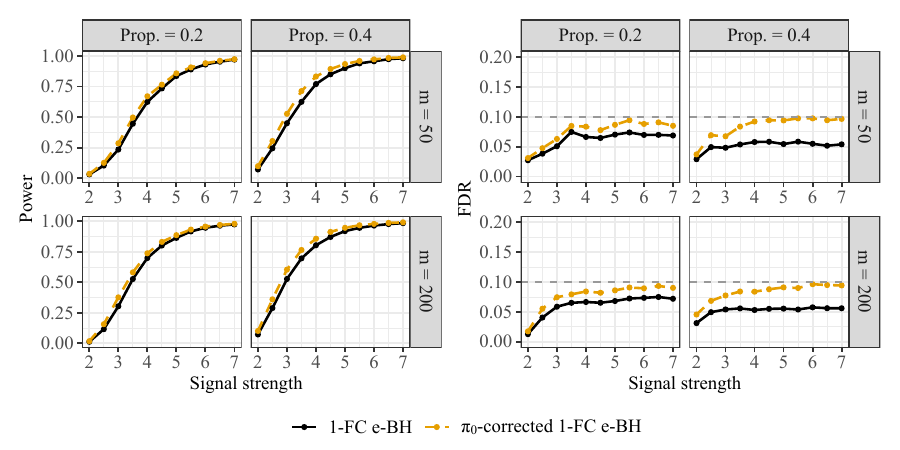}
    \caption{Null proportion correction experiments for $1$-FC e-BH and $\pi_0$-corrected $1$-FC e-BH with $m=200$, $n=150$, and $\alpha=0.1$. The left panel reports power and the right panel reports FDR. Each experiment uses 500 replications.}
    \label{fig:null_prop}
\end{figure}

\FloatBarrier

%% file: sections/appendix/algorithms.tex
\setcounter{AlgoLine}{0}
\begin{algorithm}[H]
\footnotesize
\SetAlgoNlRelativeSize{-1}
\SetAlCapFnt{\normalsize}
\SetAlCapNameFnt{\normalsize}
\DontPrintSemicolon
\caption{$K$-block full conformal ND procedure with conditional calibration}\label{alg:kfc_nd_cc}
\KwIn{reference dataset $\mathcal{D}_{\text{ref}} = \{Z_1, \dots, Z_n\}$;
test dataset $\mathcal{D}_{\text{test}} = \{Z_{n+1}, \dots, Z_{n+m}\}$;
score model-to-train $f(\cdot,\cdot)$ that is invariant to the ordering of the samples
in its first argument; number of blocks $K \in [m]$; data-independent partition rule; target FDR level $\alpha$.}

Partition $[m]$ into $I_1,\dots,I_K$ and set
$B_k \gets \{Z_{n+\ell}:\ell\in I_k\}$ for $k\in[K]$.

\For{$k \in [K]$}{
    Train $V^{(k)} (\cdot) \gets f( \Dref \cup B_k,\Dtest \backslash B_k )$.

    \For{$i \in [n+m]$}{
        $V_{i}^{(k)} \gets V^{(k)} (Z_{i}) $.
    }
    Compute threshold $T_k$ at level $\alpha$ using
    $\{V_{i}^{(k)}\}_{i\in[n+m]}$ as in \eqref{eq:kfc_eval}.

    \For{$j \in I_k$}{
        Compute $e_j$ and $p_j$ using $\{V_{i}^{(k)}\}_{i\in[n+m]}$
        as in \eqref{eq:kfc_eval} and \eqref{eq:kfc_pval}, respectively.
    }
}

Set $\cA(c) \gets \{\ell\in[m]:p_\ell\le c\}$ for $c\in[0,1]$ and
$\cR_0 \gets \cR_\alpha^\ebh(e_1,\ldots,e_m)$.

\For{$j \in [m]$}{
    Set $S_j \gets \Lbag\Dref\cup\{Z_{n+j}\}\Rbag
    \cup(\Dtest\backslash\{Z_{n+j}\})$ and
    $q_j \gets \min\{c\in[0,1]:j\in\cR_0\cup\cA(c)\}$.

    $\widetilde\phi_j(q_j;S_j) \gets 0$.

    \If{$j \in \cR_0$}{
        $e_j^\boost \gets \frac{m}{\alpha |\cR_0|}$.

        \textbf{continue}  {\color{gray}\tcp{Assign the e-BH critical value directly for an already rejected $j$.}}
    }
    \Else{
        \For{$i \in [n] \cup \{n+j\}$}{
            Let $\widetilde{\mathbf Z}^{(i)}
            =(\widetilde Z_1^{(i)},\dots,\widetilde Z_{n+m}^{(i)})$
            be $(Z_1,\dots,Z_{n+m})$ with components $i$ and $n+j$ exchanged.

            Set $\widetilde\Dref^{(i)}\gets
            \{\widetilde Z_\ell^{(i)}:\ell\in[n]\}$,
            $\widetilde\Dtest^{(i)}\gets
            \{\widetilde Z_{n+\ell}^{(i)}:\ell\in[m]\}$, and
            $\widetilde B_b^{(i)}\gets
            \{\widetilde Z_{n+\ell}^{(i)}:\ell\in I_b\}$ for $b\in[K]$.

            \If{$\widetilde{\mathbf Z}^{(i)}
            \notin \Omega_j^{(1)}\cup\widetilde\Omega_j^{(2)}$
            as per~\eqref{eq:est_support}}{
                \textbf{continue} {\color{gray}\tcp{This resample is omitted by the restricted upper-bound shortcut.}}
            }
            \Else{
                \For{$b \in [K]$}{
                    Train $\widetilde V^{(b,i)}(\cdot)\gets
                    f(\widetilde\Dref^{(i)}\cup\widetilde B_b^{(i)},
                    \widetilde\Dtest^{(i)}\backslash\widetilde B_b^{(i)})$.

                    Set $\widetilde V_\ell^{(b,i)}\gets
                    \widetilde V^{(b,i)}(\widetilde Z_\ell^{(i)})$
                    for every $\ell\in[n+m]$.
                }

                Compute all resampled e-values $\widetilde e_\ell^{(i)}$ and
                p-values $\widetilde p_\ell^{(i)}$ from the indexed scores as in
                \eqref{eq:kfc_eval} and \eqref{eq:kfc_pval}, using level $\alpha$.

                Set $\widetilde\cR^{(i)}\gets
                \cR_\alpha^\ebh(\widetilde e_1^{(i)},\dots,\widetilde e_m^{(i)})$ and
                \[
                \widetilde f_j^{(i)}\gets
                \frac{m\ind\!\left\{
                \widetilde e_j^{(i)}\ge
                \frac{m}{\alpha|\widetilde\cR^{(i)}\cup\{j\}|}
                \text{ or }\widetilde p_j^{(i)}\le q_j\right\}}
                {\alpha|\widetilde\cR^{(i)}\cup\{j\}|}
                -\widetilde e_j^{(i)}.
                \]

                $\widetilde\phi_j(q_j;S_j) \gets
                \widetilde\phi_j(q_j;S_j)+\frac{1}{n+1}\widetilde f_j^{(i)}$.
            }
        }
        {\color{gray}\tcp{$\widetilde\phi_j(q_j;S_j)\ge\phi_j(q_j;S_j)$.}}

        $e_j^\boost\gets
        \frac{m}{\alpha|\cR_0\cup\{j\}|}
        \ind\{\widetilde\phi_j(q_j;S_j)\le0\}$.
    }
}
$\mathcal{R} \gets \cR_\alpha^\ebh(e_1^\boost,\dots,e_m^\boost)$.

\KwOut{Rejection set $\mathcal{R}$}
\end{algorithm}

\setcounter{AlgoLine}{0}
\begin{algorithm}[H]
\footnotesize
\SetAlgoNlRelativeSize{-1}
\SetAlCapFnt{\normalsize}
\SetAlCapNameFnt{\normalsize}
\DontPrintSemicolon
\caption{Weighted $K$-block full conformal ND procedure with conditional calibration}\label{alg:weighted_kfc_nd_cc}
\KwIn{reference dataset $\mathcal{D}_{\text{ref}} = \{Z_1, \dots, Z_n\}$;
test dataset $\mathcal{D}_{\text{test}} = \{Z_{n+1}, \dots, Z_{n+m}\}$;
weight function $w(\cdot)$; score model-to-train $f(\cdot,\cdot)$ that is invariant to the
ordering of the samples in its first argument; number of blocks $K \in [m]$;
data-independent partition rule; fixed measurable auxiliary rule
$\mathfrak A_\alpha:(p_1,\dots,p_m)\mapsto\{\cA(c)\}_{c\in[0,1]}$, nondecreasing with
$\cA(0)=\varnothing$ and attainable $q_j$; target FDR level $\alpha$.}

Set $w_i\gets w(Z_i)$ for every $i\in[n+m]$.
Partition $[m]$ into $I_1,\dots,I_K$ and set
$B_k \gets \{Z_{n+\ell}:\ell\in I_k\}$ for $k\in[K]$.

\For{$k \in [K]$}{
    Train $V^{(k)} (\cdot) \gets f( \Dref \cup B_k,\Dtest \backslash B_k)$.

    \For{$i \in [n+m]$}{
        $V_{i}^{(k)} \gets V^{(k)} (Z_{i}) $.
    }
    \For{$j \in I_k$}{
        Compute threshold $T_j$ and e-value $e_j$ at level $\alpha$ from
        $\{V_i^{(k)}\}_{i\in[n+m]}$ and $\{w_i\}_{i\in[n+m]}$
        as in \eqref{eq:weighted_kfc_eval}.

        Compute $p_j$ from the same scores and weights as in
        \eqref{eq:weighted_kfc_pval}.
    }
}

Set $\{\cA(c)\}_{c\in[0,1]}\gets\mathfrak A_\alpha(p_1,\dots,p_m)$ and
$\cR_0 \gets \cR_\alpha^\ebh(e_1,\ldots,e_m)$.

\For{$j \in [m]$}{
    Set $S_j \gets \Lbag\Dref\cup\{Z_{n+j}\}\Rbag
    \cup(\Dtest\backslash\{Z_{n+j}\})$ and
    $q_j \gets \min\{c\in[0,1]:j\in\cR_0\cup\cA(c)\}$.

    Set $\phi_j(q_j;S_j)\gets0$ and
    $w^{(j)}\gets w_{n+j}+\sum_{i=1}^n w_i$.

    \If{$j \in \cR_0$}{
        $e_j^\boost \gets \frac{m}{\alpha |\cR_0|}$.

        \textbf{continue}  {\color{gray}\tcp{Assign the e-BH critical value directly for an already rejected $j$.}}
    }
    \Else{
        \For{$i \in [n] \cup \{n+j\}$}{
            Let $\widetilde{\mathbf Z}^{(i)}
            =(\widetilde Z_1^{(i)},\dots,\widetilde Z_{n+m}^{(i)})$
            be $(Z_1,\dots,Z_{n+m})$ with components $i$ and $n+j$ exchanged.

            Set $\widetilde\Dref^{(i)}\gets
            \{\widetilde Z_\ell^{(i)}:\ell\in[n]\}$,
            $\widetilde\Dtest^{(i)}\gets
            \{\widetilde Z_{n+\ell}^{(i)}:\ell\in[m]\}$, and
            $\widetilde B_b^{(i)}\gets
            \{\widetilde Z_{n+\ell}^{(i)}:\ell\in I_b\}$ for $b\in[K]$.

            Set $\widetilde w_\ell^{(i)}\gets
            w(\widetilde Z_\ell^{(i)})$ for every $\ell\in[n+m]$.

            \For{$b \in [K]$}{
                Train $\widetilde V^{(b,i)}(\cdot)\gets
                f(\widetilde\Dref^{(i)}\cup\widetilde B_b^{(i)},
                \widetilde\Dtest^{(i)}\backslash\widetilde B_b^{(i)})$.

                Set $\widetilde V_\ell^{(b,i)}\gets
                \widetilde V^{(b,i)}(\widetilde Z_\ell^{(i)})$
                for every $\ell\in[n+m]$.
            }

            Compute all resampled e-values $\widetilde e_\ell^{(i)}$ and
            p-values $\widetilde p_\ell^{(i)}$ from the indexed scores and weights as in
            \eqref{eq:weighted_kfc_eval} and \eqref{eq:weighted_kfc_pval},
            using level $\alpha$.

            Set $\widetilde\cR^{(i)}\gets
            \cR_\alpha^\ebh(\widetilde e_1^{(i)},\dots,\widetilde e_m^{(i)})$ and
            $\{\widetilde\cA^{(i)}(c)\}_{c\in[0,1]}\gets
            \mathfrak A_\alpha(\widetilde p_1^{(i)},\dots,\widetilde p_m^{(i)})$.

            Set
            \[
            \widetilde f_j^{(i)}\gets
            \frac{m\ind\!\left\{
            \widetilde e_j^{(i)}\ge
            \frac{m}{\alpha|\widetilde\cR^{(i)}\cup\{j\}|}
            \text{ or }j\in\widetilde\cA^{(i)}(q_j)\right\}}
            {\alpha|\widetilde\cR^{(i)}\cup\{j\}|}
            -\widetilde e_j^{(i)}.
            \]

            $\phi_j(q_j;S_j)\gets
            \phi_j(q_j;S_j)+\frac{w_i}{w^{(j)}}\widetilde f_j^{(i)}$.
        }

        $e_j^\boost\gets
        \frac{m}{\alpha|\cR_0\cup\{j\}|}
        \ind\{\phi_j(q_j;S_j)\le0\}$.
    }
}
$\mathcal{R} \gets \cR_\alpha^\ebh(e_1^\boost,\dots,e_m^\boost)$.

\KwOut{Rejection set $\mathcal{R}$}
\end{algorithm}